\documentclass[11pt]{amsart}
\allowdisplaybreaks
\usepackage[foot]{amsaddr}
\usepackage{amssymb,amsmath,amsthm}
\usepackage{graphics}
\usepackage{bbm}
\usepackage{enumerate}    
\usepackage{multirow}
\usepackage{adjustbox}
\usepackage{todonotes}
\usepackage{url}
\usepackage{textcomp}
\usepackage{booktabs}
\usepackage{commath}
\usepackage{thmtools}
\usepackage[hypertexnames=false]{hyperref}
\usepackage[capitalise]{cleveref}
\usepackage{float}
\usepackage{tikz}
\usepackage[center]{caption}
\usepackage{subcaption}
\usepackage{textcmds}
\usepackage{dsfont}
\usepackage{appendix}
\newtheorem{thm}{Theorem}[section]
\newtheorem{cor}[thm]{Corollary}
\newtheorem{lem}[thm]{Lemma}

\theoremstyle{definition}
\newtheorem{defn}[thm]{Definition}
\theoremstyle{remark}
\newtheorem{rem}[thm]{Remark}

\numberwithin{equation}{section}

\newcommand{\sgn}{\mathrm{sgn}}
\newcommand{\RR}{\mathbb{R}}

\newcommand{\QQ}{\mathbb{Q}}

\newcommand{\NN}{\mathbb{N}}
\newcommand{\EE}{\mathbb{E}}

\newcommand{\VV}{\mathbb{V}}

\newcommand{\pp}{\textup{\texttt{+}}}
\newcommand{\mm}{\textup{\texttt{-}}}

\title[Term structure shapes in the Hull--White model]{Term structure shapes in the Hull--White model with Svensson--parameterized initial yield curves}
\author{Felix Sachse}
\address{Department of Mathematics, Saarland University, Campus E2 4, D-66123 Saarbrücken, Germany}
\email{sachse@math.uni-sb.de}
\thanks{This article is based on two chapters of the doctoral thesis of the author \cite{sachse2026thesis}.}

\begin{document}
\begin{abstract}
	
We examine the shapes attainable by the forward and yield curve in the Hull--White model with Svensson--parameterized initial yield curves. For Nelson--Siegel--parameterized initial yield curves, we provide a complete classification of all attainable shapes and partition the parameter space and the state space according to these shapes. Our analysis relies on the theory of Tchebycheff systems and the envelope method. For Bliss-- and Svensson--parameterized initial yield curves, we derive lower and upper bounds on the number of attainable local extrema. We examine the asymptotic behavior of the model and demonstrate that the yield curve attains only the shapes \texttt{normal}, \texttt{inverse} and \texttt{humped} with positive probability as $t\rightarrow\infty.$
	
\end{abstract}
\maketitle

\section{Introduction}

The shape of the term structure of interest rates, commonly expressed in terms of \emph{level}, \emph{slope} and \emph{curvature} is an important economic indicator \cite{litterman1991common, diebold2006macroeconomy}. Yield curves with negative slope, for example, have been identified as a signal of economic recession \cite{bauer2018economic}. Beyond these simple characteristics, however, the term structure may exhibit rather complex shapes with multiple local extrema \cite{gurkaynak2007us, KRS26}. Given a model of the term structure, it is therefore of interest to determine which shapes\footnote{In our terminology, the term structure’s \textbf{shape} is characterized by the number and types of its local extrema; see also \cref{tab:shape}.} can be attained within the model. 

Already when introducing the single--factor Vasicek model \cite{vasicek1977equilibrium}, Oldrich Vasicek showed that the yield curve can only attain three different shapes: \texttt{normal}, \texttt{inverse} and \texttt{humped}. More recently, the shapes of forward and yield curves have been studied more systematically. For one--factor affine short rate models it was shown in \cite{keller2008yield, keller2018correction} that, as in the Vasicek model, forward and yield curve only attain the shapes \texttt{normal}, \texttt{inverse} and \texttt{humped}, with the shape being solely determined by the level of the short rate. For the two--factor Vasicek model, a complete classification of the much richer set of attainable shapes was obtained using the theory of Descartes systems \cite{keller2021classification}. Moreover, the state space was segmented according to the attainable shapes using the concept of the envelope of a family of lines \cite{KRS23}. Tchebycheff systems and envelopes were subsequently employed to obtain a complete classification of attainable shapes and a corresponding segmentation of the parameter space for the widely used Svensson family, also yielding insights into the consistent dynamic Svensson model \cite{KRS26}.

Building on these latest results, we now investigate the attainable shapes in the Hull--White model \cite{HW90}. As an extension of the Vasicek model, the Hull--White model must be able to produce a richer variety of shapes. In fact, as the initial term structure can be chosen arbitrarily, every shape is attainable in the Hull--White model. We must therefore restrict our analysis to a specific family of initial (forward or yield) curves, namely the curves belonging to the Svensson family.

\begin{table}[hbtp]
	\begin{center}
		\begin{tabular}{p{3cm}p{5cm}p{2.5cm}} 
			\toprule
			Shape of the term structure & Description & Sign sequence of derivative\\ 
			\midrule 
			\texttt{flat} & constant & $0$\\
			\texttt{normal} & strictly increasing & $\pp$\\
			\texttt{inverse} & strictly decreasing & $\mm$\\
			\texttt{humped} & single local maximum & $\pp \mm$\\
			\texttt{dipped} & single local minimum & $\mm \pp$\\
			\texttt{hd} & hump-dip, i.e. local maximum followed by local minimum & $\pp \mm \pp$\\
			\texttt{dh}, \texttt{hdh}, etc. & further sequences of multiple `dips'  and `humps'  & $\dotsc$ \\
			\bottomrule
		\end{tabular}
	\end{center}
	\caption{Shapes of the term structure; adapted from \cite{keller2021classification}.\label{tab:shape}}
\end{table}

\subsection{The Hull--White model}

We introduce the Hull--White model, as considered in \cite{Fil09}. The dynamics of the short rate $\mathbf{r}$ under the risk neutral measure $\QQ$ are given by
\begin{align}
	\dif \mathbf{r}(t) = (b(t) - \lambda \mathbf{r}(t))\dif t + \sigma\dif W(t),
\end{align}
where $\lambda$ and $\sigma$ are positive constants and $b$ is chosen such that a given initial forward curve $f_0$ is matched, i.e.
\begin{align*}
	b(t) = \partial_t(f_0(t) + g(t)) + \lambda(f_0(t) + g(t))
\end{align*}
where
\begin{align*}
	g(t) = -\frac{\sigma^2}{2\lambda^2}\left({\mathrm{e}}^{-\lambda t} - 1\right)^2.
\end{align*}

The maps
\begin{align*}
	\RR_{ >0}\ni x\mapsto f(x,t,\mathbf{r}(t))\qquad \text{and} \qquad \RR_{ >0}\ni x\mapsto y(x,t,\mathbf{r}(t)),
\end{align*}
with
\begin{align*}
	\begin{split}f(x,t,r) &= f_0(x+t) + \left(\frac{\sigma^2}{2\lambda^2}\left(1-{\mathrm{e}}^{-2\lambda t}\right) - f_0(t) + r\right)\left[{\mathrm{e}}^{-\lambda x}\right]\\ &\qquad- \frac{\sigma^2}{2\lambda^2}\left(1 - {\mathrm{e}}^{-2\lambda t}\right)\left[{\mathrm{e}}^{-2\lambda x}\right],\end{split}\\
	\begin{split}y(x,t,r) 
		&= \frac{1}{x}\int_0^x f_0(\xi+t)\dif \xi\\
		&\qquad  + \left(\frac{\sigma^2}{2\lambda^2}\left(1-{\mathrm{e}}^{-2\lambda t}\right) - f_0(t) + r\right)\left[\frac{1 - {\mathrm{e}}^{-\lambda x}}{\lambda x}\right]\\ &\qquad- \frac{\sigma^2}{2\lambda^2}\left(1 - {\mathrm{e}}^{-2\lambda t}\right)\left[\frac{1 - {\mathrm{e}}^{-2\lambda x}}{2\lambda x}\right],\end{split}
\end{align*}
are called forward and yield curve, respectively. Both curves can be transformed into each other via the relation
\begin{align*}
	y(x) = \frac{1}{x}\int_0^x f(\xi)\dif \xi \qquad\text{and}\qquad f(x) = y(x) + xy'(x),
\end{align*}
which allows us to speak about an initial curve without having to specify if we mean an initial forward curve or an initial yield curve.

The initial curve will be chosen from the Svensson family \cite{Sve94}, with parameters $(\tau_1, \tau_2, \beta_0, \beta_1, \beta_2, \beta_3)\in\RR_{ >0}^2\times\RR^4$, expressed in terms of the forward curve here,
\begin{align}\label{eq:init_Sve}
	f_0^S(x) = \beta_0 + \beta_1{\mathrm{e}}^{-\frac{x}{\tau_1}} + \frac{\beta_2}{\tau_1}x{\mathrm{e}}^{-\frac{x}{\tau_1}} + \frac{\beta_3}{\tau_2}x{\mathrm{e}}^{-\frac{x}{\tau_2}},
\end{align} 
or from one of its subfamilies: The Bliss family \cite{Bli96}
\begin{align}\label{eq:init_Bli}
	f_0^B(x) = \beta_0 + \beta_1{\mathrm{e}}^{-\frac{x}{\tau_1}} + \frac{\beta_3}{\tau_2}x{\mathrm{e}}^{-\frac{x}{\tau_2}},
\end{align}
the Nelson--Siegel family \cite{NS87}
\begin{align}\label{eq:init_NS}
	f_0^{NS}(x) = \beta_0 + \beta_1{\mathrm{e}}^{-\frac{x}{\tau}} + \frac{\beta_2}{\tau}x{\mathrm{e}}^{-\frac{x}{\tau}},
\end{align}
the Two--factor family \cite{DPR05} with monotone curves
\begin{align}\label{eq:init_mon}
	f_0^M(x) = \beta_0 + \beta_1{\mathrm{e}}^{-\frac{x}{\tau}},
\end{align}
or the One--factor family with flat curves
\begin{align}\label{eq:init_const}
	f_0^F(x) = \beta_0.
\end{align}
When working with the Bliss family or the Svensson family we will always assume $\tau_1\neq\tau_2$. 

We want to answer, at least partially, the following questions:
\begin{description}
	\item[Q1] \label{qu:classification} Which shapes, exactly, can be represented in the Hull--White model with initial curves from the Svensson family (or one of its subfamilies), and which cannot? 
	\item[Q2] \label{qu:separation} What is the shape of the forward curve and the yield curve for a given parameter vector $(\lambda, \sigma, \tau_1, \tau_2, \beta_0, \beta_1, \beta_2, \beta_3)\in\RR_{ >0}^4\times\RR^4$, at a given time $t\in\RR_{ >0}$, conditional on the state $\mathbf{r}(t) = r\in\RR$?
	\item[Q3] \label{qu:asymptotics} Which shapes are observed with strictly positive probability when $t\rightarrow \infty$?
\end{description}
In line with \cite{KRS26}, we will refer to the first question as the \emph{classification problem} and the second question as the \emph{segmentation problem}.

\subsection{Upper bounds}

In order to determine the shape of a curve, we have to calculate the number of transversal zeros of its derivative. For the forward curve this derivative is given by
\begin{align}\label{eq:for_der}
	\begin{split}\partial_xf(x,t,r) &= f_0'(x+t) - \left(\frac{\sigma^2}{2\lambda^2}\left(1-{\mathrm{e}}^{-2\lambda t}\right) - f_0(t) + r\right)\left[\lambda{\mathrm{e}}^{-\lambda x}\right]\\ &\qquad+ \frac{\sigma^2}{2\lambda^2}\left(1 - {\mathrm{e}}^{-2\lambda t}\right)\left[2\lambda{\mathrm{e}}^{-2\lambda x}\right],\end{split}
\end{align}
and for the yield curve we obtain, using the integral representation from Lemma~\ref{lem:integral_representation},
\begin{align}\label{eq:yield_der}
	\begin{split}&\partial_xy(x,t,r) = \frac{1}{x^2}\int_0^x \xi\partial_xf(\xi,t,r) \dif \xi\\
		&\quad= \frac{1}{x^2}\int_0^x \xi f_0'(\xi+t)\dif \xi\\&\qquad + \left(\frac{\sigma^2}{2\lambda^2}\left(1-{\mathrm{e}}^{-2\lambda t}\right) - f_0(t) + r\right)\left[\frac{{\mathrm{e}}^{-\lambda x}(1+\lambda x) - 1}{\lambda x^2}\right]\\ &\qquad - \frac{\sigma^2}{2\lambda^2}\left(1 - {\mathrm{e}}^{-2\lambda t}\right)\left[\frac{{\mathrm{e}}^{-2\lambda x}(1+2\lambda x) - 1}{2\lambda x^2}\right].\end{split}
\end{align}
Since
\begin{align}\label{eq:first_der_Sve}
	\frac{\dif f_0^S}{\dif x}(x) &= \left(\frac{\beta_2 - \beta_1}{\tau_1}\right){\mathrm{e}}^{-\frac{x}{\tau_1}} - \frac{\beta_2}{\tau_1^2}x{\mathrm{e}}^{-\frac{x}{\tau_1}} + \frac{\beta_3}{\tau_2}{\mathrm{e}}^{-\frac{x}{\tau_2}} - \frac{\beta_3}{\tau_2^2}x{\mathrm{e}}^{-\frac{x}{\tau_2}},
\end{align}
upper bounds for the number of local extrema of the forward or yield curve are easily determined.
\begin{thm}\label{thm:bounds}
	Forward and yield curve in the Hull--White model attain at most
	\begin{enumerate}
		\item one local extremum, if the initial curve is chosen from the One--factor family,
		\item two local extrema, if the initial curve is chosen from the Two--factor family,
		\item three local extrema, if the initial curve is chosen from the Nelson--Siegel family,
		\item four local extrema, if the initial curve is chosen the Bliss family,
		\item five local extrema, if the initial curve is chosen from the Svensson family.
	\end{enumerate}
\end{thm}
\begin{proof}
	Choose a curve from the Svensson family as an initial forward curve. The function  $x\mapsto \partial_xf(t,x)$, given in \eqref{eq:for_der}, with $f_0$ from \eqref{eq:init_Sve}, is a polynomial of the system
	\begin{align}\label{eq:ECT_system_f}
		\left({\mathrm{e}}^{-\frac{x}{\tau_1}}\ , \ x{\mathrm{e}}^{-\frac{x}{\tau_1}}\ ,\ {\mathrm{e}}^{-\frac{x}{\tau_2}}\ ,\ x{\mathrm{e}}^{-\frac{x}{\tau_2}}\ ,\ {\mathrm{e}}^{-2\lambda x}\ ,\ {\mathrm{e}}^{-\lambda x}\right).
	\end{align}
	After potentially changing the order of the functions above, one obtains an ECT-system; cf. Corollary~\ref{cor:ECT_forward}. The last statement of the theorem thus follows from Theorem~\ref{thm:T-system_upper_bounds}. When considering initial curves from one of the subfamilies, the bounds can be obtained by choosing an appropriate subsystem. The upper bounds for the yield curve are obtained by applying \cite[Thm. 3.3]{keller2021classification} or Lemma~\ref{lem:ECT_yield}.
\end{proof}

\begin{rem}\label{rem:bounds}
	The bounds in Theorem~\ref{thm:bounds} can be improved if, for example, $\lambda = \frac{1}{\tau_1}$.
\end{rem}

\begin{rem}
	The number of local extrema of the initial curves was already determined in \cite{KRS26}. Comparing these results with the theorem above shows that, for each family, the Hull--White model can generate at most two additional local extrema.
\end{rem}

\section{Tracking functions and envelopes}

\subsection{A family of curves}

To determine the sign sequences of \eqref{eq:for_der} and \eqref{eq:yield_der}, we need to study the zero sets of both functions. Let us define
\begin{align*}
	\Theta := \RR_{ >0}\times\RR.
\end{align*}
We introduce two parameterized families
\begin{align*}
	\mathcal{F}^{\mathrm{f}} := (\zeta^{\mathrm{f}}_x)_{x \in (0,\infty)}\qquad\text{and}\qquad \mathcal{F}^{\mathrm{y}} := (\zeta^{\mathrm{y}}_x)_{x \in (0,\infty)},
\end{align*}
where
\begin{align*}
	\zeta^{\mathrm{f}}_x &:= \{(t,r)\in\Theta\ |\ \partial_xf(x,t,r) = 0\},\\
	\zeta^{\mathrm{y}}_x &:= \{(t,r)\in\Theta\ |\ \partial_xy(x,t,r) = 0\},
\end{align*}
for $x > 0$. Additionally we consider the limiting sets $\zeta_0$ and $\zeta_\infty$, which may separate $\Theta$ into regions $\zeta^{\pm}_0$ and $\zeta^{\pm}_\infty$, such that the forward or yield curve is initially increasing (decreasing) in $\zeta^{+}_0$ ($\zeta^{-}_0$), as well as terminally increasing (decreasing) in $\zeta^{+}_\infty$ ($\zeta^{-}_\infty$). 
As mentioned in \cite{KRS23}, where the setting was quite similar, the complete information regarding all attainable shapes of the forward curve resp. yield curve is encoded in these families.

As a convenient representation of the zero sets, we employ parameterized functions $t\mapsto r_x(t)$, whose graphs are precisely the zero sets, and hence are connected to $(\zeta_x)_{x \in (0,\infty)}$ by the identity
\begin{align*}
	\zeta_x = \{(t,r_x(t)) | t > 0\}.
\end{align*}
They are explicitly given by
\begin{align}\label{eq:env_1}
	\begin{split} r_x^{\mathrm{f}}(t) &= \frac{f_0'(x+t){\mathrm{e}}^{\lambda x}}{\lambda} - \left(\frac{\sigma^2}{2\lambda^2}\left(1 - {\mathrm{e}}^{-2\lambda t}\right) - f_0(t)\right)\\ &\qquad + \frac{\sigma^2}{\lambda^2}\left(1 - {\mathrm{e}}^{-2\lambda t}\right){\mathrm{e}}^{-\lambda x}\end{split}
\end{align}
and
\begin{align}\label{eq:ell_infty_gen}
	\begin{split}r^{\mathrm{y}}_x(t) &= -\frac{\lambda}{{\mathrm{e}}^{-\lambda x}(1+\lambda x) - 1}\int_0^x \xi f_0'(\xi+t)\dif \xi\\ &\qquad- \left(\frac{\sigma^2}{2\lambda^2}\left(1-{\mathrm{e}}^{-2\lambda t}\right) - f_0(t)\right) \\& \qquad + \frac{\sigma^2}{2\lambda^2}\left(1 - {\mathrm{e}}^{-2\lambda t}\right)\left[\frac{{\mathrm{e}}^{-2\lambda x}(1+2\lambda x) - 1}{2({\mathrm{e}}^{-\lambda x}(1+\lambda x) - 1)}\right].\end{split}
\end{align}
By plugging $x = 0$ into \eqref{eq:for_der} we obtain
\begin{align*}
	r_0(t) = f_0(t) + \frac{f_0'(t)}{\lambda} + \frac{\sigma^2}{2\lambda^2}\left(1 - {\mathrm{e}}^{-2\lambda t}\right).
\end{align*}
By Lemma~\ref{lem:integral_representation}, $\partial_xy(0,t,r) = \frac{1}{2}\partial_xf(0,t,r)$. Thus, the graph of $r_0(\cdot)$ partitions $\Theta$ into two regions, corresponding to initially increasing and initially decreasing forward and yield curves.
Assuming now $f_0 = f_0^S$, we can determine the terminal monotonic behavior of the forward curve by comparing the exponents in \eqref{eq:env_1}. If $\frac{1}{\tau_1} > \lambda$ and $\frac{1}{\tau_2} > \lambda$, $\zeta_\infty^{\mathrm{f}} = \{(t,r_\infty^{\mathrm{f}}(t)) | t > 0\}$, where
\begin{align*}
	r_\infty^{\mathrm{f}}(t) = f_0(t) - \frac{\sigma^2}{2\lambda^2}\left(1 - {\mathrm{e}}^{-2\lambda t}\right),
\end{align*}
and if $\frac{1}{\tau_2} > \frac{1}{\tau_1} = \lambda$ and $\beta_2 = 0$, 
\begin{align*}
	r_\infty^{\mathrm{f}}(t) = \beta_0 + \frac{\beta_3}{\tau_2}t{\mathrm{e}}^{-\frac{t}{\tau_2}} - \frac{\sigma^2}{2\lambda^2}\left(1 - {\mathrm{e}}^{-2\lambda t}\right).
\end{align*}
In all remaining cases, the set $\zeta^{\mathrm{f}}_\infty$ is empty and the forward curve exhibits the same terminal monotonic behavior for every choice of $t >  0$ and $r\in\RR$.
The forward curve is terminally increasing ($\zeta_\infty^{{\mathrm{f}}+} = \Theta$ and $\zeta_\infty^{{\mathrm{f}}-} = \emptyset$), if
\begin{align}\label{eq:forward_term_inc}
	\begin{split}
		&\lambda \ge \frac{1}{\tau_1},\ \frac{1}{\tau_2} > \frac{1}{\tau_1},\ \beta_2 < 0,\text{ or}\\
		&\lambda > \frac{1}{\tau_1},\ \frac{1}{\tau_2} > \frac{1}{\tau_1},\ \beta_2 = 0,\ \beta_1 < 0,\text{ or}\\
		&\lambda \ge \frac{1}{\tau_2},\ \frac{1}{\tau_1} > \frac{1}{\tau_2},\ \beta_3 < 0,
	\end{split}
\end{align}
and terminally decreasing ($\zeta_\infty^{{\mathrm{f}}+} = \emptyset$ and $\zeta_\infty^{{\mathrm{f}}-} = \Theta$), if
\begin{align}\label{eq:forward_term_dec}
	\begin{split}
		&\lambda \ge \frac{1}{\tau_1},\ \frac{1}{\tau_2} > \frac{1}{\tau_1},\ \beta_2 > 0,\text{ or}\\
		&\lambda > \frac{1}{\tau_1},\ \frac{1}{\tau_2} > \frac{1}{\tau_1},\ \beta_2 = 0,\ \beta_1 > 0,\text{ or}\\
		&\lambda \ge \frac{1}{\tau_2},\ \frac{1}{\tau_1} > \frac{1}{\tau_2},\ \beta_3 > 0.
	\end{split}
\end{align}
To understand the terminal monotonic behavior of the yield curve, we begin by evaluating the integral in \eqref{eq:ell_infty_gen}. 
\begin{align*}
	\int_0^x \xi f_0'(\xi+t)\dif \xi &= \left[\beta_2(h_2(x,\tau_1) - h_1(x,\tau_1)) + \beta_1h_1(x,\tau_1)\right]\tau_1{\mathrm{e}}^{-\frac{t}{\tau_1}}\\ &\qquad + \beta_2h_1(x,\tau_1)t{\mathrm{e}}^{-\frac{t}{\tau_1}}\\ &\qquad + \beta_3(h_2(x,\tau_2) -  h_1(x,\tau_2))\tau_2{\mathrm{e}}^{-\frac{t}{\tau_2}}\\ &\qquad + \beta_3h_1(x,\tau_2)t{\mathrm{e}}^{-\frac{t}{\tau_2}},
\end{align*}
where
\begin{align*}
	h_1(x,\tau) &:= \left(1 + \frac{x}{\tau}\right){\mathrm{e}}^{-\frac{x}{\tau}} - 1\quad,\ \tau > 0,\\
	h_2(x,\tau) &:= \left(2 + 2\frac{x}{\tau} + \frac{x^2}{\tau^2}\right){\mathrm{e}}^{-\frac{x}{\tau}} - 2\quad,\ \tau > 0.
\end{align*}
Since
\begin{align*}
	\lim_{x\rightarrow\infty} h_1(x,\tau) = -1\qquad , \qquad \lim_{x\rightarrow\infty} h_2(x,\tau) = -2,
\end{align*}
taking the limit $x\rightarrow\infty$ in \eqref{eq:ell_infty_gen} yields
\begin{align*}
	r_\infty^{\mathrm{y}}(t) &= \beta_0 + \beta_1(1-\lambda\tau_1){\mathrm{e}}^{-\frac{t}{\tau_1}} + \beta_2\left(\frac{t}{\tau_1} - \lambda(\tau_1 + t)\right){\mathrm{e}}^{-\frac{t}{\tau_1}}\\ &\qquad + \beta_3{\mathrm{e}}^{\frac{t}{\tau_2} - \lambda(\tau_2 + t)} - \frac{\sigma^2}{4\lambda^2}\left(1 - {\mathrm{e}}^{-2\lambda t}\right).
\end{align*}
Finally, we disclose the limits for $t\downarrow 0$ and $t\uparrow\infty$.
\begin{align*}
	\lim_{t\downarrow0}r_0(t) &= \beta_0 + \beta_1\left(1 - \frac{1}{\lambda\tau_1}\right) + \frac{\beta_2}{\lambda\tau_1} + \frac{\beta_3}{\lambda\tau_2},\\
	\lim_{t\downarrow0}r_\infty^{\mathrm{f}}(t) &= \begin{cases}\beta_0 + \beta_1 &\text{, if } \frac{1}{\tau_1} > \lambda \text{ and } \frac{1}{\tau_2} > \lambda,\\ \beta_0 &\text{, if } \frac{1}{\tau_2} > \frac{1}{\tau_1} = \lambda \text{ and } \beta_2 = 0,\end{cases}\\
	\lim_{t\downarrow0}r_\infty^{\mathrm{y}}(t) &= \beta_0 + \beta_1(1-\lambda\tau_1) - \beta_2\lambda\tau_1 - \beta_3\lambda\tau_2,\\
	\lim_{t\uparrow\infty} r_0(t) &= \beta_0 + \frac{\sigma^2}{2\lambda^2},\\
	\lim_{t\uparrow\infty} r_\infty^{\mathrm{f}}(t) &=\beta_0 - \frac{\sigma^2}{2\lambda^2},\\
	\lim_{t\uparrow\infty} r_\infty^{\mathrm{y}}(t) &=\beta_0 - \frac{\sigma^2}{4\lambda^2}.
\end{align*}
An illustration of the family $\mathcal{F}$ is shown in \cref{fig:family}.
\begin{figure}
	\centering
	\includegraphics[width=\textwidth]{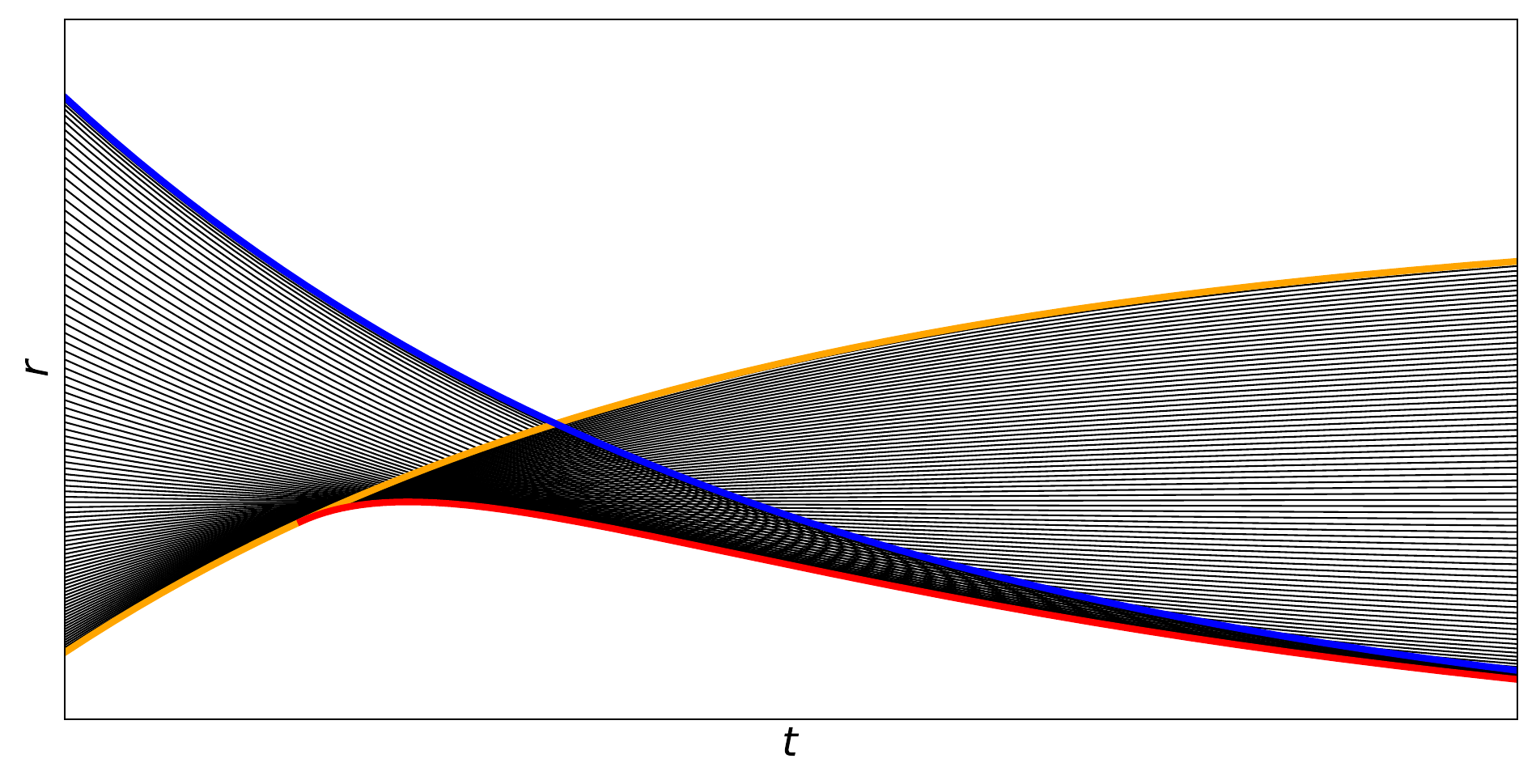}
	\caption{Samples from the family $\mathcal{F}$. Highlighted are $\zeta_0$ (orange), $\zeta_\infty$ (blue) and the envelope (red).}
	\label{fig:family}
\end{figure}

\subsection{The tracking function and the envelope}

Let us now fix $t > 0$ and introduce, as another important tool, the \textbf{tracking function} $x\mapsto r_x(t)$, which tracks the evolution of zero sets $(\zeta_x)_{x \in (0,\infty)}$ along the lines $t = \mathrm{const.}$ in $\Theta$. Determining the number and value of the extremal points $r_{x^*}(t)$ is the key to solving the classification problem \textbf{Q1} and the segmentation problem \textbf{Q2}, since the regions corresponding to different shapes of forward curve or yield curve are exactly separated by the points $(t, r_{x^*}(t))$, along with $(t, r_{0}(t))$ and $(t, r_{\infty}(t))$; as can be concluded from the proofs in \cite{KRS23}, in particular Lem. 4.2. Analogous to \cite[Lem. 3.3]{KRS23}, we make the following observation:
\begin{lem}\label{lem:tracking_envelope}
	For any $x > 0$
	\begin{align*}
		\partial_x r^{\mathrm{f}}_x(t) = 0&\qquad \Longleftrightarrow\qquad \partial_{xx}f(x,t,r^{\mathrm{f}}_x(t)) = 0,\\
		\partial_x r^{\mathrm{y}}_x(t) = 0&\qquad \Longleftrightarrow\qquad \partial_{xx}y(x,t,r^{\mathrm{y}}_x(t)) = 0,
	\end{align*}
	In addition, if $\partial_x r^{\mathrm{f/y}}_x(t) = 0$, then
	\begin{align*}
		\partial_{xx} r^{\mathrm{f}}_x(t) = 0&\qquad \Longleftrightarrow\qquad \partial_{xxx}f(x,t,r^{\mathrm{f}}_x(t)) = 0,\\
		\partial_{xx} r^{\mathrm{y}}_x(t) = 0&\qquad \Longleftrightarrow\qquad \partial_{xxx}y(x,t,r^{\mathrm{y}}_x(t)) = 0.
	\end{align*}
\end{lem}
\begin{proof}
	Since $(x\mapsto\partial_xf(x,t,r^{\mathrm{f}}_x(t)))\equiv 0$, differentiating once with respect to $x$ yields
	\begin{align*}
		\partial_{xx}f(x,t,r_x^{\mathrm{f}}(t)) + \partial_xr^{\mathrm{f}}_x(t)\partial_{xr}f(x,t,r^{\mathrm{f}}_x(t)) = 0.
	\end{align*}
	When $\partial_x r^{\mathrm{f}}_x(t) = 0$, differentiating once again yields
	\begin{align*}
		\partial_{xxx}f(x,t,r_x^{\mathrm{f}}(t)) + \partial_{xx}r^{\mathrm{f}}_x(t)\partial_{xr}f(x,t,r^{\mathrm{f}}_x(t)) = 0.
	\end{align*}
	Since $\partial_{xr}f(x,t,r) = -\lambda{\mathrm{e}}^{-\lambda x} < 0$ we obtain the desired result for the forward curve. By carrying out similar calculations, one can establish the result also for the yield curve.
\end{proof}

We conclude that there is a ono-to-one correspondence between the critical points of the tracking function and the solutions of the system
\begin{align}\label{eq:env_forward}
	\begin{cases}
		\partial_{x}f(x,t,r) = 0,\\
		\partial_{xx}f(x,t,r) = 0,
	\end{cases}
\end{align}
resp.
\begin{align}\label{eq:env_yield}
	\begin{cases}
		\partial_{x}y(x,t,r) = 0,\\
		\partial_{xx}y(x,t,r) = 0.
	\end{cases}
\end{align}
Furthermore, a critical point of the tracking function is extremal, if it does not correspond to a solution of
\begin{align}\label{eq:por_forward}
	\begin{cases}
		\partial_{x}f(x,t,r) = 0,\\
		\partial_{xx}f(x,t,r) = 0,\\
		\partial_{xxx}f(x,t,r) = 0,
	\end{cases}
\end{align}
resp.
\begin{align}\label{eq:por_yield}
	\begin{cases}
		\partial_{x}y(x,t,r) = 0,\\
		\partial_{xx}y(x,t,r) = 0,\\
		\partial_{xxx}y(x,t,r) = 0.
	\end{cases}
\end{align}
In line with \cite[Def. 92, Def. 5.26]{bruce1992curves}, we define the \textbf{envelopes} of $\mathcal{F}^{\mathrm{f}}$ and $\mathcal{F}^{\mathrm{y}}$ to be the sets
\begin{align*}
	\eta^{\mathrm{f}} &:= \{(t,r)\in\Theta\ |\ \exists x > 0 : (x,t,r) \text{ is a solution to } \eqref{eq:env_forward}\},\\
	\eta^{\mathrm{y}} &:= \{(t,r)\in\Theta\ |\ \exists x > 0 : (x,t,r) \text{ is a solution to } \eqref{eq:env_yield}\}
\end{align*}
A point on the envelope which corresponds to a solution of \eqref{eq:por_forward} resp. \eqref{eq:por_yield} is called \textbf{point of regression}. The other points on the envelope are called \textbf{regular}.

In \cite{KRS23} and \cite{KRS26} the envelope appeared as a non-linear contour at the boundary of a region swept by a family of lines. Under certain regularity assumptions, the envelope could be described as a  parameterized continuous curve and the winding number of an augmented envelope could be used to determine the shape of the forward curve or yield curve. Here, however, the families $\mathcal{F}^{\mathrm{f}}$ and $\mathcal{F}^{\mathrm{y}}$ are not any longer families of lines. Hence, for a fixed value of $x > 0$, there might be no or multiple solutions to \eqref{eq:env_forward} resp. \eqref{eq:env_yield}. As a result, rather than representing the envelope as a parameterized curve, we will need to describe it as the graph of a (multi--valued) function of $t$ and cannot rely on a formula to calculate the number of transversal zeros.

\subsection{Equivalent formulations and the number of solutions}

We proceed by giving a brief overview of some essential properties of the envelopes. Starting with the envelope of $\mathcal{F}^{\mathrm{f}}$, we note that a solution $(x,t,r)$ of \eqref{eq:env_forward} also solves
\begin{align}\label{eq:env_2}
	f_0''(x+t) + \lambda f_0'(x+t) = \sigma^2\left(1 - {\mathrm{e}}^{-2\lambda t}\right){\mathrm{e}}^{-2\lambda x}.
\end{align}
When picking the initial curve from the Svensson family, both sides of the equation are polynomials of the same ECT-system \eqref{eq:exppolsys}, with $x$ being the argument. As a result, there are only finitely many solutions for a given $t > 0$. More specifically, there are at most four solutions in the case of the Svensson family, one less for the Bliss family, and, ultimately, no solutions for flat initial curves. A similar behavior can be observed with the number of points of regression:
\begin{lem}\label{lem:por_finite_forward}
	There are finitely many points of regression on $\eta^{\mathrm{f}}$ when the initial curve is chosen from the Svensson family.
\end{lem}
\begin{proof}
	For a solution $(x,t,r)$ of \eqref{eq:por_forward} it must hold
	\begin{align}\label{eq:env_3}
		\begin{split}
			0 &= \partial_{xxx}f(x,t,r) + 3\lambda\partial_{xx}f(x,t,r) + 2\lambda^2\partial_{x}f(x,t,r)\\
			&= f_0'''(x + t) + 3\lambda f_0''(x + t) + 2\lambda^2f_0'(x + t).
		\end{split}
	\end{align}
	Since the right hand side is a polynomial of the ECT-system \eqref{eq:exppolsys}, with the argument $x + t$, the equation is solved only for finitely many $x + t$. Let us fix such a solution $x + t =: c$. Equation \eqref{eq:env_2} now reads
	\begin{align*}
		C = \sigma^2\left(1 - {\mathrm{e}}^{-2\lambda (c - x)}\right){\mathrm{e}}^{-2\lambda x},
	\end{align*}
	for some $C\in\RR$. Rearranging yields
	\begin{align*}
		\frac{C}{\sigma^2} + {\mathrm{e}}^{-2\lambda c} = {\mathrm{e}}^{-2\lambda x}.
	\end{align*}
	This equation has at most one solution in $x$. Thus, there is at most one solution $(x,t,r)$ of \eqref{eq:por_forward} for every solution $x + t$ of \eqref{eq:env_3}.
\end{proof}
There are at most three points of regression in the case of Svensson--parameterized initial curves, one less for the Bliss family, and, ultimately, none for initial curves from the families of monotone curves \eqref{eq:init_mon} and flat curves. \eqref{eq:init_const}.

Continuing with the envelope $\eta^{\mathrm{y}}$ of $\mathcal{F}^{\mathrm{y}}$, we note that
\begin{align*}
	\partial_{xx}y(x,t,r) &= \frac{1}{x}\partial_{x}f(x,t,r) - \frac{2}{x}\partial_{x}y(x,t,r),\\
	\partial_{xxx}y(x,t,r) &= -\frac{3}{x^2}\partial_{x}f(x,t,r) + \frac{1}{x}\partial_{xx}f(x,t,r) + \frac{6}{x^2}\partial_{x}y(x,t,r),
\end{align*}
so, instead of solving \eqref{eq:env_yield} or \eqref{eq:por_yield}, one could equivalently solve
\begin{align}\label{eq:env_yield_alt}
	\begin{cases}
		\partial_{x}y(x,t,r) = 0,\\
		\partial_{x}f(x,t,r) = 0,
	\end{cases}
\end{align}
or
\begin{align}\label{eq:por_yield_alt}
	\begin{cases}
		\partial_{x}y(x,t,r) = 0,\\
		\partial_{x}f(x,t,r) = 0,\\
		\partial_{xx}f(x,t,r) = 0.
	\end{cases}
\end{align}
We wish to particularly emphasize the following fact: points of regression on $\eta^{\mathrm{y}}$ also lie on $\eta^{\mathrm{f}}$ (but are generally not points of regression thereof).

We now establish upper bounds to the number of solutions of \eqref{eq:env_yield}, when $t$ is fixed. By Lemma~\ref{lem:tracking_envelope}, we might therefore determine upper bounds to the number of zeros of $x\mapsto\partial_x r_x^{\mathrm{y}}(t)$. 
When the initial curve is Svensson--parameterized, the derivative of the forward curve \eqref{eq:for_der} is a polynomial of
\begin{align*}
	\left({\mathrm{e}}^{-\lambda x}, {\mathrm{e}}^{-2\lambda x}, {\mathrm{e}}^{-\frac{x}{\tau_1}}, x{\mathrm{e}}^{-\frac{x}{\tau_1}}, {\mathrm{e}}^{-\frac{x}{\tau_2}}, x{\mathrm{e}}^{-\frac{x}{\tau_2}}\right).
\end{align*}
Assuming that $\lambda, 1/\tau_1$ and $1/\tau_2$ are pairwise distinct, we obtain an ECT-system, after possibly altering the signs of the last four functions in the system; see the proof of Corollary~\ref{cor:ECT_forward}. By Lemma~\ref{lem:ECT_yield}, the corresponding system for the derivative of the yield curve also constitutes an ECT-system. We highlight the fact that the function, which is getting multiplied with $r$ in \eqref{eq:yield_der}, is the first one in the system. Applying Lemma~\ref{lem:ECT_der}, we conclude that $x\mapsto\partial_x r_x^{\mathrm{y}}(t)$ is a polynomial of an ECT-system consisting of five functions. Consequently, for every $t > 0$, the derivative of the tracking function has at most four zeros, counting multiplicities; see Theorem~\ref{thm:T-system_upper_bounds}. Improved bounds can be obtained when the initial curve is chosen from a  subfamily or with matching exponents. In particular, when choosing from \eqref{eq:init_mon}, there is at most one solution to \eqref{eq:env_yield} for every $t > 0$ and the envelope exhibits no points of regression. For the Nelson--Siegel family, there are at most two solutions for every $t > 0$, and points of regression do not share their corresponding value of $t$ with other points on the envelope.

\subsection{Initial and terminal monotonic behavior of the tracking function}

Whenever $x\mapsto(t,r_x(t))$ hits a regular point on the envelope of $\mathcal{F}^{\mathrm{f}}$ resp. $\mathcal{F}^{\mathrm{y}}$, the tracking function $x\mapsto r_x(t)$ reverses its monotonicity, switching from decreasing to increasing or from increasing to decreasing. Thus, knowing how often the tracking function changes direction is key to determining the exact number of regular solutions of \eqref{eq:env_forward} resp. \eqref{eq:env_yield}. When this number can be bounded by one, we only have to check if initial and terminal direction of the tracking function match (no regular solutions) or do not match (exactly one regular solution). Since
\begin{align*}
	\frac{\partial}{\partial x} r_x^{\mathrm{f}}(t) = \frac{f_0''(x + t){\mathrm{e}}^{\lambda x}}{\lambda} + f_0'(x + t){\mathrm{e}}^{\lambda x} - \frac{\sigma^2}{\lambda}\left(1 - {\mathrm{e}}^{-2\lambda t}\right){\mathrm{e}}^{-\lambda x},
\end{align*}
the initial direction of the tracking function in the forward case is determined by the sign of 
\begin{align*}
	\frac{\partial}{\partial x} r_x^{\mathrm{f}}(t)|_{x = 0} = \frac{f_0''(t)}{\lambda} + f_0'(t) - \frac{\sigma^2}{\lambda}\left(1 - {\mathrm{e}}^{-2\lambda t}\right),
\end{align*}
where a positive (negative) sign indicates that the tracking function is initially increasing (decreasing). Notably,
\begin{align*}
	\frac{\partial}{\partial x} r_x^{\mathrm{f}}(t)|_{x = 0} \overset{t\rightarrow\infty}{\longrightarrow} -\frac{\sigma^2}{\lambda},
\end{align*}
so, asymptotically, the tracking function is initially decreasing. The derivatives of the forward curve $x\mapsto \partial_x f(x,t,r)$ and the yield curve $x\mapsto \partial_x y(x,t,r)$ have the same initial sign and 
\begin{align*}
	\partial_{xx}y(0^+,t,r) = \frac{\partial_{xx}f(0,t,r)}{3} \qquad\text{and}\qquad \partial_{xr}y(0^+,t,r) = \frac{\partial_{xr}f(0,t,r)}{2};
\end{align*}
see Lemma~\ref{lem:integral_representation} for the first identity. Since
\begin{align*}
	\partial_x r_x^{\mathrm{f}}(t) &= -\frac{\partial_{xx} f(x,t,r_x^{\mathrm{f}}(t))}{\partial_{xr} f(x,t,r_x^{\mathrm{f}}(t))},\\ 
	\partial_x r_x^{\mathrm{y}}(t), &= -\frac{\partial_{xx} y(x,t,r_x^{\mathrm{y}}(t))}{\partial_{xr} y(x,t,r_x^{\mathrm{y}}(t))},
\end{align*}
it holds
\begin{align*}
	\partial_x r_x^{\mathrm{y}}(t)|_{x = 0^+} = \frac{2}{3}\partial_x r_x^{\mathrm{f}}(t)|_{x = 0}.
\end{align*}
So, for every $t > 0$, both tracking functions have the same initial direction. 
To determine the terminal direction of the tracking function in the yield case, we first compute its derivative.
\begin{align}\label{eq:der_tracking_yield}
	\begin{split}
		&\frac{\partial}{\partial x} r_x^{\mathrm{y}}(t) =\\&\quad \frac{\sigma^2}{4\lambda^2}\left(1-{\mathrm{e}}^{-2\lambda t}\right)g_1\left(x,\frac{1}{2\lambda}\right) + \beta_1\left(-\lambda\tau_1{\mathrm{e}}^{-\frac{t}{\tau_1}}g_1(x,\tau_1)\right)\\&\quad + \beta_2\left(-\lambda\tau_1{\mathrm{e}}^{-\frac{t}{\tau_1}}\left(g_2(x,\tau_1) - g_1(x,\tau_1)\right) - \lambda t{\mathrm{e}}^{-\frac{t}{\tau_1}}g_1(x,\tau_1)\right)\\&\quad + \beta_3\left(-\lambda\tau_2{\mathrm{e}}^{-\frac{t}{\tau_2}}\left(g_2(x,\tau_2) - g_1(x,\tau_2)\right) - \lambda t{\mathrm{e}}^{-\frac{t}{\tau_2}}g_1(x,\tau_2)\right),
	\end{split}
\end{align}
where
\begin{align*}
	&g_1(x,\tau) := \frac{x\left(\frac{1}{\tau^2}{\mathrm{e}}^{-\frac{x}{\tau}} + (\lambda - \frac{1}{\tau})(\lambda(1 + \frac{x}{\tau}) + \frac{1}{\tau}){\mathrm{e}}^{-\lambda x - \frac{x}{\tau}} - \lambda^2{\mathrm{e}}^{-\lambda x}\right)}{\left({\mathrm{e}}^{-\lambda x}(1 + \lambda x) - 1\right)^2}\\
	&g_2(x,\tau) :=\\ &\qquad\frac{x\left(\frac{x}{\tau^3}{\mathrm{e}}^{-\frac{x}{\tau}} + \left[\lambda\frac{x^2}{\tau^2}(\lambda - \frac{1}{\tau}) + \frac{x}{\tau}(2\lambda^2 - \frac{1}{\tau^2}) + 2\lambda^2\right]{\mathrm{e}}^{-\lambda x - \frac{x}{\tau}} - 2\lambda^2{\mathrm{e}}^{-\lambda x}\right)}{\left({\mathrm{e}}^{-\lambda x}(1 + \lambda x) - 1\right)^2}.
\end{align*}
A comparison of the exponential-monomial terms in \eqref{eq:der_tracking_yield} yields the following result: The terminal direction of the tracking function is determined by the sign of
\begin{align}\label{eq:term_dir_track_y}
	\begin{cases}
		\begin{aligned}
			&-\frac{\sigma^2}{4\lambda^2}\left(1 - {\mathrm{e}}^{-2\lambda t}\right) + \beta_1\lambda\tau_1{\mathrm{e}}^{-\frac{t}{\tau_1}}\\&\qquad + \beta_2\lambda(\tau_1 + t){\mathrm{e}}^{-\frac{t}{\tau_1}}\\&\qquad + \beta_3\lambda(\tau_2 + t){\mathrm{e}}^{-\frac{t}{\tau_2}},
		\end{aligned} & \text{if } \frac{1}{\tau_1} > \lambda,\frac{1}{\tau_2} > \lambda,\\
		-\frac{\sigma^2}{4\lambda^2}\left(1 - {\mathrm{e}}^{-2\lambda t}\right) + \beta_3\lambda(\tau_2 + t){\mathrm{e}}^{-\frac{t}{\tau_2}}, & \text{if } \frac{1}{\tau_2} > \frac{1}{\tau_1} = \lambda,\beta_2 = 0,\\
		-\beta_3, & \text{if } \lambda \ge \frac{1}{\tau_2}, \frac{1}{\tau_1} > \frac{1}{\tau_2},\\
		-\beta_2, & \text{if } \lambda \ge \frac{1}{\tau_1}, \frac{1}{\tau_2} > \frac{1}{\tau_1}, \beta_2 \neq 0,\\
		-\beta_1, & \text{if } \lambda \ge \frac{1}{\tau_1}, \frac{1}{\tau_2} > \frac{1}{\tau_1}, \beta_2 = 0,
	\end{cases}
\end{align}
where a positive (negative) sign indicates that the tracking function is terminally increasing (decreasing). It is terminally decreasing if $\beta_1 = \beta_2 = \beta_3 = 0$.

\section{Flat initial curves}

We begin our analysis by examining the shapes of the term structure in the case of a flat initial curve.\footnote{This case has already been studied in \cite[Proposition 1]{schlogl1997factor}.} By Theorem~\ref{thm:bounds}, both the forward and yield curve admit at most one local extremum; consequently, their shape is fully characterized by the initial and terminal signs of their respective derivatives, which are determined from
\begin{align}\label{l0_linf_const}
	\begin{split}r_0(t) &= \beta_0 + \frac{\sigma^2}{2\lambda^2}\left(1 - {\mathrm{e}}^{-2\lambda t}\right),\\
		r_\infty^{\mathrm{f}}(t) &= \beta_0 - \frac{\sigma^2}{2\lambda^2}\left(1 - {\mathrm{e}}^{-2\lambda t}\right),\\
		r_\infty^{\mathrm{y}}(t) &= \beta_0 - \frac{\sigma^2}{4\lambda^2}\left(1 - {\mathrm{e}}^{-2\lambda t}\right).\end{split}
\end{align}
$\zeta_0$ and $\zeta_\infty^{\mathrm{f/y}}$ clearly do not intersect. Hence, the only attainable shapes for both the forward curve and the yield curve are \texttt{normal}, \texttt{inverse} and \texttt{humped}. This answers \textbf{Q1} and \textbf{Q2}. The decomposition of $\Theta$ is illustrated in \cref{fig:const}.
\begin{figure}
	\centering
	\begin{subfigure}{0.49\textwidth}
		\includegraphics[width=\textwidth]{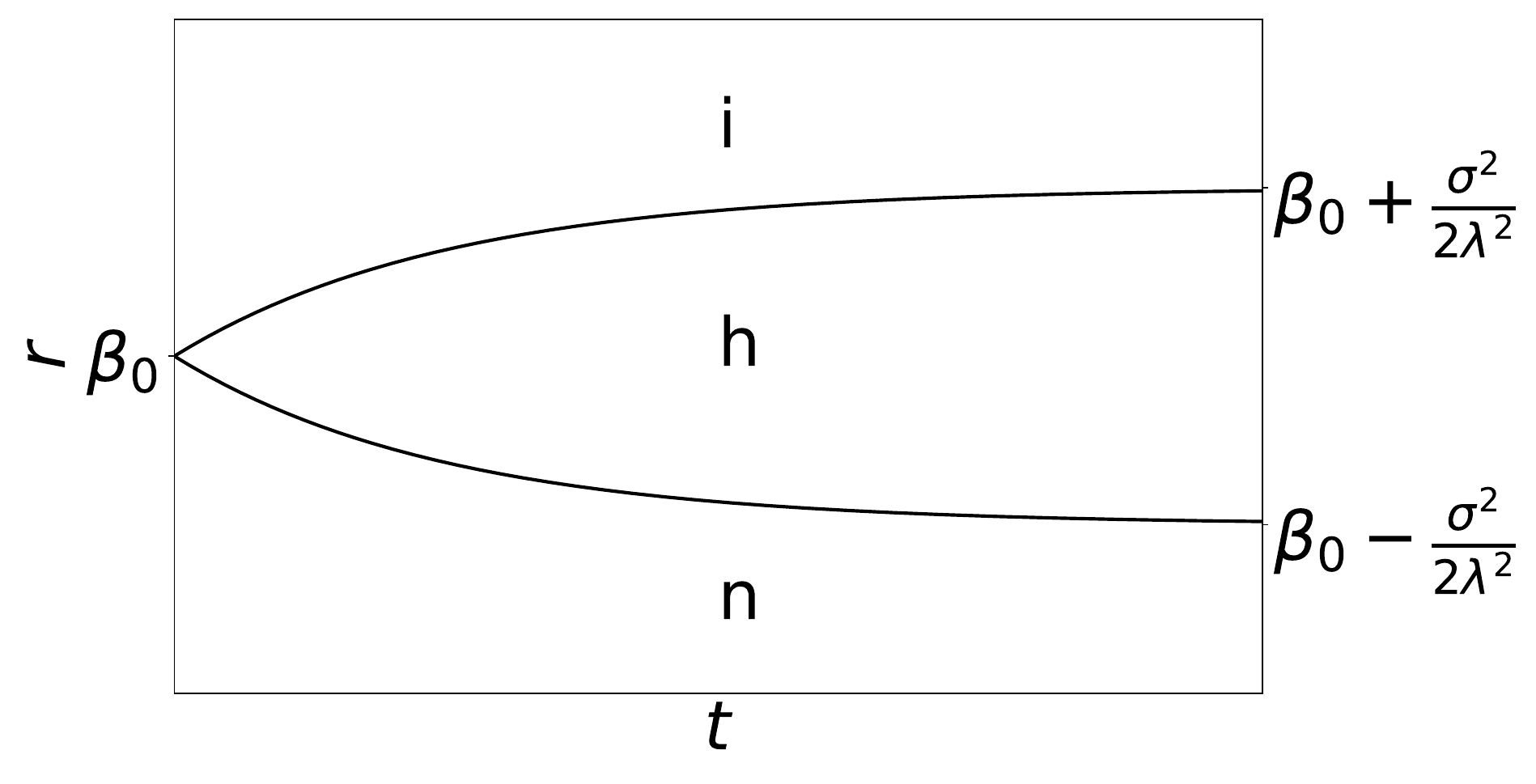}
		\caption{Forward curve.}
	\end{subfigure}
	\hfill
	\begin{subfigure}{0.49\textwidth}
		\includegraphics[width=\textwidth]{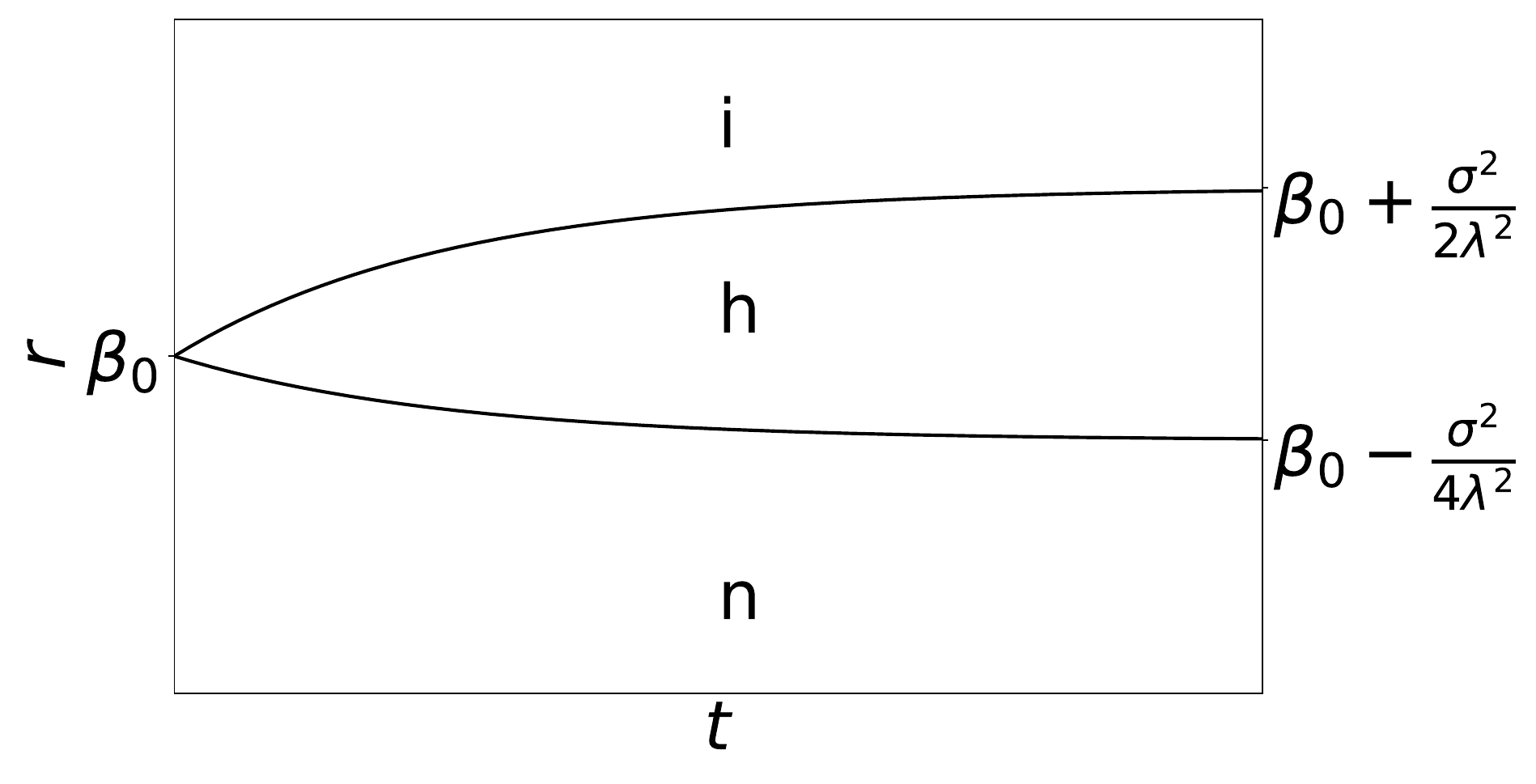}
		\caption{Yield curve.}
	\end{subfigure}
	\caption{Shapes of the term structure for flat initial curves.}
	\label{fig:const}
\end{figure}

\section{Monotone initial curves}

We now assume that the initial curve is monotone \eqref{eq:init_mon}. We will answer the questions \textbf{Q1} and \textbf{Q2}, and start by studying the number of intersections of $\zeta_0$ and $\zeta_\infty^{\mathrm{f/y}}$.
Since
\begin{align*}
	r_0(t) - r_\infty^{\mathrm{f}}(t) &= \begin{cases}
		-\frac{\beta_1}{\tau\lambda}{\mathrm{e}}^{-\frac{t}{\tau}} + \frac{\sigma^2}{\lambda^2}\left(1-{\mathrm{e}}^{-2\lambda t}\right),& \lambda < \frac{1}{\tau},\\
		\frac{\sigma^2}{\lambda^2}\left(1-{\mathrm{e}}^{-2\lambda t}\right),& \lambda = \frac{1}{\tau},
	\end{cases}\\
	r_0(t) - r_\infty^{\mathrm{y}}(t) &= \beta_1\left(\lambda\tau - \frac{1}{\lambda\tau}\right){\mathrm{e}}^{-\frac{t}{\tau}} + \frac{3\sigma^2}{4\lambda^2}\left(1-{\mathrm{e}}^{-2\lambda t}\right)
\end{align*}
are polynomials of an ECT-system, $({\mathrm{e}}^{-\frac{t}{\tau}}, {\mathrm{e}}^{-2\lambda t}, 1)$ or $({\mathrm{e}}^{-2\lambda t}, {\mathrm{e}}^{-\frac{t}{\tau}}, 1)$, there are at most two intersections, respectively. With only positive factors in the polynomial above, there cannot be any intersections and if $\beta_1 > 0$ resp. $\beta_1\left(\lambda\tau - \frac{1}{\lambda\tau}\right) < 0$, there must be exactly one.

\subsection{Segmentation for $\lambda = \frac{1}{\tau}$ or $2\lambda = \frac{1}{\tau}$}

In these cases, forward and yield curve admit at most one local extremum; see Remark~\ref{rem:bounds}. Hence, their shape is solely determined by the initial and terminal signs of their respective derivatives. Choosing $\lambda = \frac{1}{\tau}$ yields exactly the same functions $r_0$ and $r_\infty^{\mathrm{f/y}}$ as in the case of flat initial curves. Therefore, the attainable shapes are \texttt{normal, inverse} and \texttt{humped}. The resulting decomposition of $\Theta$ is illustrated in \cref{fig:const}.

With $2\lambda = \frac{1}{\tau}$, $\zeta_0$ and $\zeta_\infty^{\mathrm{f/y}}$ intersect exactly once, respectively,  if and only if $\beta_1 > 0$. Hence, for $\beta_1 < 0$ the attainable shapes are \texttt{normal, inverse} and \texttt{humped}. If $\beta_1 > 0$, the shapes \texttt{normal, inverse, humped} and \texttt{dipped} are attainable. Also, there is a \texttt{flat} curve, where $\zeta_0$ and $\zeta_\infty^{\mathrm{f/y}}$ intersect. The decomposition of $\Theta$ is illustrated in \cref{fig:mon_dep_1_2}.
\begin{figure}
	\centering
	\begin{subfigure}{0.49\textwidth}
		\includegraphics[width=\textwidth]{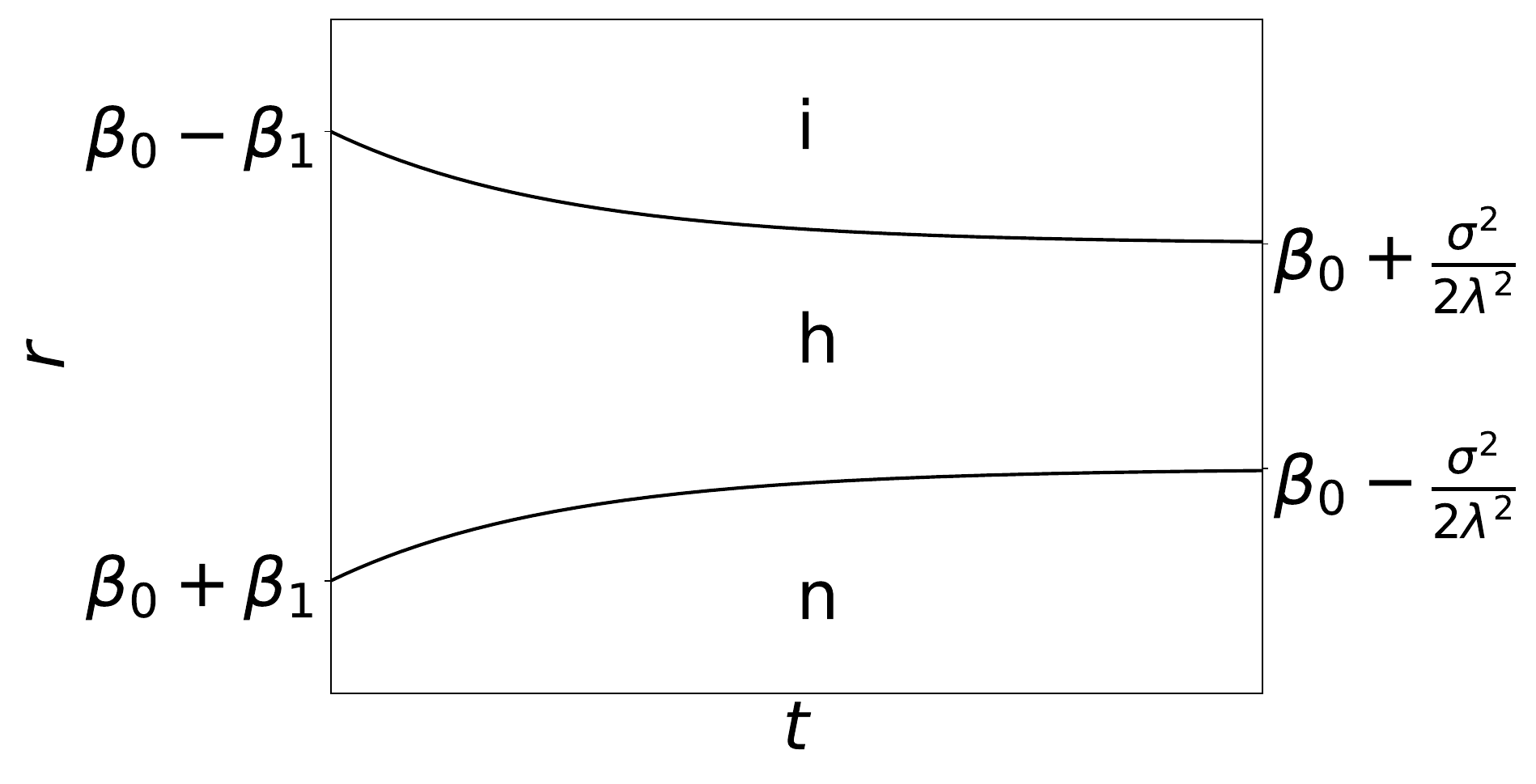}
		\caption{Forward curve with $\beta_1 < 0$.}
	\end{subfigure}
	\hfill
	\begin{subfigure}{0.49\textwidth}
		\includegraphics[width=\textwidth]{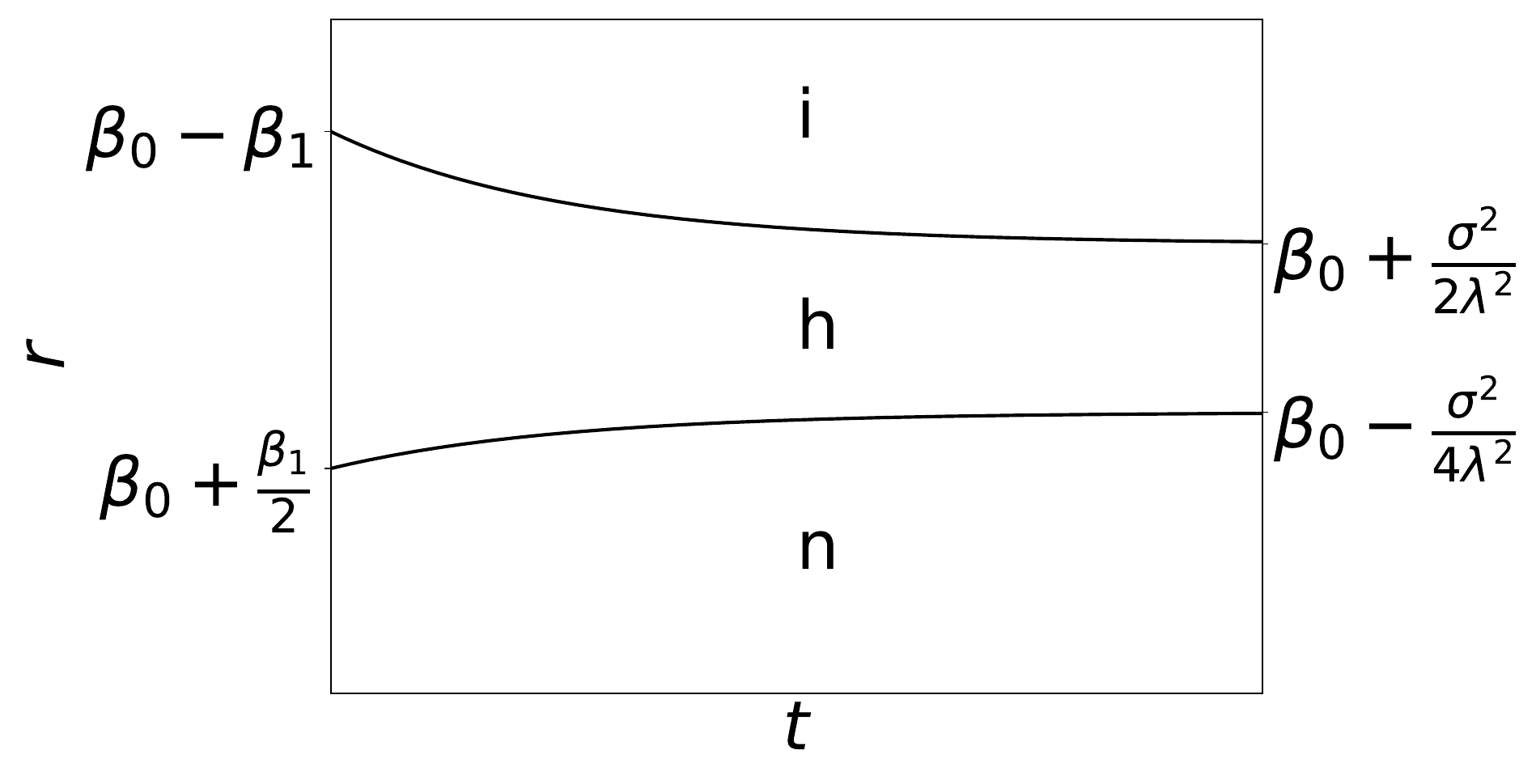}
		\caption{Yield curve with $\beta_1 < 0$.}
	\end{subfigure}
	
	\begin{subfigure}{0.49\textwidth}
		\includegraphics[width=\textwidth]{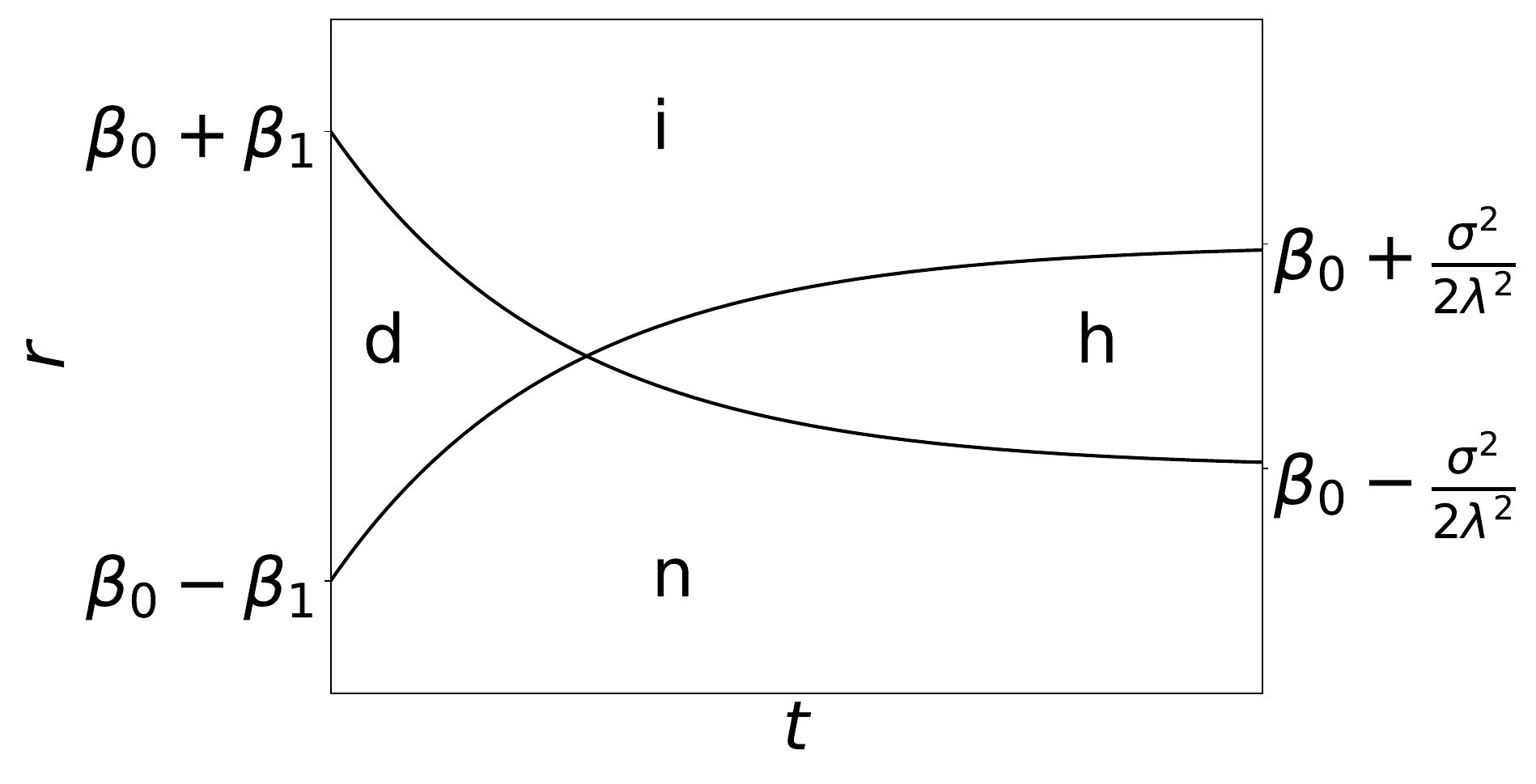}
		\caption{Forward curve with $\beta_1 > 0$.}
	\end{subfigure}
	\hfill
	\begin{subfigure}{0.49\textwidth}
		\includegraphics[width=\textwidth]{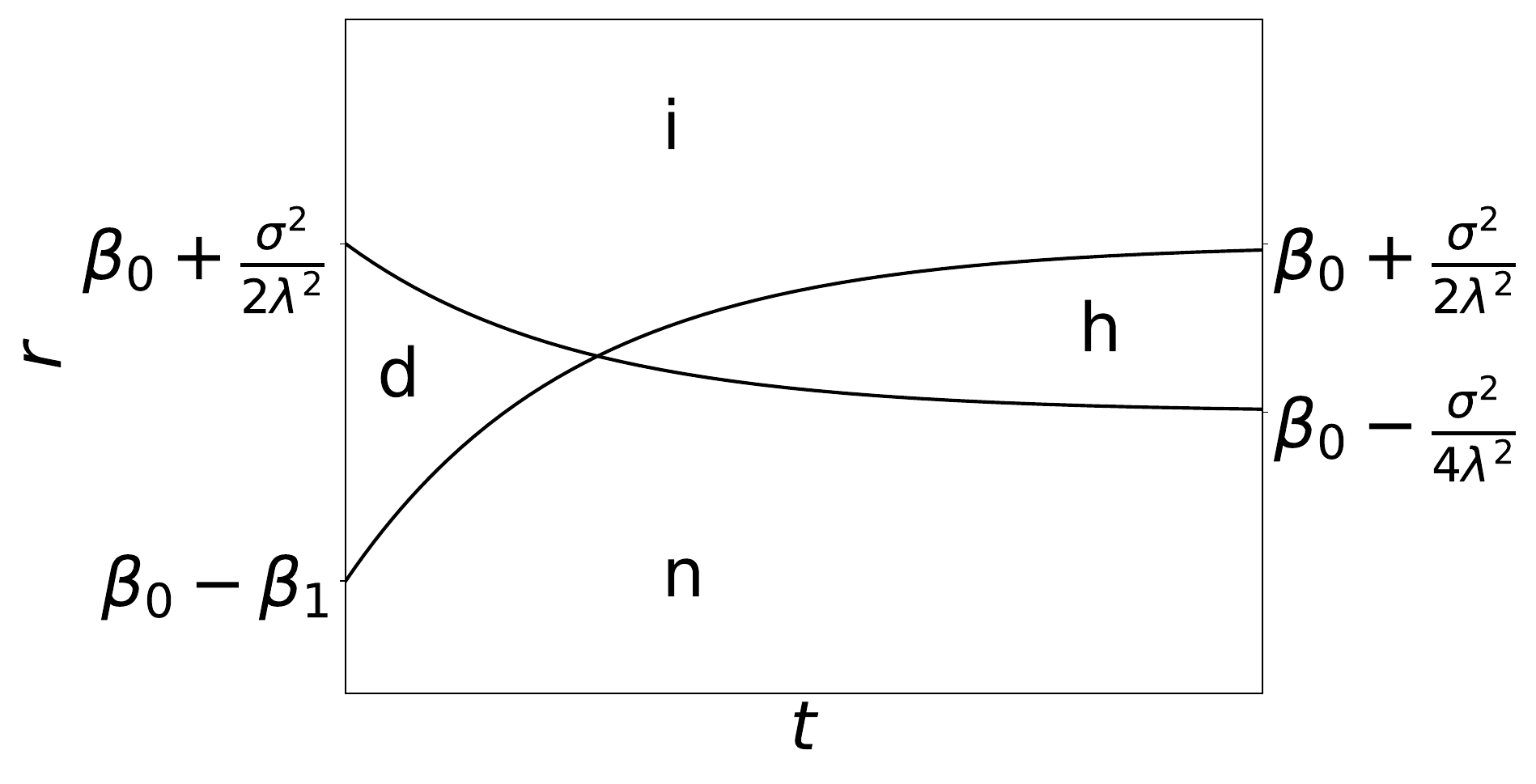}
		\caption{Yield curve with $\beta_1 > 0$.}
	\end{subfigure}
	\caption{Shapes of the term structure for monotone initial curves with $2\lambda = \frac{1}{\tau}$.}
	\label{fig:mon_dep_1_2}
\end{figure}

\subsection{Segmentation for $2\lambda < \frac{1}{\tau}$}	

When we are not dealing with one of the special cases, we need to additionally consider the envelope. By plugging $f_0^M$ into \eqref{eq:env_2} we obtain
\begin{align}\label{eq:env_monotone_forward}
	x = \underbrace{\frac{1}{2\lambda - \frac{1}{\tau}}}_{=: \alpha}\log\Bigg(\underbrace{\frac{\sigma^2\tau\left(1-{\mathrm{e}}^{-2\lambda t}\right)\exp\left(\frac{t}{\tau}\right)}{\beta_1\left(\frac{1}{\tau} - \lambda\right)}}_{=:\gamma(t)}\Bigg)
\end{align}
The numerator of $\gamma$ is positive for $t > 0$, hence, the existence of solutions to \eqref{eq:env_forward} for some $t > 0$ depends on the signs of $\beta_1\left(\frac{1}{\tau} - \lambda\right)$ and $\frac{1}{\tau} - 2\lambda$. With $2\lambda < \frac{1}{\tau}$, $\alpha$ is negative. Thus, if and only if $\gamma$ takes a value in $(0,1)$, a solution to \eqref{eq:env_forward} exists. The sets containing the values of $t$, for which there is a corresponding point on the envelope will be denoted with
\begin{align*}
	T^{\mathrm{f}} &:= \{t > 0\ |\ \exists x > 0,r\in\RR : (x,t,r)\text{ solves }\eqref{eq:env_forward}\}\\
	& = \{t > 0\ |\ \exists x > 0 : (x,t)\text{ solves }\eqref{eq:env_2}\}\\
	T^{\mathrm{y}} &:= \{t > 0\ |\ \exists x > 0,r\in\RR : (x,t,r)\text{ solves }\eqref{eq:env_yield}\}.
\end{align*}

For $\beta_1 < 0$, $\gamma$ is negative. Hence, there are no solutions to \eqref{eq:env_forward}. The shape of the forward curve is fully characterized by the initial and terminal signs of its derivative. Since $\zeta_0$ and $\zeta_\infty^{\mathrm{f}}$ do not intersect, the attainable shapes are \texttt{normal, inverse} and \texttt{humped}. By \cite[Thm.~3.3]{keller2021classification}, the yield curve cannot attain any more shapes. Since also $\zeta_0$ and $\zeta_\infty^{\mathrm{y}}$ do not intersect, the set of attainable shapes is identical.

For $\beta_1 > 0$, $\gamma$ is positive and since its numerator is strictly increasing from 0 up to $\infty$, solutions to \eqref{eq:env_forward} exist for every $t$ in some bounded interval $T^{\mathrm{f}}$ starting at 0. The envelope $\eta^{\mathrm{f}}$ approaches $\zeta_\infty^{\mathrm{f}}$ on its left end and $\zeta_0$ on its right end, in the sense that $r\rightarrow r_\infty^{\mathrm{f}}(0^+)$ as $t\downarrow0$ (and $x\uparrow\infty$) and $r \rightarrow r_0^{\mathrm{f}}(\sup T^{\mathrm{f}})$ as $t\uparrow\sup T^{\mathrm{f}}$ (and $x\downarrow 0$) for the solutions $(x,t,r)$ of \eqref{eq:env_forward}.
There is exactly one intersection point $(t^*,r^*)$ of $\zeta_0$ and $\zeta_\infty^{\mathrm{f}}$, and it must hold
\begin{align*}
	0 = r_0(t^*) - r_\infty^{\mathrm{f}}(t^*) = \frac{{f_0^M}'(t^*)}{\lambda} + \frac{\sigma^2}{\lambda^2}\left(1 - {\mathrm{e}}^{-2\lambda t^*}\right),
\end{align*}
and therefore
\begin{align*}
	\frac{\sigma^2\tau\left(1-{\mathrm{e}}^{-2\lambda t^*}\right)\exp\left(\frac{t^*}{\tau}\right)}{\beta_1\left(\frac{1}{\tau} - \lambda\right)} = \frac{\lambda}{\frac{1}{\tau} - \lambda} \in (0,1).
\end{align*}
Hence, $t^*$ is an element of $T^{\mathrm{f}}$. By Lemma~\ref{lem:por_finite_forward} and the following remark, there are no points of regression on the envelope and thus the tracking function $x\mapsto r_x^{\mathrm{f}}(t)$ changes its direction exactly once, whenever $t\in T^{\mathrm{f}}$. As a result, the envelope does not intersect with $\zeta_0$ or $\zeta_\infty^{\mathrm{f}}$. Since the envelope is the graph of a continuous function and $r_0(0^+) < r_\infty^{\mathrm{f}}(0^+)$, it must also hold $r_0(t) < r$ and $r_\infty^{\mathrm{f}}(t) < r$ for every $(t,r)$ on the envelope. We conclude that the shapes \texttt{normal, inverse, humped, dipped, dh} are attainable.
The yield curve cannot attain any more shapes \cite[Thm.~3.3]{keller2021classification} and, again, the envelope possesses no points of regression and \eqref{eq:env_yield} has at most one solution for every $t > 0$. The range of $t$ for these solutions can be determined by studying the initial and terminal direction of the tracking function. Whenever they do not match, the tracking function changes direction in a corresponding point on the envelope. The initial direction is determined by the sign of
\begin{align}\label{eq:init_direction_mon}
	\underbrace{\frac{\beta_1}{\tau}\left(\frac{1}{\lambda\tau} - 1\right)}_{> 0}\exp\left(-\frac{t}{\tau}\right) - \frac{\sigma^2}{\lambda}\left(1 - {\mathrm{e}}^{-2\lambda t}\right)
\end{align}
and the terminal direction is determined by the sign of
\begin{align}\label{eq:term_direction_mon}
	\beta_1\lambda\tau\exp\left(-\frac{t}{\tau}\right) - \frac{\sigma^2}{4\lambda^2}\left(1 - {\mathrm{e}}^{-2\lambda t}\right).
\end{align}
Both expressions change their sign exactly once (from $\pp$ to $\mm$). Since $2\lambda < 1/\tau$ implies
\begin{align*}
	4\lambda^2\tau < -\frac{4}{3}\lambda\left(\lambda\tau - \frac{1}{\lambda\tau}\right) < \frac{1}{\tau}\left(\frac{1}{\lambda\tau} - 1\right),
\end{align*}
the terminal direction changes first, then $\zeta_0$ and $\zeta_\infty^{\mathrm{y}}$ intersect and lastly, the initial direction of the tracking function changes. Consequently, the yield curve also attains the shapes \texttt{normal, inverse, humped, dipped, dh}. An illustration is provided in \cref{fig:mon_indep_1_2}.
\begin{figure}
	\centering
	\begin{subfigure}{0.49\textwidth}
		\includegraphics[width=\textwidth]{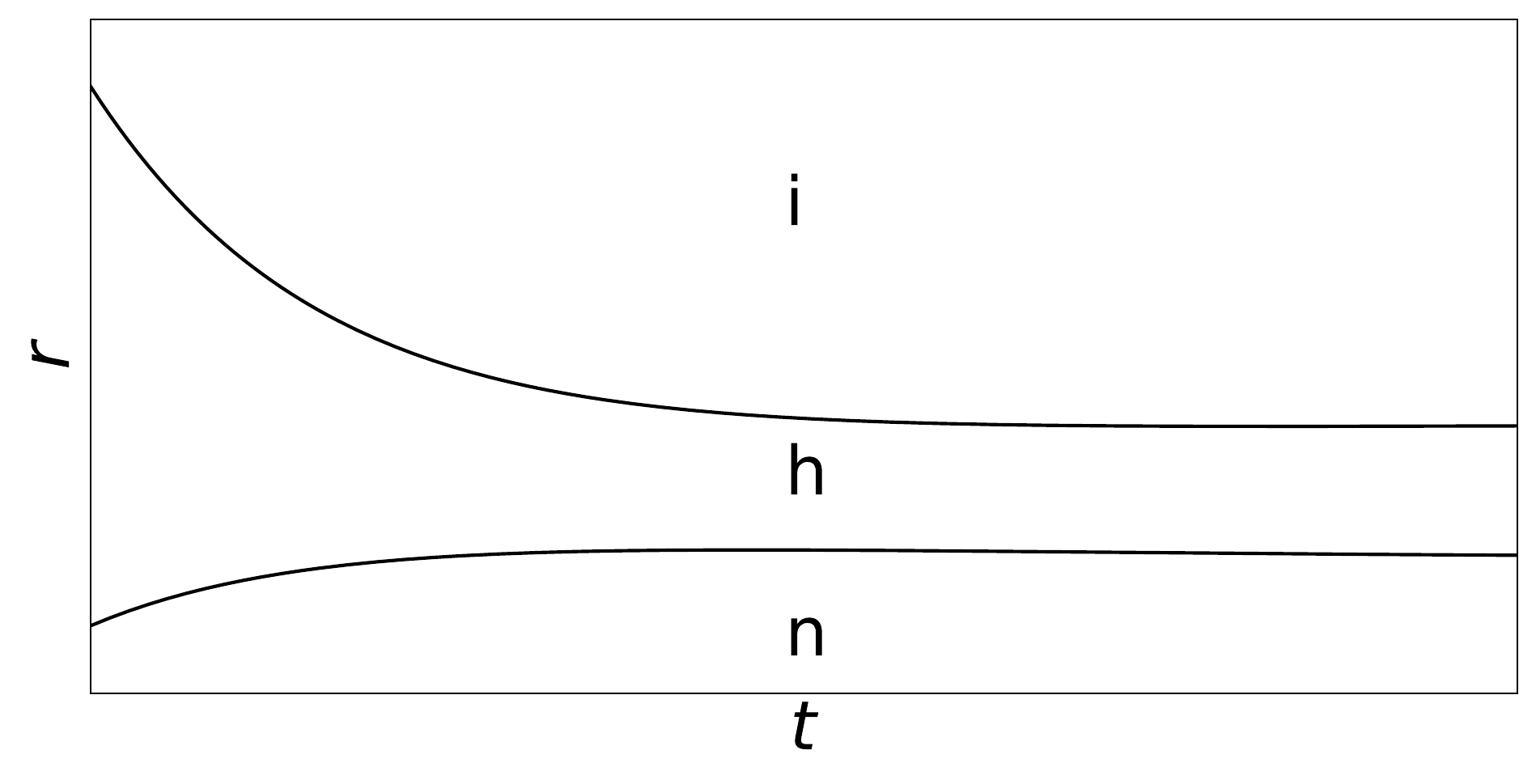}
		\caption{Forward curve with $\beta_1 < 0$.}
	\end{subfigure}
	\hfill
	\begin{subfigure}{0.49\textwidth}
		\includegraphics[width=\textwidth]{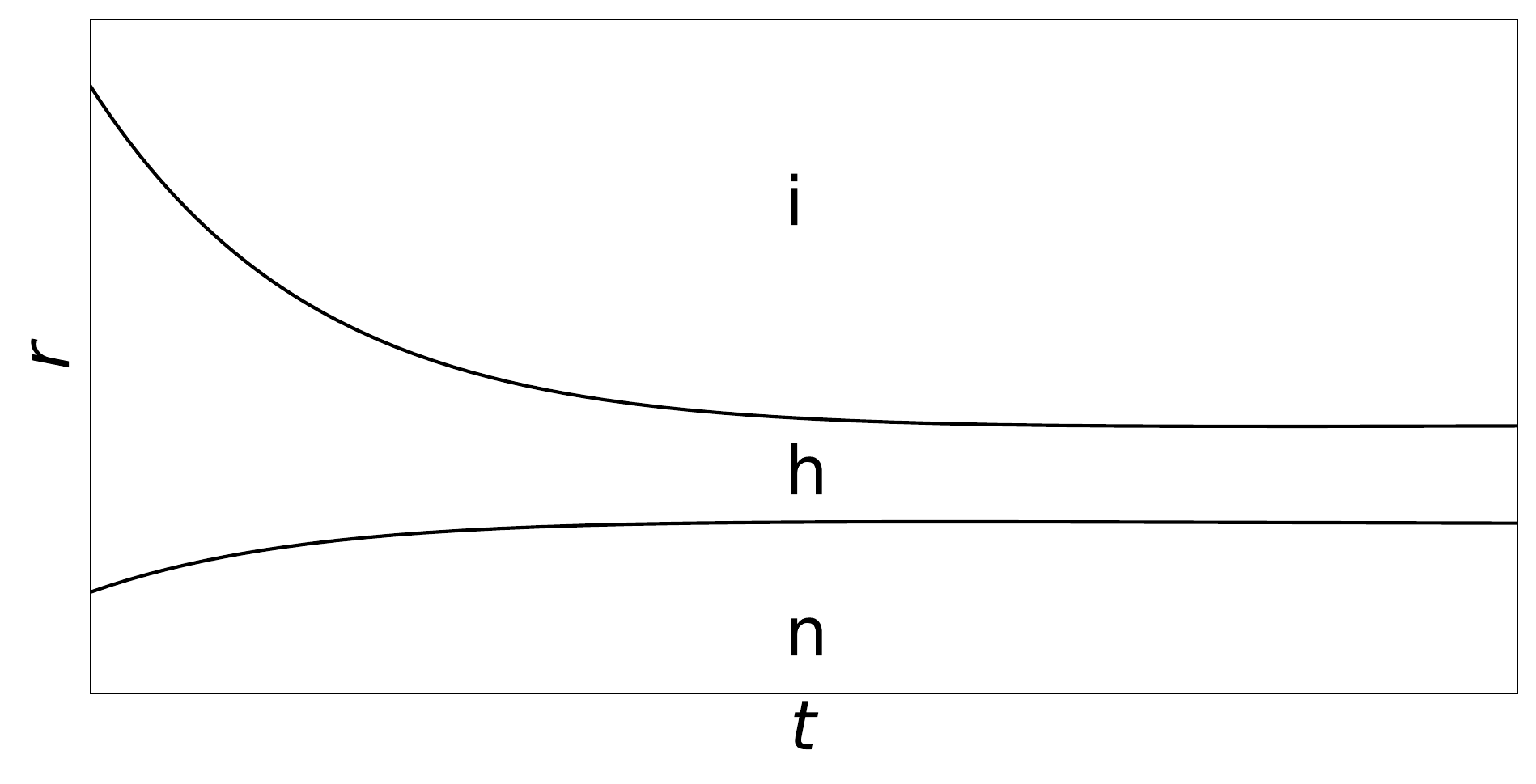}
		\caption{Yield curve with $\beta_1 < 0$.}
	\end{subfigure}
	
	\begin{subfigure}{0.49\textwidth}
		\includegraphics[width=\textwidth]{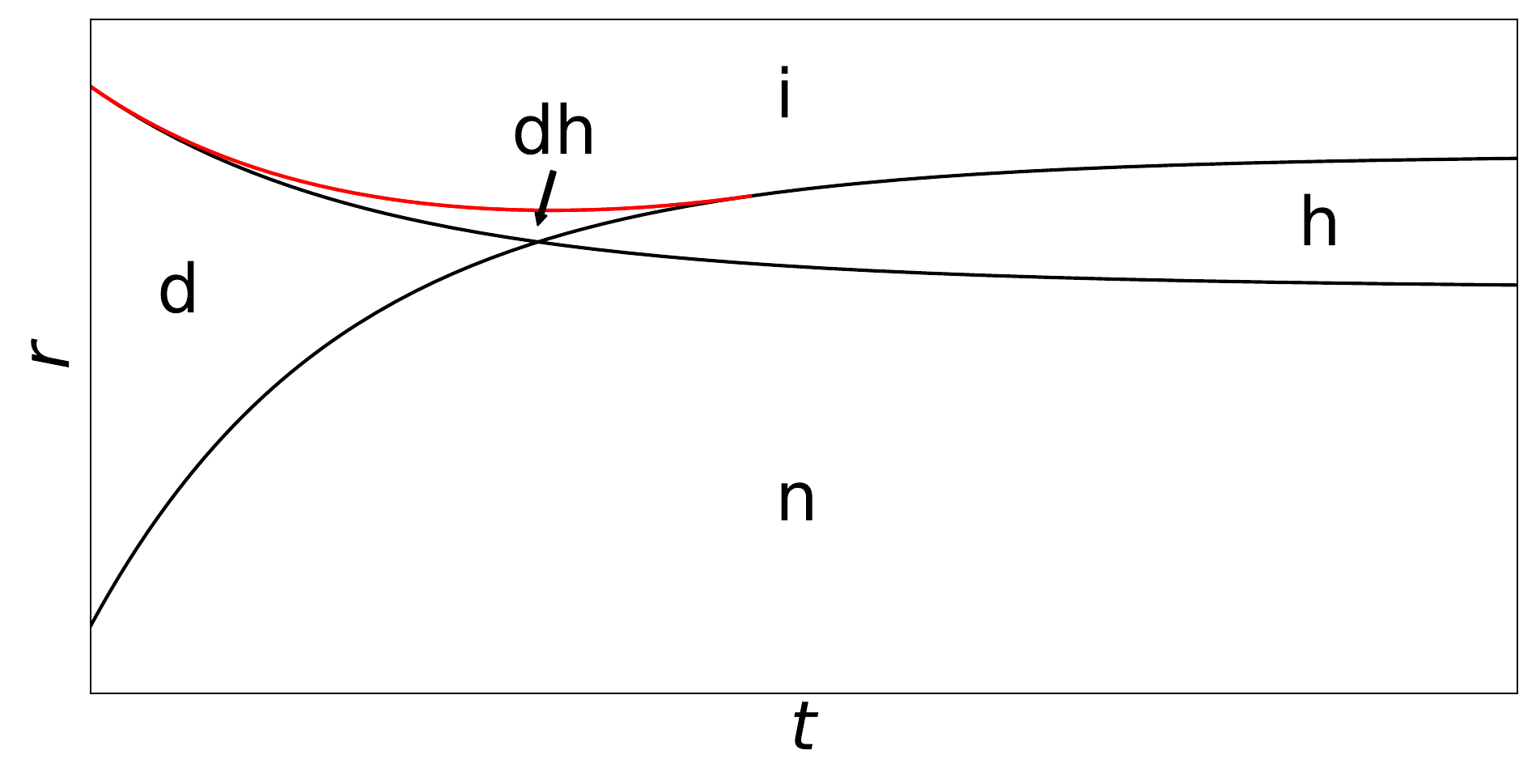}
		\caption{Forward curve with $\beta_1 > 0$.}
	\end{subfigure}
	\hfill
	\begin{subfigure}{0.49\textwidth}
		\includegraphics[width=\textwidth]{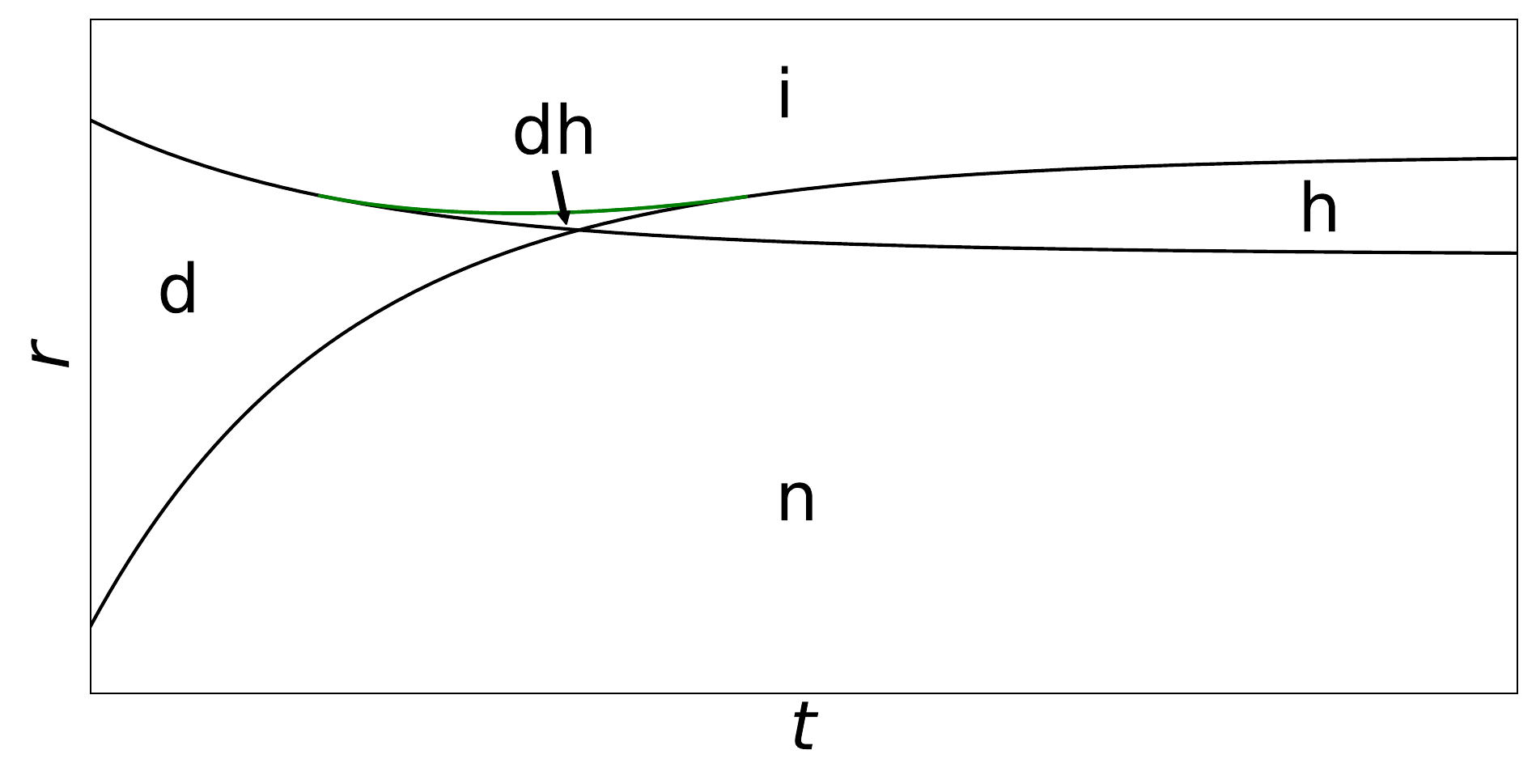}
		\caption{Yield curve with $\beta_1 > 0$.}
	\end{subfigure}
	\caption{Shapes of the term structure for monotone initial curves with $2\lambda < \frac{1}{\tau}$.}
	\label{fig:mon_indep_1_2}
\end{figure}

\subsection{Segmentation for $\lambda < \frac{1}{\tau} < 2\lambda$}

In  this case $\alpha$ is positive. Hence, the envelope $\eta^{\mathrm{f}}$ is non--empty if and only if $\gamma$ takes values in $(1,\infty)$.

For $\beta_1 < 0$, $\gamma$ is negative. Hence, there are no solutions to \eqref{eq:env_forward}. By a similar reasoning as in the case $2\lambda < \frac{1}{\tau}$, the shapes \texttt{normal, inverse} and \texttt{humped} are attainable for forward and yield curve.

For $\beta_1 > 0$, $\gamma$ is positive and solutions to \eqref{eq:env_forward} exist for every $t$ in an interval $T^{\mathrm{f}}$, which is unbounded and away from 0. The envelope $\eta^{\mathrm{f}}$ approaches $\zeta_0$ on its left end and $\zeta_\infty^{\mathrm{f}}$ asymptotically. There is exactly one intersection point $(t^*, r^*)$ of $\zeta_0$ and $\zeta_\infty^{\mathrm{f}}$ and since 
\begin{align*}
	\frac{\lambda}{\frac{1}{\tau} - \lambda} \in (1,\infty),
\end{align*}
$t^*$ is again an element of $T^{\mathrm{f}}$. Similar to the case $2\lambda < \frac{1}{\tau}$, it must hold $r < r_0(t)$ and $r < r_\infty^{\mathrm{f}}$ for every $(t,r)$ on the envelope. We conclude that the shapes \texttt{normal, inverse, humped, dipped, hd} are attainable.
To determine the shapes of the yield curve we study the initial and terminal direction of the tracking function. Still, the expressions \eqref{eq:init_direction_mon} and \eqref{eq:term_direction_mon} change their sign exactly once (from $\pp$ to $\mm$) and since
\begin{align*}
	\frac{1}{\tau}\left(\frac{1}{\lambda\tau} - 1\right) < -\frac{4}{3}\lambda\left(\lambda\tau - \frac{1}{\lambda\tau}\right) < 4\lambda^2\tau,
\end{align*}
the initial direction changes first, then $\zeta_0$ and $\zeta_\infty^{\mathrm{y}}$ intersect and lastly, the terminal direction of the tracking function changes. Consequently, the yield curve also attains the shapes \texttt{normal, inverse, humped, dipped, hd}. An illustration is provided in \cref{fig:mon_indep_3_4}.
\begin{figure}
	\centering
	\begin{subfigure}{0.49\textwidth}
		\includegraphics[width=\textwidth]{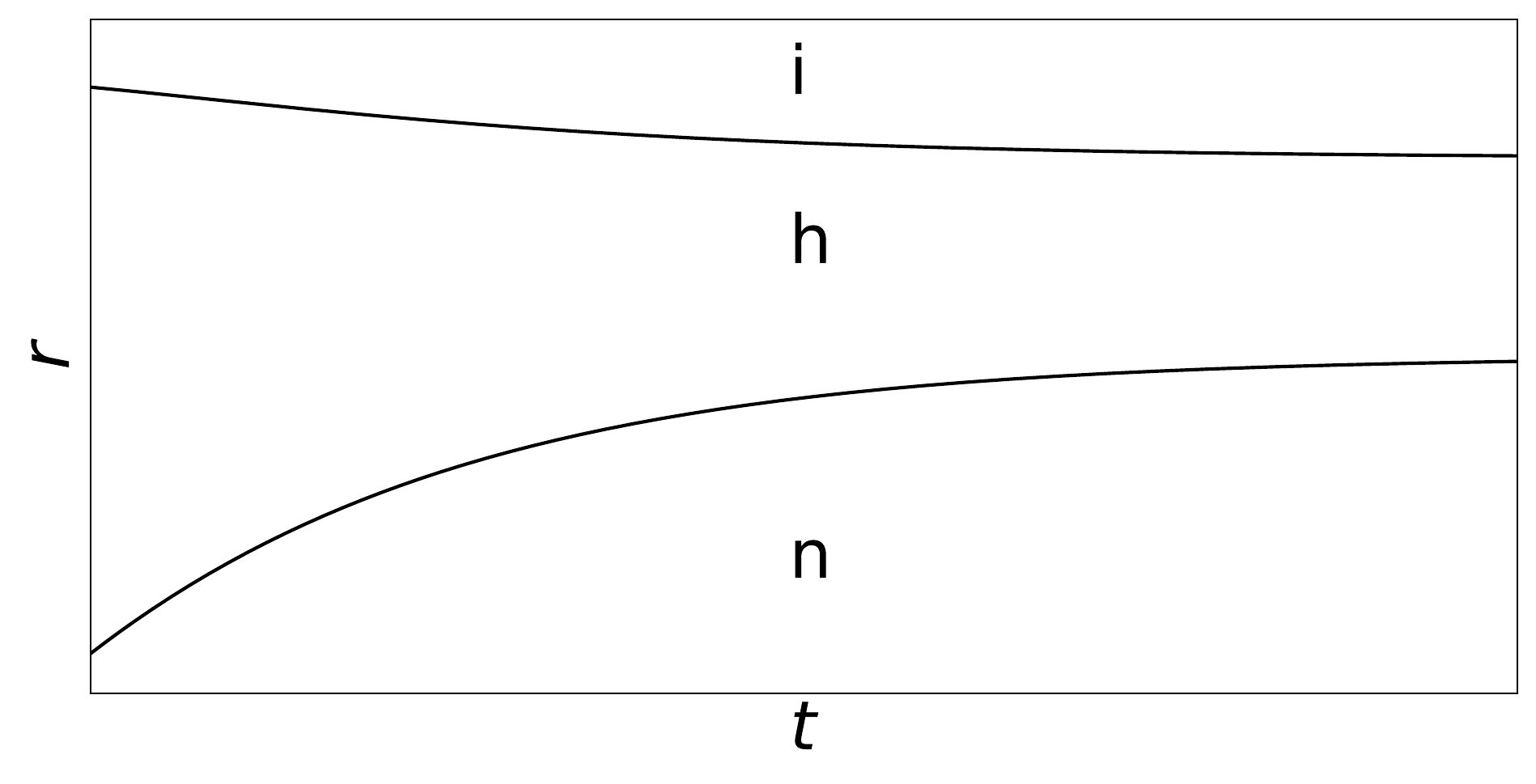}
		\caption{Forward curve with $\beta_1 < 0$.}
	\end{subfigure}
	\hfill
	\begin{subfigure}{0.49\textwidth}
		\includegraphics[width=\textwidth]{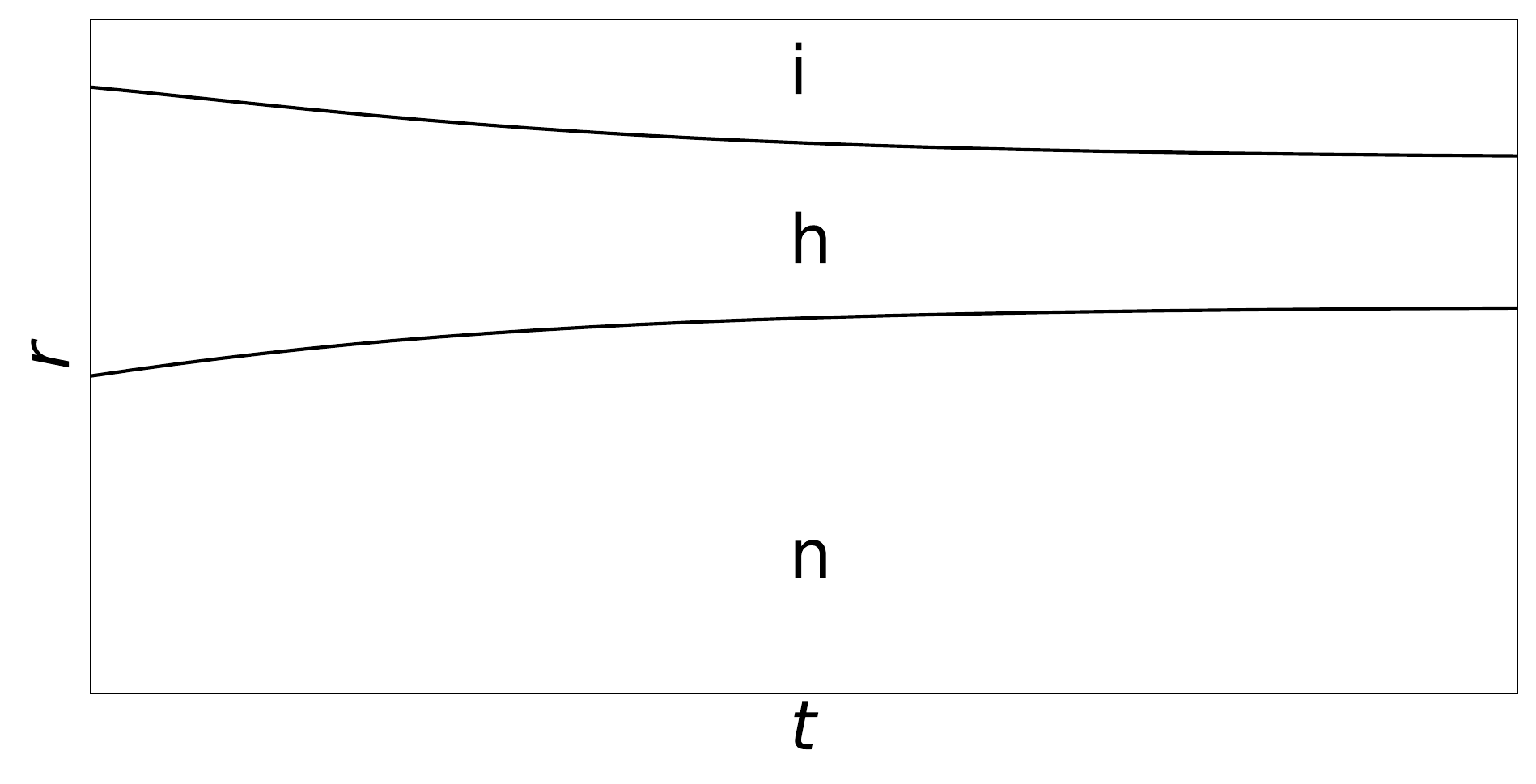}
		\caption{Yield curve with $\beta_1 < 0$.}
	\end{subfigure}
	
	\begin{subfigure}{0.49\textwidth}
		\includegraphics[width=\textwidth]{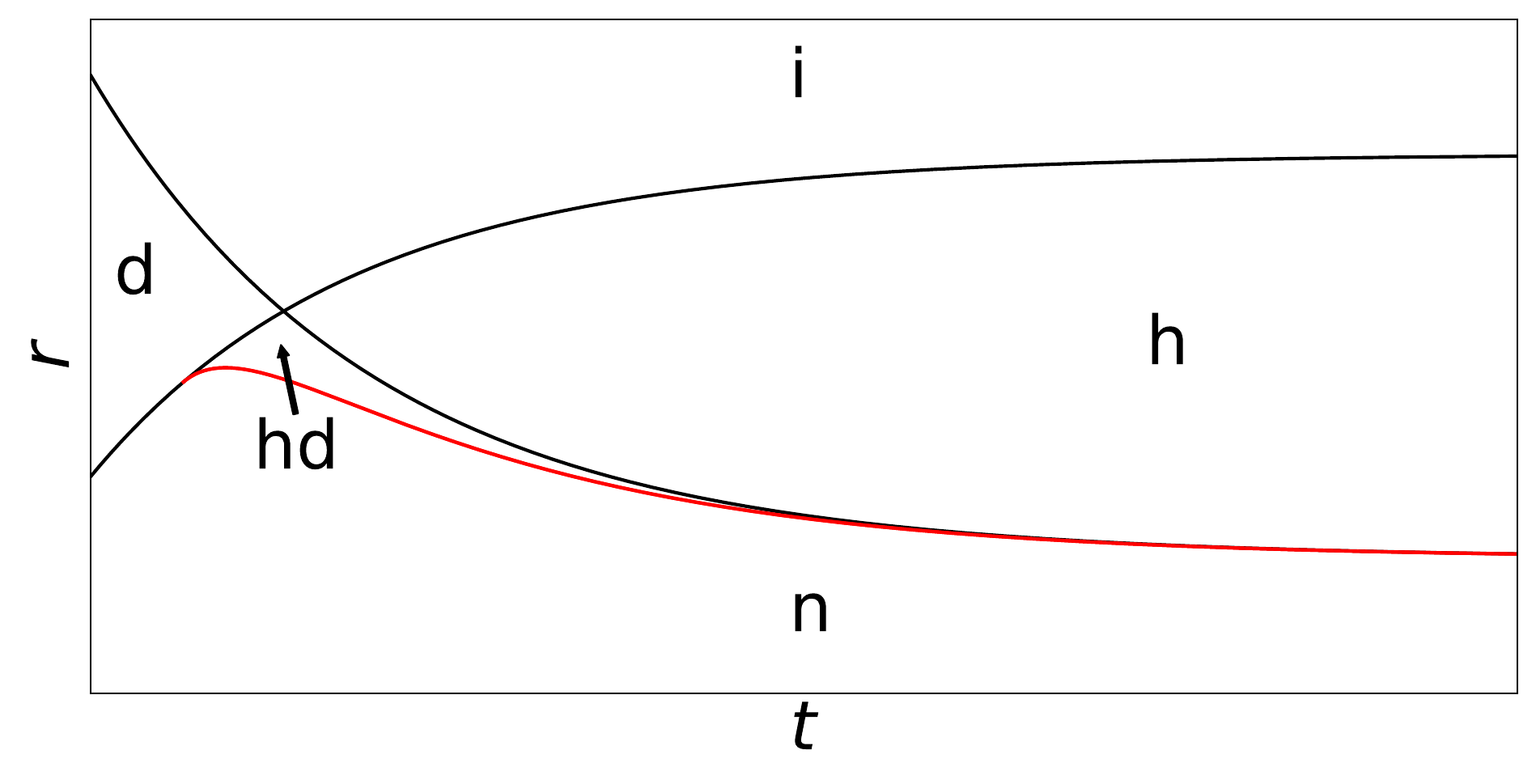}
		\caption{Forward curve with $\beta_1 > 0$.}
	\end{subfigure}
	\hfill
	\begin{subfigure}{0.49\textwidth}
		\includegraphics[width=\textwidth]{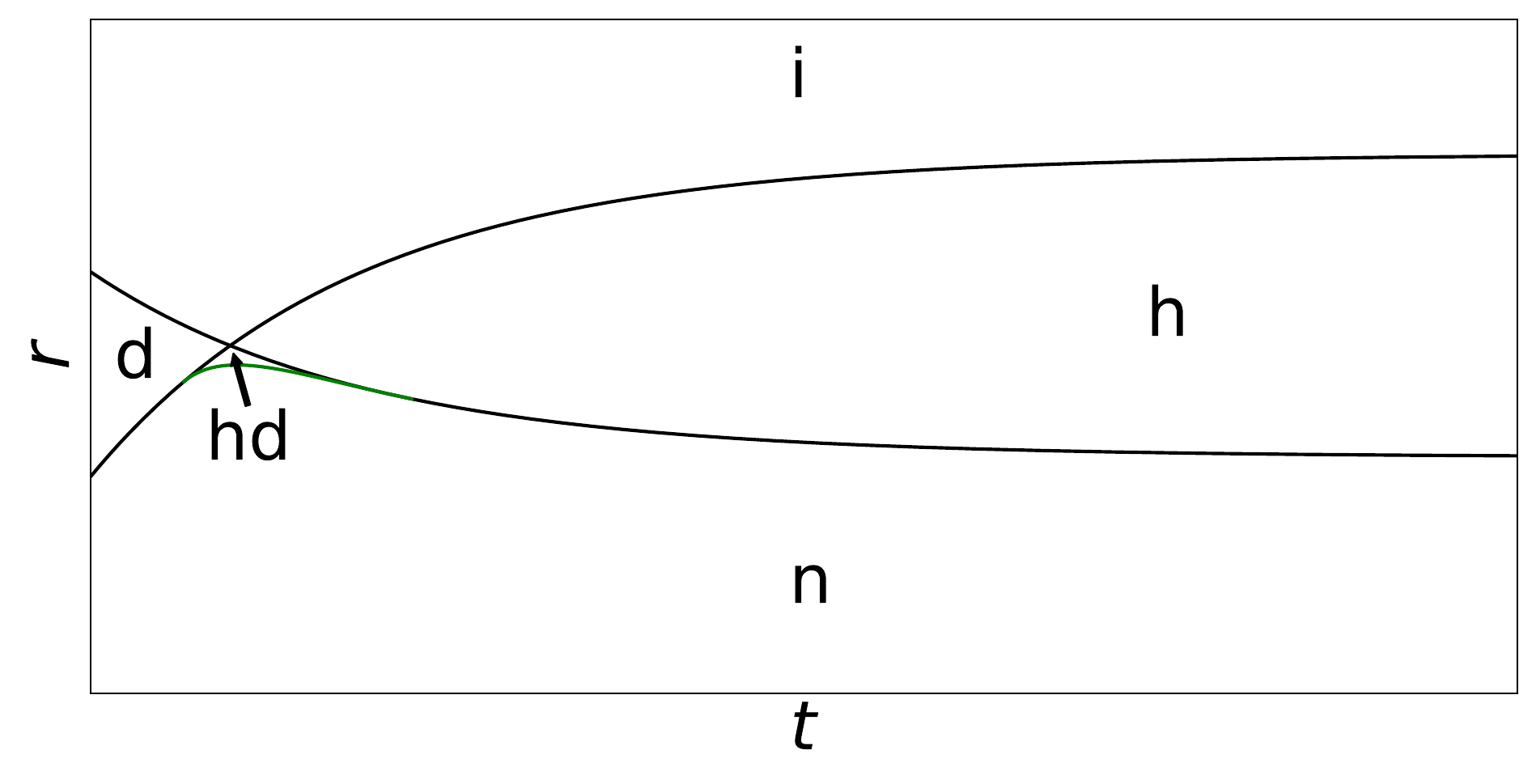}
		\caption{Yield curve with $\beta_1 > 0$.}
	\end{subfigure}
	\caption{Shapes of the term structure for monotone initial curves with $\lambda < \frac{1}{\tau} < 2\lambda$.}
	\label{fig:mon_indep_3_4}
\end{figure}

\subsection{Segmentation for $\frac{1}{\tau} < \lambda$}

In  this case, $\alpha$ is positive. Hence, the envelope $\eta^{\mathrm{f}}$ is non--empty if and only if $\gamma$ takes values in $(1,\infty)$.

For $\beta_1 > 0$, $\gamma$ is negative. Therefore, there are no solutions to \eqref{eq:env_forward}. The shape of the forward curve is fully characterized by the initial and terminal sign of its derivative. The terminal sign is always $\mm$. Thus, the attainable shapes are \texttt{inverse} and \texttt{humped}. The yield curve cannot attain  more complex shapes and since $\zeta_0$ and $\zeta_\infty^{\mathrm{y}}$ are non--empty and do not intersect, the shapes \texttt{normal, inverse} and \texttt{humped} are attainable.

For $\beta_1 < 0$, $\gamma$ is positive and solutions to \eqref{eq:env_forward} exist for every $t$ in an interval $T^{\mathrm{f}}$, which is unbounded and away from 0. The set $\zeta_\infty^{\mathrm{f}}$ is empty. The envelope approaches $\zeta_0$ on its left end and does not intersect with $\zeta_0$. The terminal sign of the derivative of the forward curve is always $\pp$. Thus, the shapes \texttt{normal, dipped} and \texttt{hd} are attainable. 
Terminally, the tracking function is always increasing in the yield case. Its initial direction changes from increasing to decreasing when \eqref{eq:init_direction_mon} changes its sign, so the envelope exists for every $t$ in an unbounded interval $T^{\mathrm{y}}$, which is away from 0. There is exactly one intersection point $(t^*, r^*)$ of $\zeta_0$ and $\zeta_\infty$ and since
\begin{align*}
	\frac{1}{\tau}\left(\frac{1}{\lambda\tau} - 1\right) < \frac{4}{3}\lambda\left(\lambda\tau - \frac{1}{\lambda\tau}\right)\quad\Longleftrightarrow\quad \frac{1}{\tau} < \lambda,
\end{align*}
it must hold $t^*\in T^{\mathrm{y}}$. Consequently, the yield curve attains the shapes \texttt{normal, inverse, humped, dipped} and \texttt{hd}. An illustration is provided in \cref{fig:mon_indep_5_6}.
\begin{figure}
	\centering
	\begin{subfigure}{0.49\textwidth}
		\includegraphics[width=\textwidth]{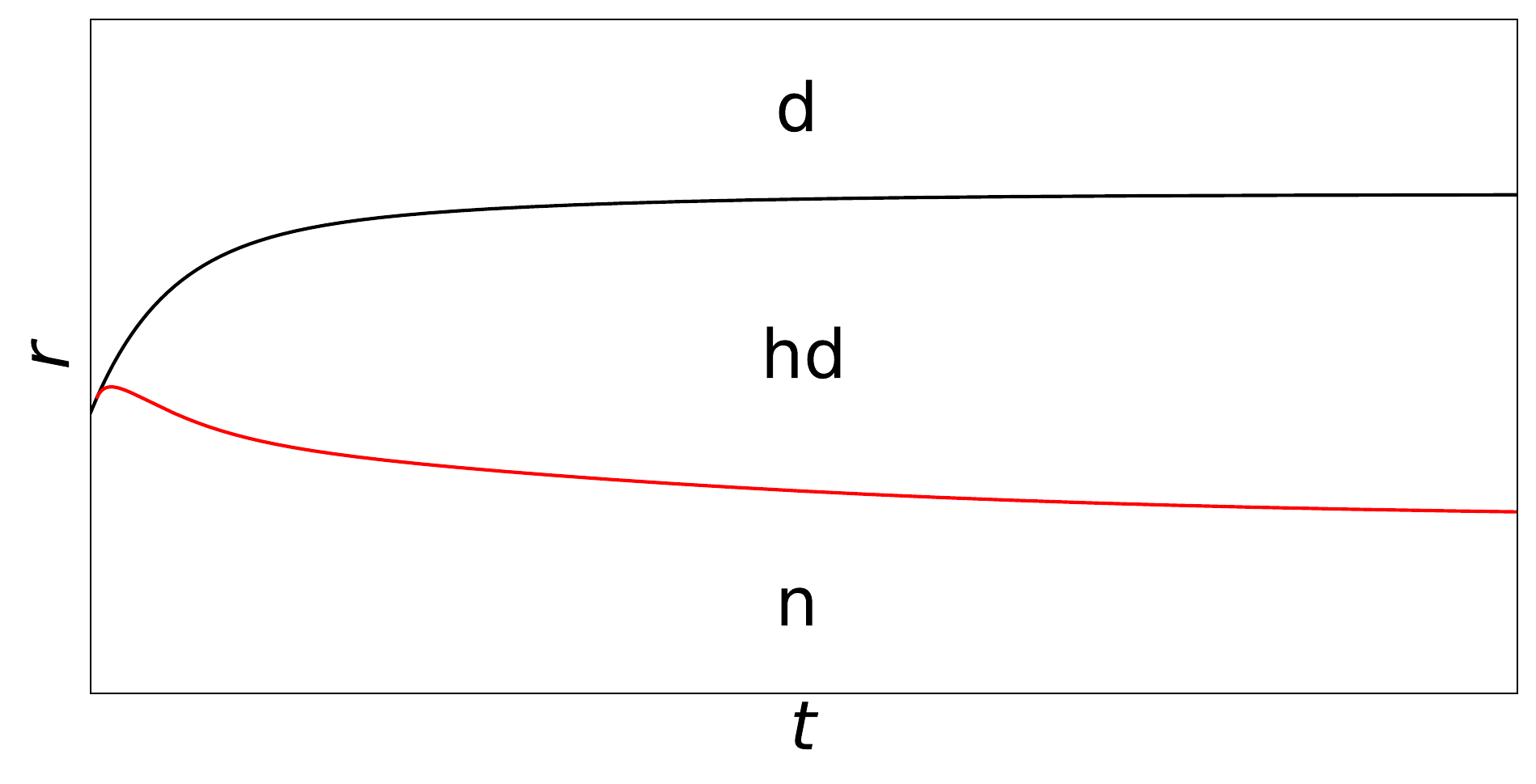}
		\caption{Forward curve with $\beta_1 < 0$.}
	\end{subfigure}
	\hfill
	\begin{subfigure}{0.49\textwidth}
		\includegraphics[width=\textwidth]{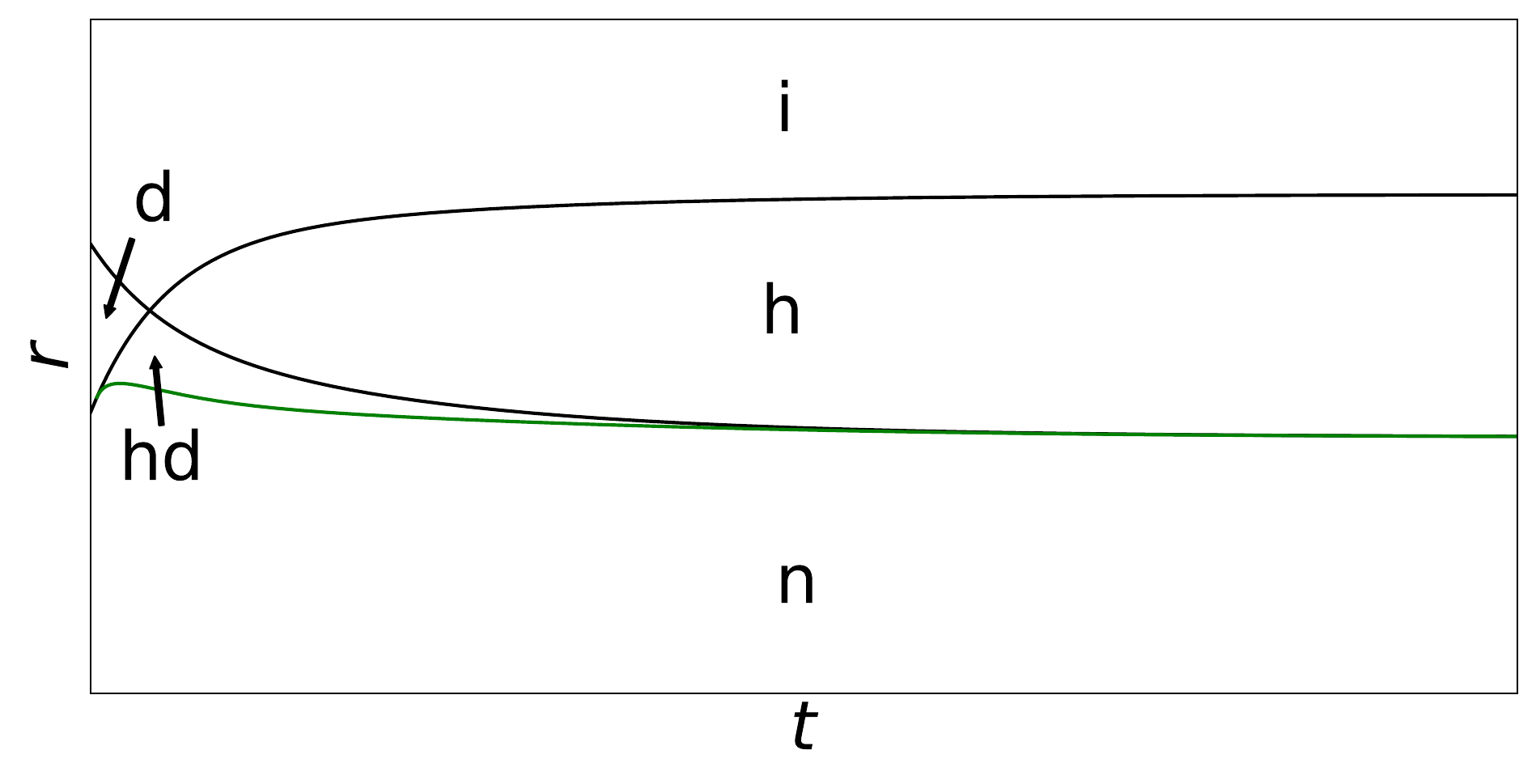}
		\caption{Yield curve with $\beta_1 < 0$.}
	\end{subfigure}
	
	\begin{subfigure}{0.49\textwidth}
		\includegraphics[width=\textwidth]{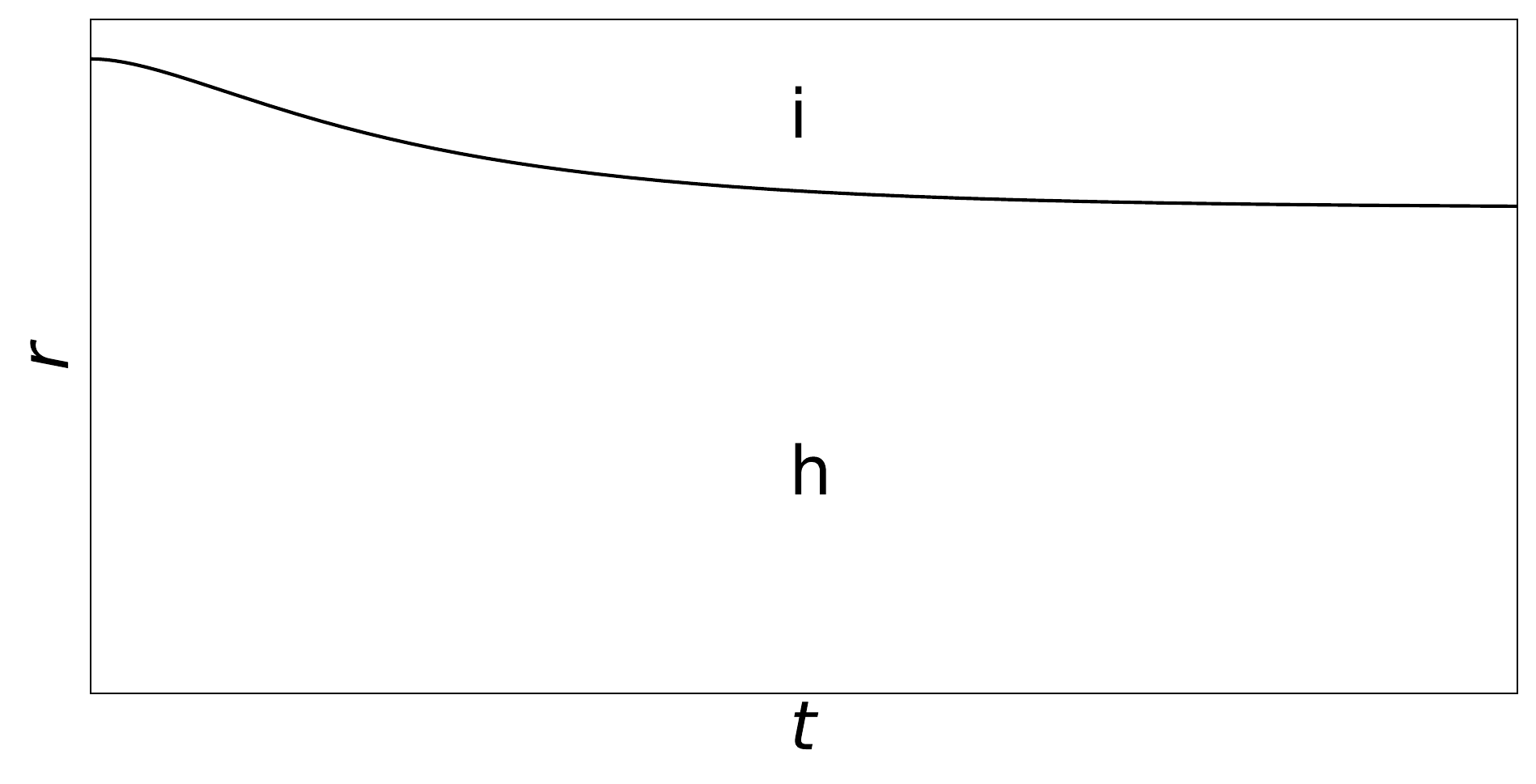}
		\caption{Forward curve with $\beta_1 > 0$.}
	\end{subfigure}
	\hfill
	\begin{subfigure}{0.49\textwidth}
		\includegraphics[width=\textwidth]{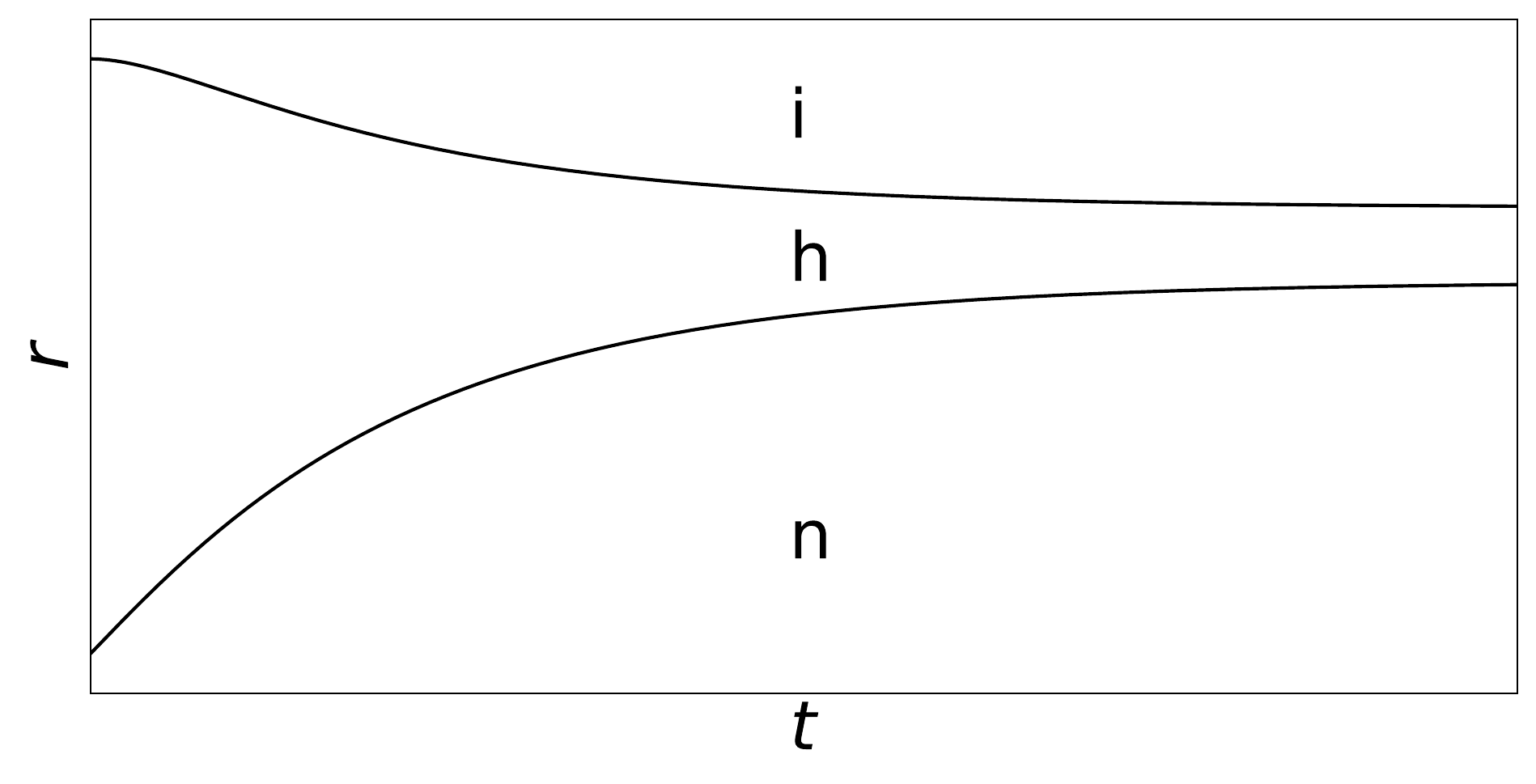}
		\caption{Yield curve with $\beta_1 > 0$.}
	\end{subfigure}
	\caption{Shapes of the term structure for monotone initial curves with $\frac{1}{\tau} < \lambda$.}
	\label{fig:mon_indep_5_6}
\end{figure}

\section{Nelson--Siegel--parameterized initial curves}

\subsection{Segmentation for $\lambda = \frac{1}{\tau}$}	

We start again by studying the case $\lambda = \frac{1}{\tau}$. Here, forward and yield curve possess at most two local extrema; see Remark~\ref{rem:bounds}. 
The initial sign of the derivatives of forward and yield curve can be determined from
\begin{align*}
	r_0(t) = \beta_0 + \beta_2{\mathrm{e}}^{-\lambda t} + \frac{\sigma^2}{2\lambda^2}\left(1 - {\mathrm{e}}^{-2\lambda t}\right).
\end{align*}
Terminally, the forward curve is increasing if $\beta_2 < 0$, and decreasing if $\beta_2 > 0$. The terminal monotonic behavior of the yield curve is determined from
\begin{align*}
	r_\infty^{\mathrm{y}}(t) = \beta_0 - \beta_2{\mathrm{e}}^{-\lambda t} + \frac{\sigma^2}{2\lambda^2}\left(1 - {\mathrm{e}}^{-2\lambda t}\right).
\end{align*}
Solutions to \eqref{eq:env_2} are given by
\begin{align}\label{eq:Nelson_Siegel_dep_env}
	x = \frac{1}{\lambda}\log\Big(\underbrace{-\frac{\sigma^2\tau^2}{\beta_2}\left(1-{\mathrm{e}}^{-2\lambda t}\right){\mathrm{e}}^{\lambda t}}_{=:\delta(t)}\Big).
\end{align}
The initial direction of the tracking function is determined by the sign of
\begin{align}\label{eq:init_direction_Nelson_Siegel_dep}
	-\beta_2{\mathrm{e}}^{-\lambda  t} - \frac{\sigma^2}{\lambda^2}\left(1 - {\mathrm{e}}^{-2\lambda t}\right)
\end{align}
and the terminal direction is determined by the sign of $-\beta_2$, both in the forward case and the yield case.

If $\beta_2 < 0$, $\delta$ is monotonically increasing from 0 to $\infty$, hence, $T^{\mathrm{f}}$ is an unbounded interval and away from 0; the envelope approaches $\zeta_0^{\mathrm{f}}$ on its left end, but the sets do not intersect; $\zeta_\infty^{\mathrm{f}}$ is empty. The forward curve is terminally increasing. Thus, the shapes \texttt{normal, dipped} and \texttt{hd} are attainable.
In the case of the yield curve, the tracking function is terminally increasing. Its initial direction changes from increasing to decreasing when \eqref{eq:init_direction_Nelson_Siegel_dep} changes its sign, so $T^{\mathrm{y}}$ is an unbounded interval and away from 0. There is exactly one intersection point $(t^*, r^*)$ of $\zeta_0$ and $\zeta_\infty$, since
\begin{align*}
	r_0(t) - r_\infty^{\mathrm{y}}(t) = 2\beta_2{\mathrm{e}}^{-\lambda t} + \frac{3\sigma^2}{4\lambda^2}\left(1 - {\mathrm{e}}^{-2\lambda t}\right),
\end{align*}
and $t^*\in T^{\mathrm{y}}$. Consequently, the yield curve attains the shapes \texttt{normal, inverse, humped, dipped} and \texttt{hd}.

If $\beta_2 < 0$, $\delta$ is monotonically decreasing from 0 to $-\infty$, hence, $T^{\mathrm{f}}$ is empty; and so is $\zeta_\infty^{\mathrm{f}}$. The shape of the forward curve is purely determined by the initial and terminal sign of its derivative. The terminal sign is always $\mm$, so the shapes \texttt{inverse} and \texttt{humped} are attainable. The yield curve cannot attain  more complex shapes and since $\zeta_0$ and $\zeta_\infty^{\mathrm{y}}$ are non--empty and do not intersect, the shapes \texttt{normal, inverse} and \texttt{humped} are attainable. An illustration is provided in \cref{fig:Nelson_Siegel_dep_1_2}.	
\begin{figure}
	\centering
	\begin{subfigure}{0.49\textwidth}
		\includegraphics[width=\textwidth]{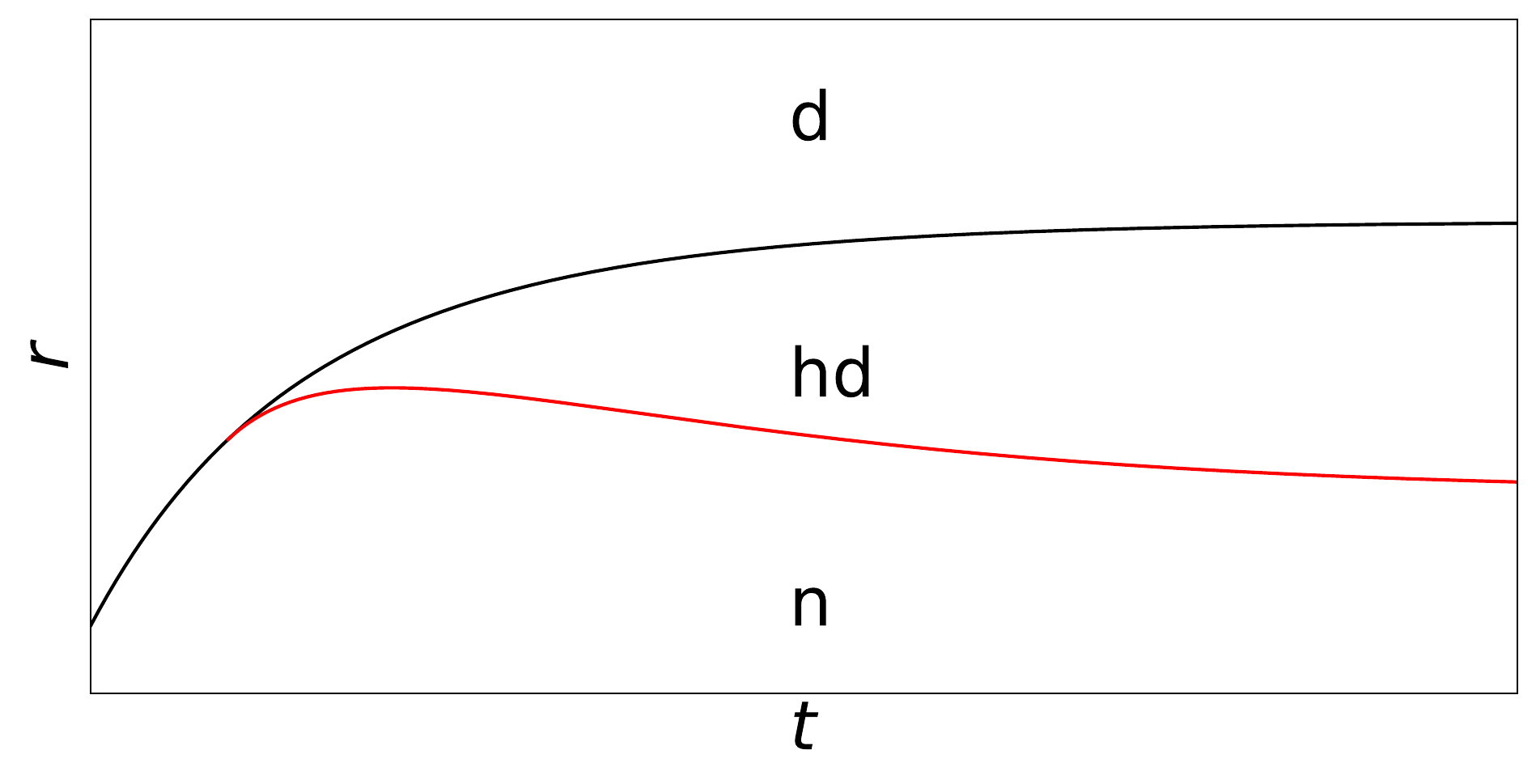}
		\caption{Forward curve with $\beta_2 < 0$.}
	\end{subfigure}
	\hfill
	\begin{subfigure}{0.49\textwidth}
		\includegraphics[width=\textwidth]{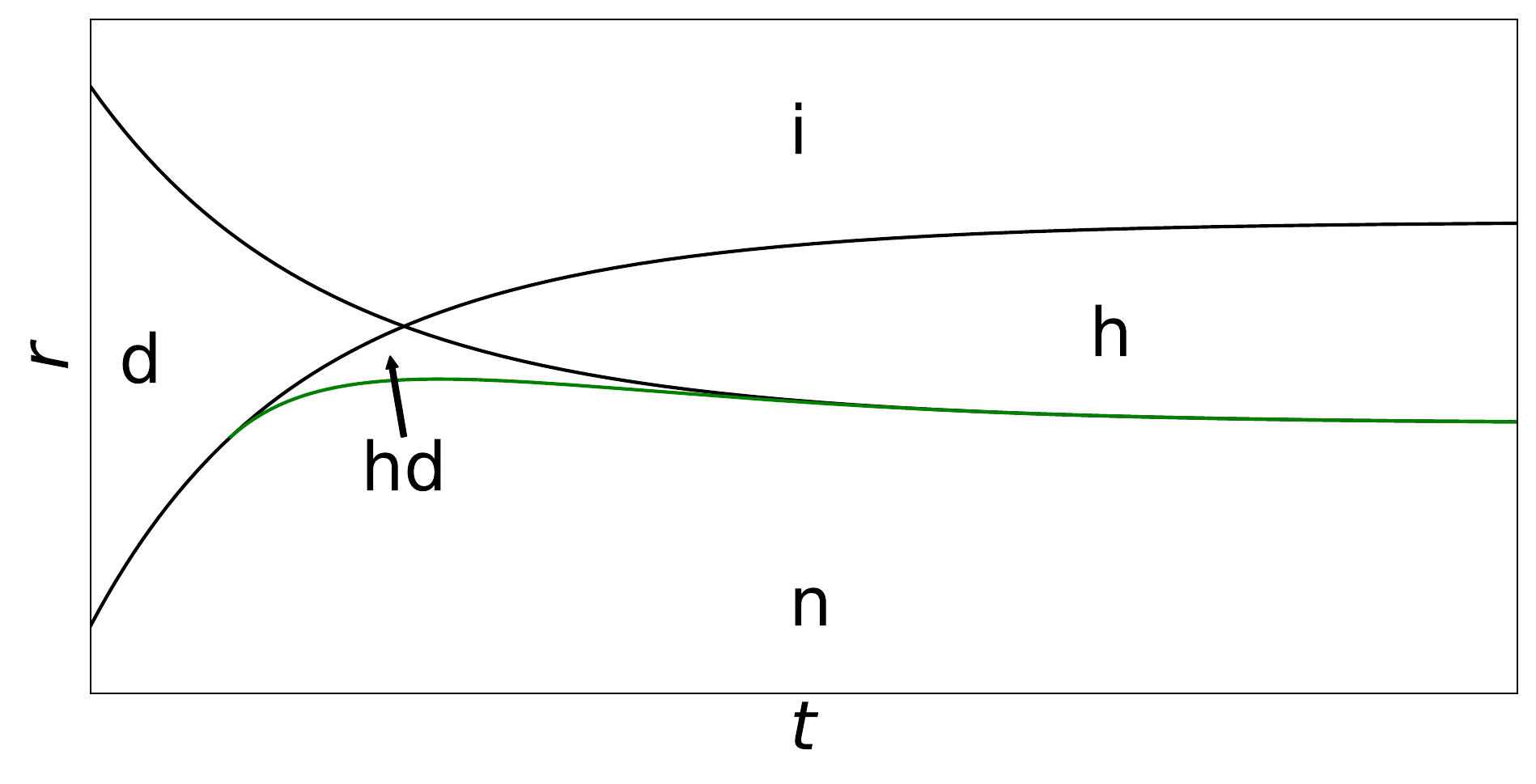}
		\caption{Yield curve with $\beta_2 < 0$.}
	\end{subfigure}
	
	\begin{subfigure}{0.49\textwidth}
		\includegraphics[width=\textwidth]{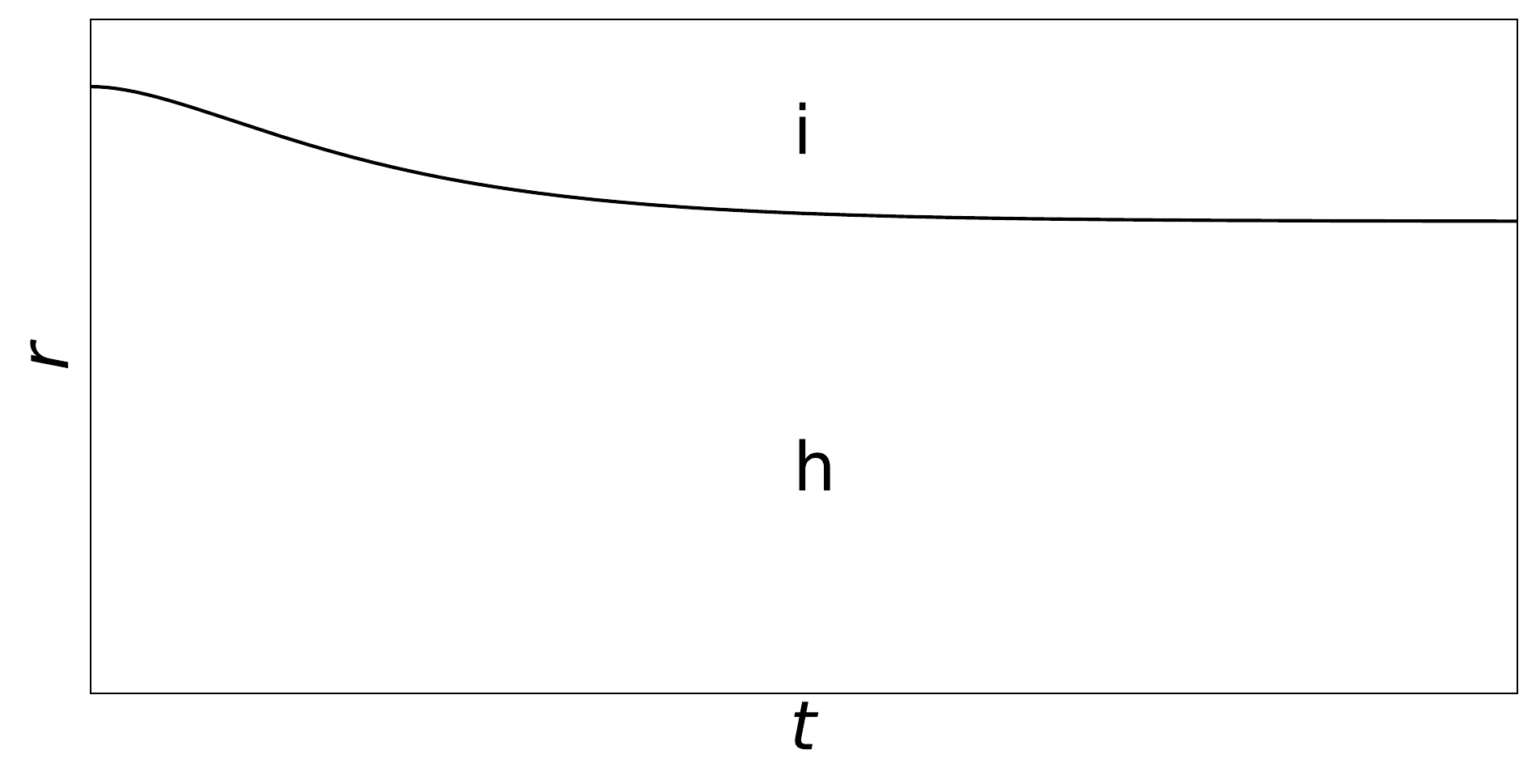}
		\caption{Forward curve with $\beta_2 > 0$.}
	\end{subfigure}
	\hfill
	\begin{subfigure}{0.49\textwidth}
		\includegraphics[width=\textwidth]{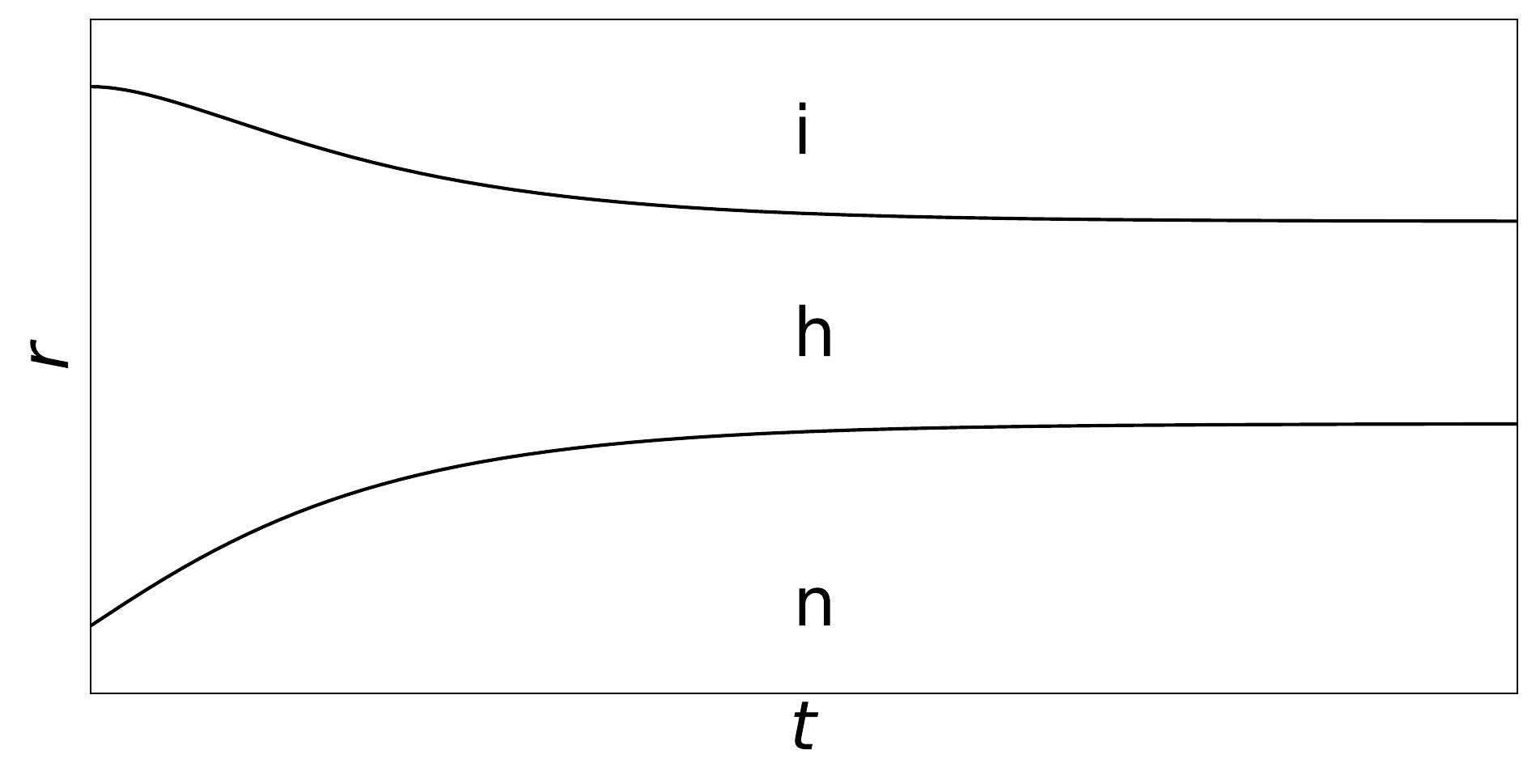}
		\caption{Yield curve with $\beta_2 > 0$.}
	\end{subfigure}
	\caption{Shapes of the term structure for Nelson--Siegel parameterized initial curves with $\frac{1}{\tau} = \lambda$.}
	\label{fig:Nelson_Siegel_dep_1_2}
\end{figure}

\subsection{Segmentation in the general case}	

Let us first consider the case of the forward curve here and assume that $2\lambda \neq \frac{1}{\tau} \neq \lambda$. With an initial curve from the Nelson--Siegel family, a rearranged equation \eqref{eq:env_2} reads
\begin{align}\label{eq:env_2_NS}
	\left(x + \phi(t)\right){\mathrm{e}}^{\rho x} = \xi(t),
\end{align}
where
\begin{align*}
	\phi(t) &:= t + \left(\frac{\beta_1}{\beta_2} - 1 - \frac{1}{1 - \lambda\tau_1}\right)\tau_1,\\
	\rho &:= - \left(\frac{1}{\tau_1} - 2\lambda\right),\\
	\xi(t) &:= \frac{\sigma^2\tau_1^2\left(1 - \exp\left(-2\lambda t\right)\right)\exp\left(\frac{t}{\tau_1}\right)}{\left(\frac{1}{\tau_1} - \lambda\right)\beta_2}.
\end{align*}
By \cite[Lemma A.3]{KRS26}, solutions of \eqref{eq:env_2_NS} must satisfy
\begin{align}\label{eq:Lambert-W-NS}
	x = \frac{{\mathcal W}(\rho\xi(t){\mathrm{e}}^{\phi(t)\chi})}{\chi} - \phi(t),
\end{align}
where ${\mathcal W}$ could either be the branch ${\mathcal W}_0$ or the branch ${\mathcal W}_{-1}$ of the Lambert ${\mathcal W}$-function \cite{CG96}. Hence, there can be two solutions of \eqref{eq:env_2} corresponding to the same value of $t$; see \cref{fig:Nelson_Siegel_1D} for a visualization. Also, applying the theory of Tchebycheff systems, we find that there are up to three intersections of $\zeta_0$ and $\zeta_\infty^{\mathrm{f}}$; it can be shown that this bound is actually sharp.
\begin{figure}
	\centering
	\includegraphics[width=\textwidth]{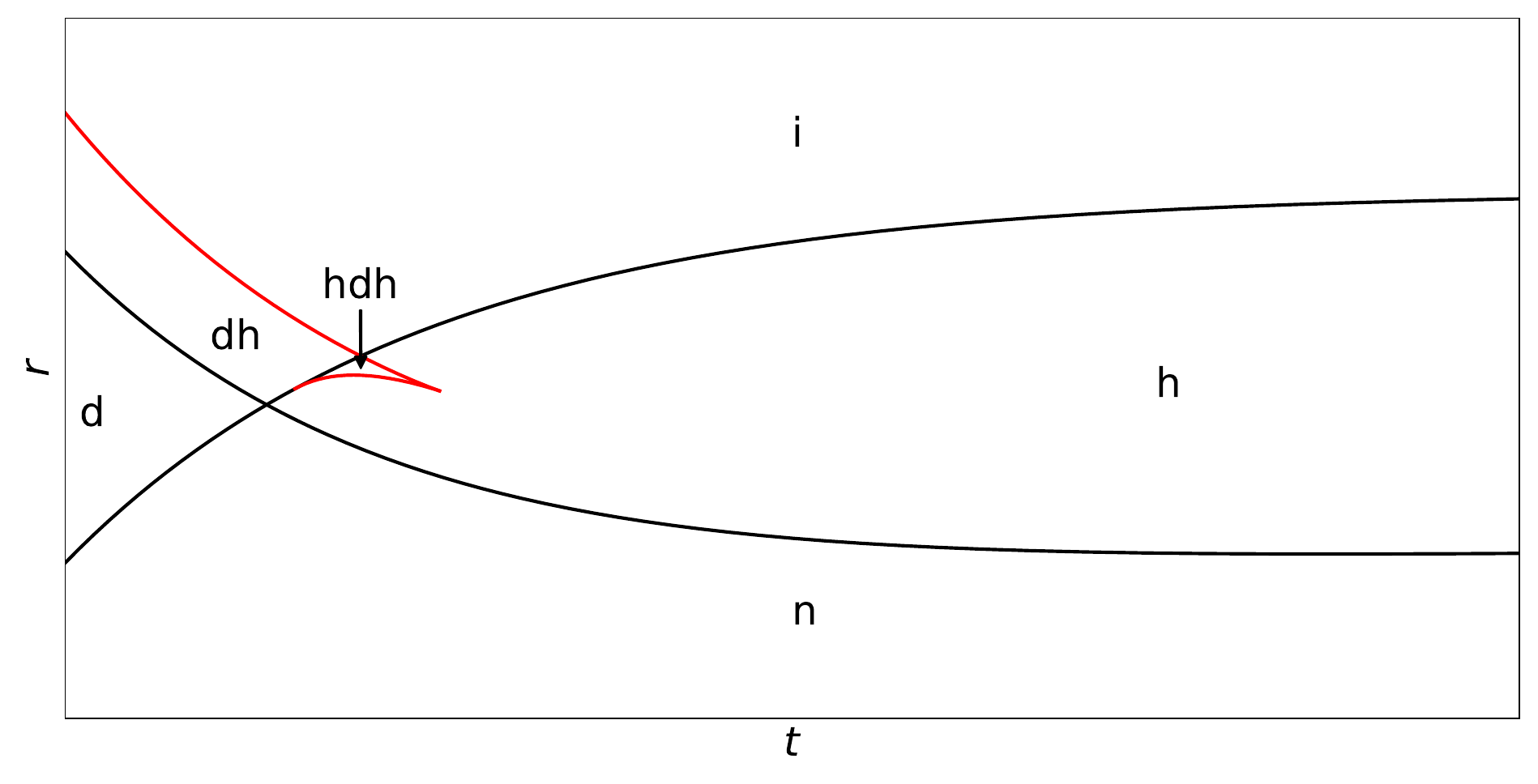}
	\caption{Shapes of the forward curve for one choice of Nelson--Siegel parameterized initial curves.}
	\label{fig:Nelson_Siegel_1D}
\end{figure}

If we were now to continue in the same manner as before and would try to determine how the sets $\zeta_0, \zeta_\infty^{\mathrm{f}}$ and $\eta^{\mathrm{f}}$ are positioned relative to each other in $\Theta$, and where they intersect, in an attempt to fully understand the full evolution of shapes in the model, we would quickly notice that it is getting increasingly difficult and tedious to do so with the sheer number of possibilities and formulae getting longer and longer; all while gaining very little novel insights into the model. Thus, we will now change our point of view instead and (temporarily) abandon our attempt to understand the evolution of shapes, so to  still achieve a complete classification of attainable term structure shapes and provide a complete segmentation of the state space. General results about the (asymptotic) evolution of shapes, however, will be provided in \cref{sec:asymptotics}.

Our goal here is to get into the setting introduced in \cite{KRS23} and refined in \cite{KRS26}, where the envelope was derived from a family of lines, yielding, under mild assumptions, a winding number formula for determining the number of local extrema, which works especially well when there are at most three local extrema, as is the case here.
In order to do so, we introduce transformed parameters
\begin{align*}
	\widetilde{\beta}_0 &:= f_0(t) - \frac{\sigma^2}{2\lambda^2}\left(1 - {\mathrm{e}}^{-2\lambda t}\right),\\
	\widetilde{\beta}_1 &:=\left(\frac{\beta_2}{\tau^2}t - \frac{\beta_2 - \beta_1}{\tau}\right){\tau^2}{\mathrm{e}}^{-\frac{t}{\tau}},\\
	\widetilde{\beta}_2 &:= \frac{\beta_2}{\tau^2}{\mathrm{e}}^{-\frac{t}{\tau}},\\
	\widetilde{\sigma}^2 &:= \frac{\sigma^2}{2\lambda^2}\left(1 - {\mathrm{e}}^{-2\lambda t}\right).
\end{align*}
Note that, for every fixed choice of $(\lambda, \tau, t)\in\RR_{ >0}^3$ there is a one-to-one correspondence between the old parameters $(\beta_0,\beta_1,\beta_2,\sigma)$ and the new parameters $(\widetilde\beta_0,\widetilde\beta_1,\widetilde\beta_2,\widetilde\sigma)$, and the parameters $\beta_2$ and $\widetilde{\beta}_2$ even share the same sign. With the new parameters we define
\begin{align*}
	a_{\mathrm{f}}(x) &:= -\widetilde{\beta}_2\cdot x{\tau^2}{\mathrm{e}}^{-\frac{x}{\tau}} + \widetilde{\beta}_0\cdot\lambda{\mathrm{e}}^{-\lambda x} + \widetilde{\sigma}^2\cdot 2\lambda{\mathrm{e}}^{-2\lambda x},\\
	b_{\mathrm{f}}(x) &:= -{\mathrm{e}}^{-\frac{x}{\tau}},\\
	c_{\mathrm{f}}(x) &:= -\lambda{\mathrm{e}}^{-\lambda x},
\end{align*}
and, with different basis functions for the yield curve,
\begin{align*}
	a_{\mathrm{y}}(x) &:= \widetilde{\beta}_2\cdot \frac{\tau^5h_2(x,\tau)}{x^2} - \widetilde{\beta}_0\cdot\frac{h_1(x,\frac{1}{\lambda})}{\lambda x^2} - \widetilde{\sigma}^2\cdot \frac{h_1(x,\frac{1}{2\lambda})}{2\lambda x^2},\\
	b_{\mathrm{y}}(x) &:= \tau^2\frac{h_1(x,\tau)}{x^2},\\
	c_{\mathrm{y}}(x) &:= \frac{h_1(x,\frac{1}{\lambda})}{\lambda x^2},
\end{align*}
omitting the dependence on any variables other than $x$. The derivative of the forward resp. yield curve is then given by
\begin{align*}
	\partial_x f(x) &= a_{\mathrm{f}}(x) + b_{\mathrm{f}}(x)\widetilde{\beta}_1 + c_{\mathrm{f}}(x)r,\\
	\partial_x y(x) &= a_{\mathrm{y}}(x) + b_{\mathrm{y}}(x)\widetilde{\beta}_1 + c_{\mathrm{y}}(x)r.
\end{align*}
As practiced before, we will drop the subscripts f and y and only add them again whenever it becomes necessary to distinguish the two cases.

The coefficient $b(x)$ does not vanish for $x \in (0,\infty)$, hence, for every $x\in(0,\infty)$,
\begin{align*}
	\ell_x := \{(\widetilde{\beta}_1, r)\in\Theta'\ |\ a(x) + b(x)\widetilde{\beta}_1 + c(x)r = 0\}.
\end{align*}
defines a line in the plane $\Theta' := \RR^2$, all of which are collected in a family
\begin{align*}
	\mathcal{F} := (\ell_x)_{x\in(0,\infty)}.
\end{align*}
A summarized explanation of the relevance of this family was provided in \cite{KRS26}; we supply an adapted version here: 
\begin{quote}
	\textit{
		For every choice of $(\widetilde{\beta}_0, \widetilde{\beta}_1, \widetilde{\beta}_2, \tau), (\lambda, \widetilde{\sigma})$ and $t$ with their counterparts $(\beta_0,\beta_1,\beta_2,\tau)$ and $(\lambda,\sigma)$, the forward curve $x\mapsto f(x,t,r)$ (yield curve $x\mapsto y(x,t,r)$) in the Hull--White model with parameters $(\lambda,\sigma)$ and Nelson--Siegel--parameterized initial curves with parameters $(\beta_0,\beta_1,\beta_2,\tau)$ has a local extremum at $x$, if $\ell_x$ passes through the point $(\widetilde{\beta}_1, r)$.}
\end{quote}
The envelope $\eta$ of $\mathcal{F}$ consists of all points $(\widetilde{\beta}_1, r)\in\Theta'$, which satisfy
\begin{align}\label{eq:env_lines}
	\begin{cases}
		a(x) + b(x)\widetilde{\beta}_1 + c(x)r = 0,\\
		a'(x) + b'(x)\widetilde{\beta}_1 + c'(x)r = 0,
	\end{cases}
\end{align}
for some $x\in(0,\infty)$. Under mild assumptions, \eqref{eq:env_lines} is solved by a single point $\eta(x) \in\Theta'$ for each $x\in (0,\infty)$ and their totality $\eta = (\eta(x))_{x\in(0,\infty)}$ forms a continuously differentiable curve in $\Theta'$, given by
\begin{align*}
	\eta(x) = \left(\frac{W(c,a)(x)}{W(b,c)(x)},\frac{W(a,b)(x)}{W(b,c)(x)}\right).
\end{align*}
We denote with $\ell_0^+$ ($\ell_0^-, \ell_\infty^+, \ell_\infty^-$) the set of points $(\widetilde{\beta}_1, r)\in\Theta'$, such that the corresponding curve is initially increasing (initially decreasing, terminally increasing, terminally decreasing). $\ell_0$ and $\ell_\infty$ denote their boundaries.
As summarized in \cite{KRS26}, the number of local extrema can be easily calculated:
\begin{itemize}
	\item The boundaries of the regions of $\Theta'$ corresponding to different shapes of the curve are given precisely by the envelope $\eta$, together with the lines $\ell_0$ and $\ell_\infty$.
	\item By adjoining parts of $\ell_0$ and $\ell_\infty$ to $\eta$, the envelope can be turned into a closed curve, the \textbf{augmented envelope $\hat{\eta}$}. The number of local extrema of the curve, given the parameter $(\widetilde{\beta}_1, r)\in\Theta'$ can be determined from the winding number of $\hat{\eta}$ around $(\widetilde{\beta}_1, r)$; see \cite[Thm. 4.1]{KRS23}.
\end{itemize}
To apply the winding number formula, we have to check assumptions \textbf{({A}1)} -- \textbf{({A}6)} from \cite{KRS23}, which is done in Appendix~\ref{sec:assumptions}. 
\begin{rem}\label{rem:truncated}
	In order to apply the winding number formula when $\ell_\infty$ is empty, it is necessary to consider the truncated family $(\ell_x)_{x\in(0,T)}$ and its envelope for arbitrarily large $T$. This adjustment has been discussed in \cite{KRS26}, we will thus not go into further detail here.
\end{rem}
Let us recap the definition of the quadrants and rays introduced in \cite{KRS23}:
\begin{equation*}
	\begin{split}
		Q_{n} := {\ell_0^+}\cap {\ell_\infty^+},\qquad Q_{h} &:= \ell_0^+\cap\ell_\infty^-.\\
		Q_{i} := {\ell_0^-}\cap {\ell_\infty^-}, \qquad Q_{d} &:= \ell_0^-\cap\ell_\infty^+ 
	\end{split}
\end{equation*}
and
\begin{equation*}
	\begin{split}
		r_{nh} &= \partial Q_n \cap \partial Q_h, \qquad r_{ih} = \partial Q_i \cap \partial Q_h ,\\
		r_{id} &= \partial Q_i \cap \partial Q_d, \qquad r_{nd} = \partial Q_n \cap \partial Q_d.
	\end{split}
\end{equation*}
A comparison of the coordinates of $M$ and $\eta(0)$ resp. $\eta(\infty)$ (see Appendix~\ref{sec:assumptions} for their definitions), yields the following results.
\begin{lem}
	If $\lambda < \frac{1}{\tau}$, the envelope $\eta_{\mathrm{f}}$ starts on $r_{nd}$ iff
	\begin{align}\label{eq:start_env_HW}
		2\widetilde{\sigma}^2 < \frac{\widetilde{\beta}_2\tau^2 + 2\lambda^2\widetilde{\sigma^2}}{\lambda\left(\frac{1}{\tau} - \lambda\right)}
	\end{align}
	and on $r_{ih}$ iff the inequality is reversed. It asymptotically ends on $r_{nh}$ iff $\lambda < \frac{1}{\tau} < 2\lambda$ and $\widetilde{\beta}_2 > 0$. It asymptotically ends on $r_{id}$ iff $\lambda < \frac{1}{\tau} < 2\lambda$ and $\widetilde{\beta}_2 < 0$ or $2\lambda < \frac{1}{\tau}$. 
\end{lem}
\begin{lem}
	\begin{enumerate}[(a)]
		\item If $\frac{1}{\tau} < \lambda$, the envelope $\eta_{\mathrm{y}}$ starts on $r_{nd}$ iff
		\begin{align}\label{eq:start_env_HW_y_1}
			\eta_{\mathrm{y}, 1}(0) < M_{\mathrm{y}, 1}
		\end{align}
		and on $r_{ih}$ iff the inequality is reversed. It asymptotically ends on $r_{id}$ iff
		\begin{align}\label{eq:end_env_HW_y_1}
			\widetilde{\beta}_2 > 0
		\end{align}
		and on $r_{nh}$ iff the inequality is reversed.
		\item If $\lambda < \frac{1}{\tau}$, the envelope $\eta_{\mathrm{y}}$ starts on $r_{nd}$ iff
		\begin{align}\label{eq:start_env_HW_y_2}
			\eta_{\mathrm{y}, 2}(0) < M_{\mathrm{y}, 2}
		\end{align}
		and on $r_{ih}$ iff the inequality is reversed. It ends on $r_{id}$ iff
		\begin{align}\label{eq:end_env_HW_y_2}
			\eta_{\mathrm{y}, 1}(\infty) > M_{\mathrm{y}, 1}
		\end{align}
		and on $r_{nh}$ iff the inequality is reversed.
	\end{enumerate}
\end{lem}

We proceed with a case by case analysis, similar to \cite{KRS23} and \cite{KRS26}. Defining  
\begin{align*}
	D := \widetilde{\beta}_2\tau^2\left(\frac{1}{\tau} - \lambda\right) - 2\lambda^2\widetilde{\sigma}^2\left(\frac{1}{\tau} - 2\lambda\right).
\end{align*}
and using the results above, as well as from Appendix~\ref{sec:assumptions} and \cite[Lem. B.6]{KRS26}, we find seven different ways of separating $\Theta'$, when considering the forward curve, all of which can be attained by the right choice of $\lambda,\tau$ and $\widetilde{\beta}_2$, thus answering \textbf{Q1} and \textbf{Q2}. Visualizations are provided in \cref{fig:Nelson_Siegel_indep_f}.
\begin{enumerate}[1.]
	\item $\frac{1}{\tau} < \lambda$ and $\widetilde{\beta}_2 < 0$: $\ell_\infty^{\mathrm{f}}$ is empty. The envelope initially moves into $\ell_0^{\mathrm{f}\pp}$ and has no cusps, thus no intersections with $\ell_0^{\mathrm{f}}$. The forward curve is terminally increasing. The shapes \texttt{normal, dipped} and \texttt{hd} are attainable.
	\item $\frac{1}{\tau} < \lambda$, $\widetilde{\beta}_2 > 0$ and $D < 0$: $\ell_\infty^{\mathrm{f}}$ is empty. The envelope initially moves into $\ell_0^{\mathrm{f}\mm}$ and has no cusps, thus no intersections with $\ell_0^{\mathrm{f}}$. The forward curve is terminally decreasing. The shapes \texttt{inverse, humped} and \texttt{dh} are attainable.
	\item $\frac{1}{\tau} < \lambda$, $\widetilde{\beta}_2 > 0$ and $D > 0$: $\ell_\infty^{\mathrm{f}}$ is empty. The envelope initially moves into $\ell_0^{\mathrm{f}\pp}$ and has exactly one cusps, after which it intersects with $\ell_0^{\mathrm{f}}$ moving into $\ell_0^{\mathrm{f}\mm}$. The forward curve is terminally decreasing. The shapes \texttt{inverse, humped, dh} and \texttt{hdh} are attainable.
	\item $\lambda < \frac{1}{\tau} \le 2\lambda$ and $\widetilde{\beta}_2 > 0$: The envelope starts on $r_{nd}$ moving into $Q_{n}$, has no cusps, thus no intersections with $\ell_0^{\mathrm{f}}$ or $\ell_\infty^{\mathrm{f}}$ and asymptotically ends on $r_{nh}$. The shapes \texttt{normal, inverse, humped, dipped} and \texttt{hd} are attainable.
	\item $\{\lambda < \frac{1}{\tau} \le 2\lambda$, $\widetilde{\beta}_2 < 0$ and $D < 0\}$ or $\{2\lambda < \frac{1}{\tau}$ and $\{\widetilde{\beta}_2 < 0$ or $\widetilde{\beta}_2 > 0$ and $D < 0\}\}$: The envelope starts on $r_{ih}$ moving into $Q_{i}$, has no cusps, thus no intersections with $\ell_0^{\mathrm{f}}$ or $\ell_\infty^{\mathrm{f}}$ and asymptotically ends on $r_{id}$. The shapes \texttt{normal, inverse, humped, dipped} and \texttt{dh} are attainable.
	\item $\{\lambda < \frac{1}{\tau} < 2\lambda$ and $\widetilde{\beta}_2 < 0$ or $2\lambda < \frac{1}{\tau}$ and $\widetilde{\beta}_2 > 0\}$ and $D > 0$, but not \eqref{eq:start_env_HW}: The envelope starts on $r_{ih}$, moving into $Q_{h}$, has a cusp there, after which it intersects with $\ell_0$ and asymptotically ends on $\ell_\infty^{\mathrm{f}}$ coming from $Q_{i}$. The shapes \texttt{normal, inverse, humped, dipped, dh} and \texttt{hdh} are attainable.
	\item $\{\lambda < \frac{1}{\tau} < 2\lambda$ and $\widetilde{\beta}_2 < 0$ or $2\lambda < \frac{1}{\tau}$ and $\widetilde{\beta}_2 > 0\}$ and $D > 0$ and \eqref{eq:start_env_HW} holds: The envelope starts on $r_{nd}$, moving into $Q_{n}$, crosses $\ell_\infty^{\mathrm{f}}$ into $Q_{h}$, has a cusp there, after which it intersects with $\ell_0$ and asymptotically ends on $\ell_\infty^{\mathrm{f}}$ coming from $Q_{i}$. The shapes \texttt{normal, inverse, humped, dipped, hd, dh} and \texttt{hdh} are attainable.
\end{enumerate}
\begin{figure}
	\centering
	\begin{subfigure}{0.49\textwidth}
		\includegraphics[width=\textwidth]{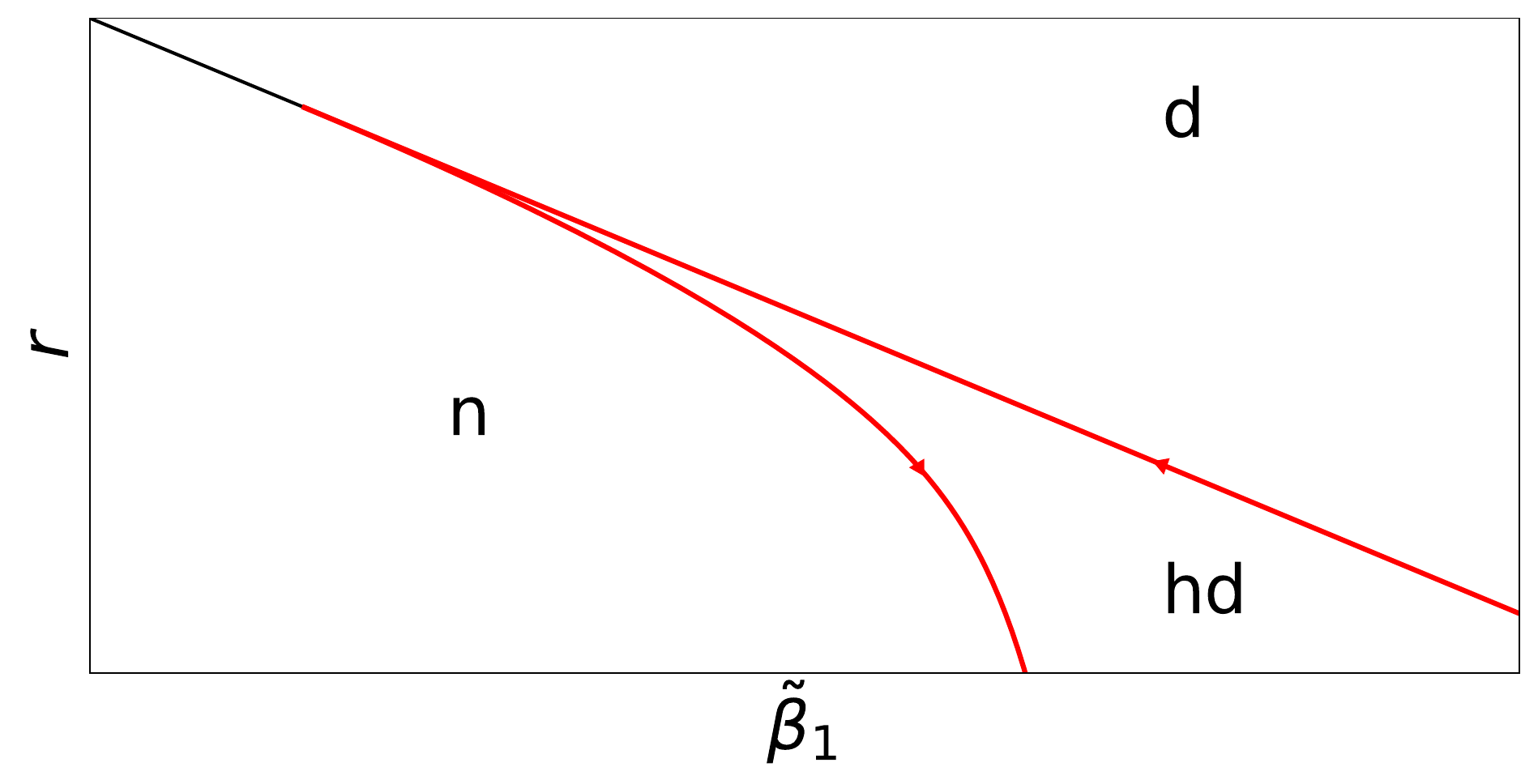}
		\caption{Regime 1.}
	\end{subfigure}
	\hfill
	\begin{subfigure}{0.49\textwidth}
		\includegraphics[width=\textwidth]{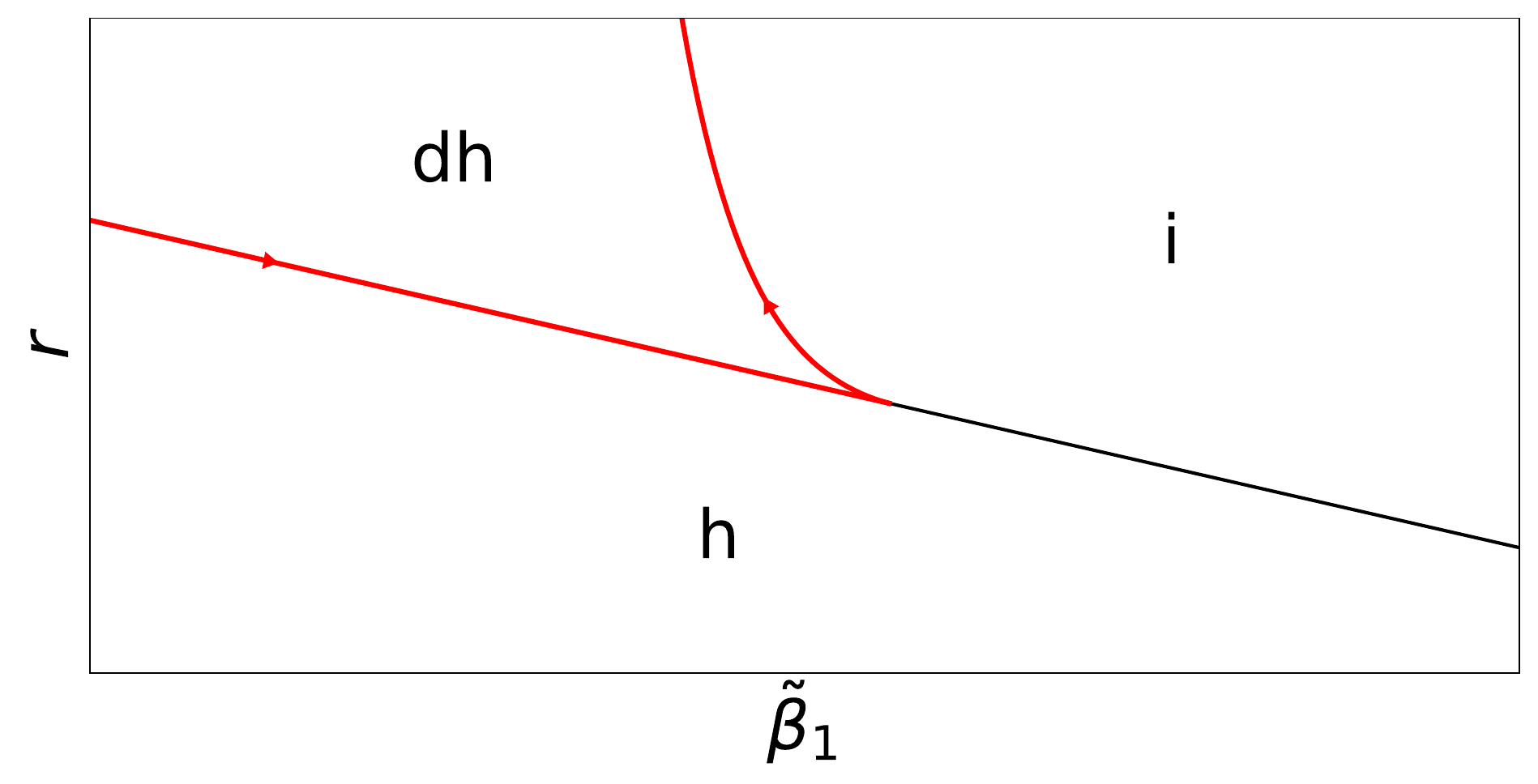}
		\caption{Regime 2.}
	\end{subfigure}
	
	\begin{subfigure}{0.49\textwidth}
		\includegraphics[width=\textwidth]{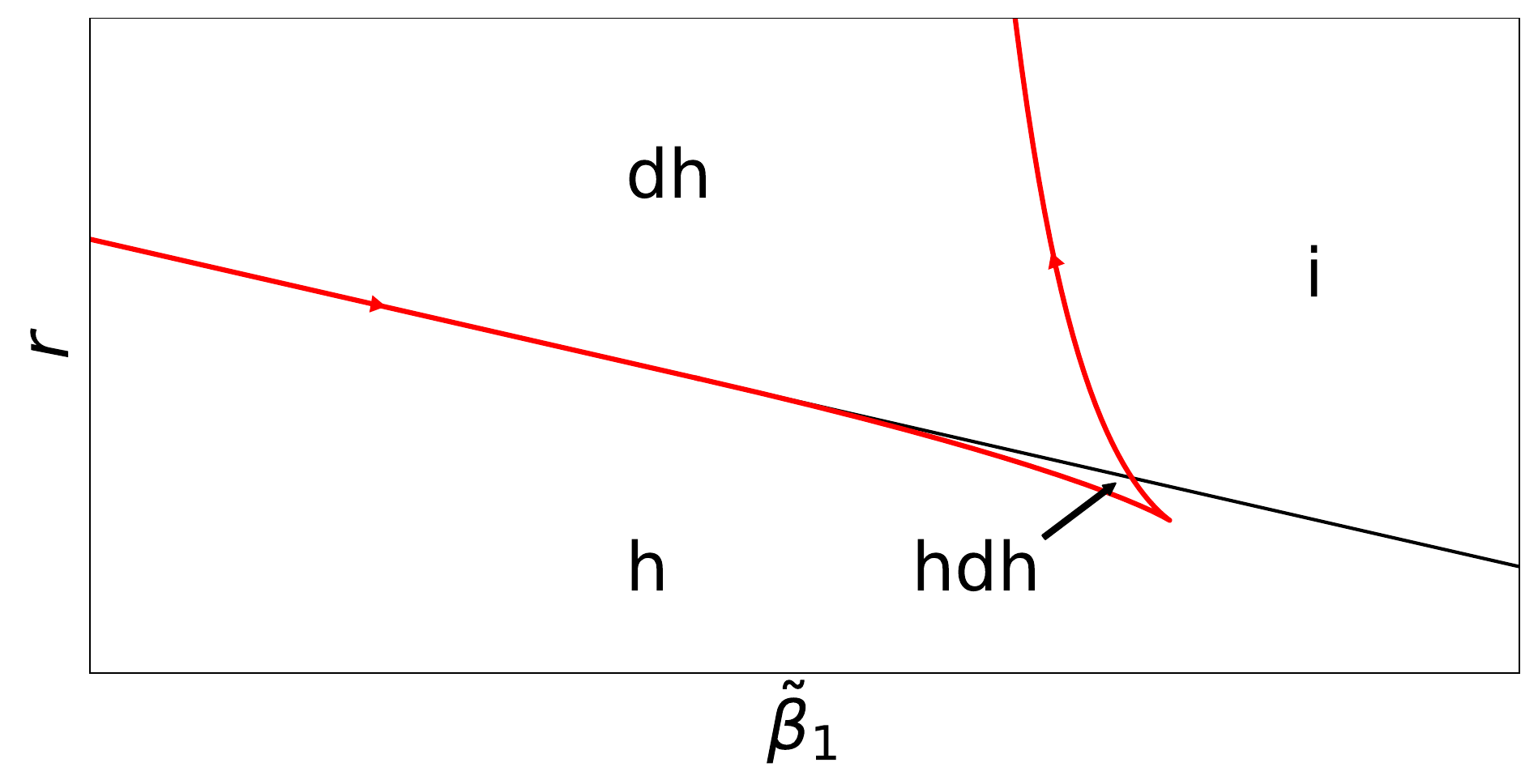}
		\caption{Regime 3.}
	\end{subfigure}
	\hfill
	\begin{subfigure}{0.49\textwidth}
		\includegraphics[width=\textwidth]{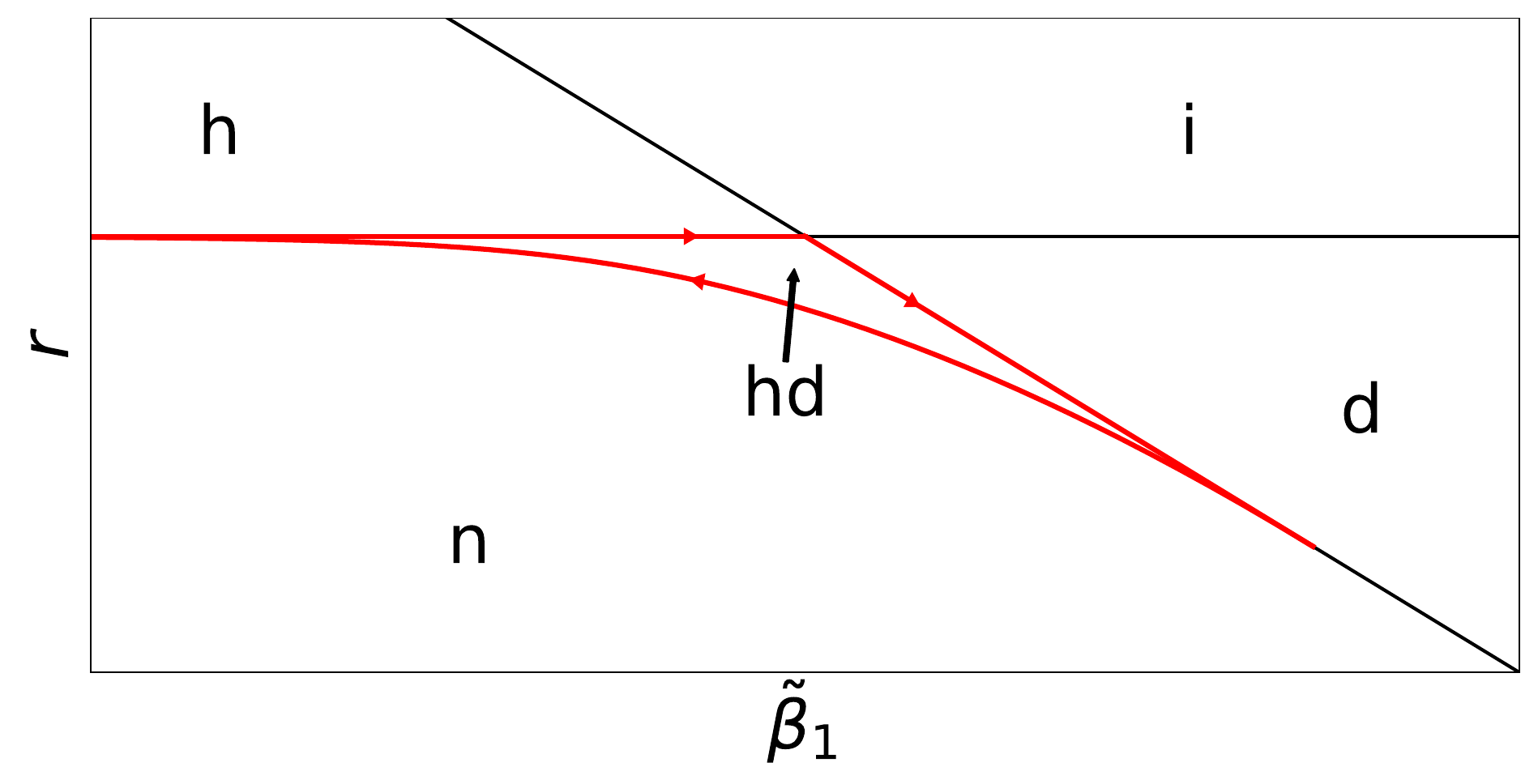}
		\caption{Regime 4.}
	\end{subfigure}
	
	\begin{subfigure}{0.49\textwidth}
		\includegraphics[width=\textwidth]{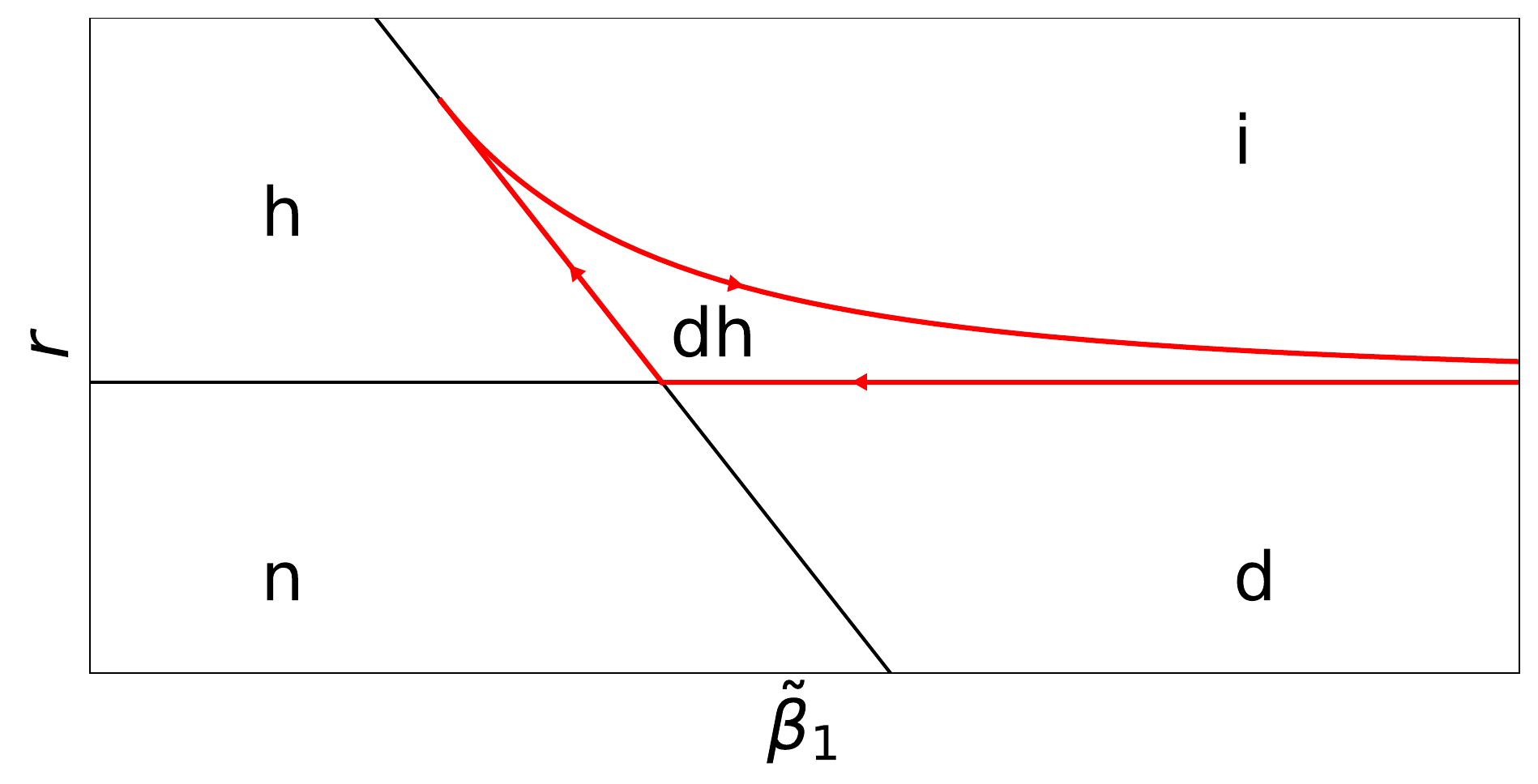}
		\caption{Regime 5.}
	\end{subfigure}
	\hfill
	\begin{subfigure}{0.49\textwidth}
		\includegraphics[width=\textwidth]{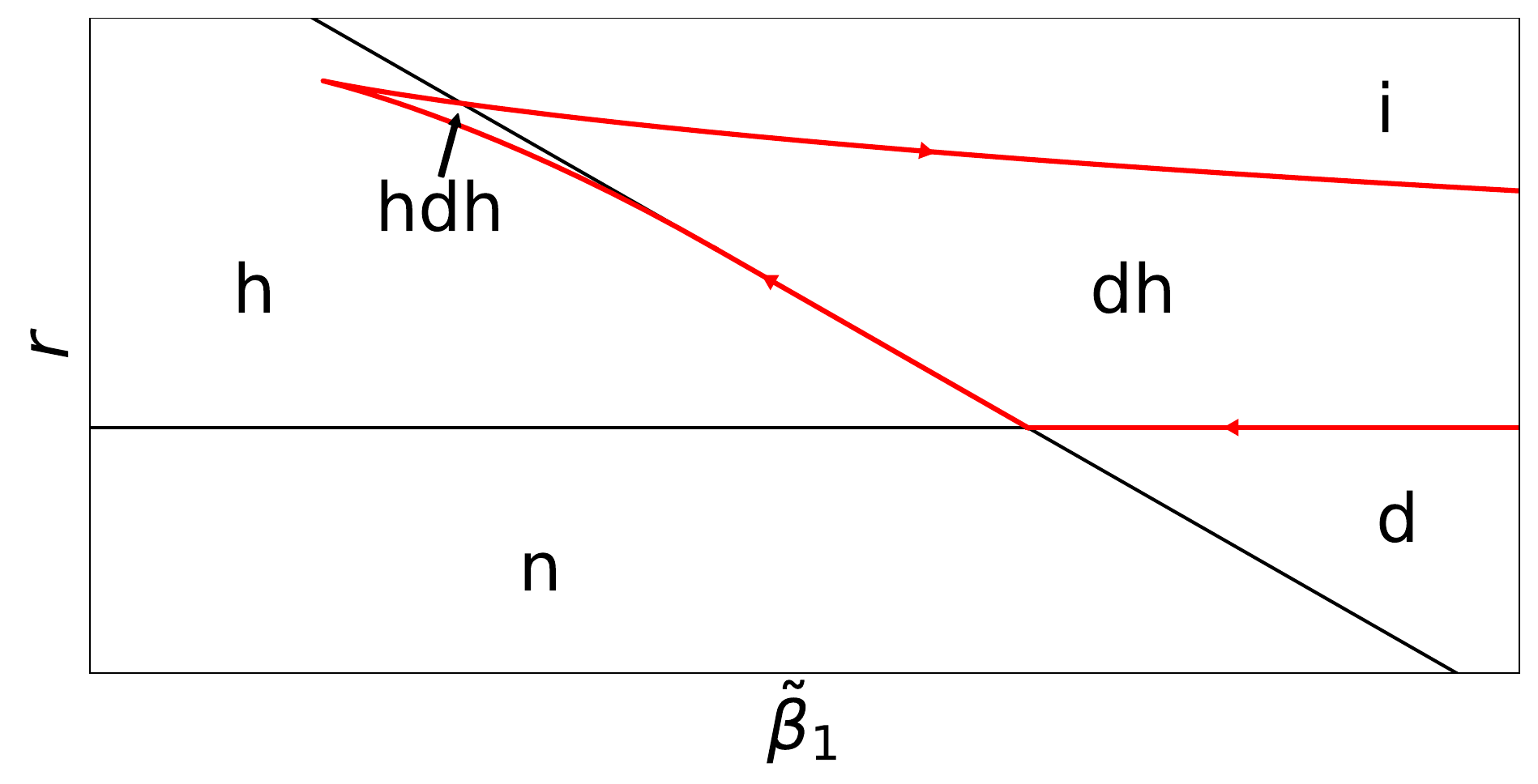}
		\caption{Regime 6.}
	\end{subfigure}
	
	\begin{subfigure}{0.49\textwidth}
		\includegraphics[width=\textwidth]{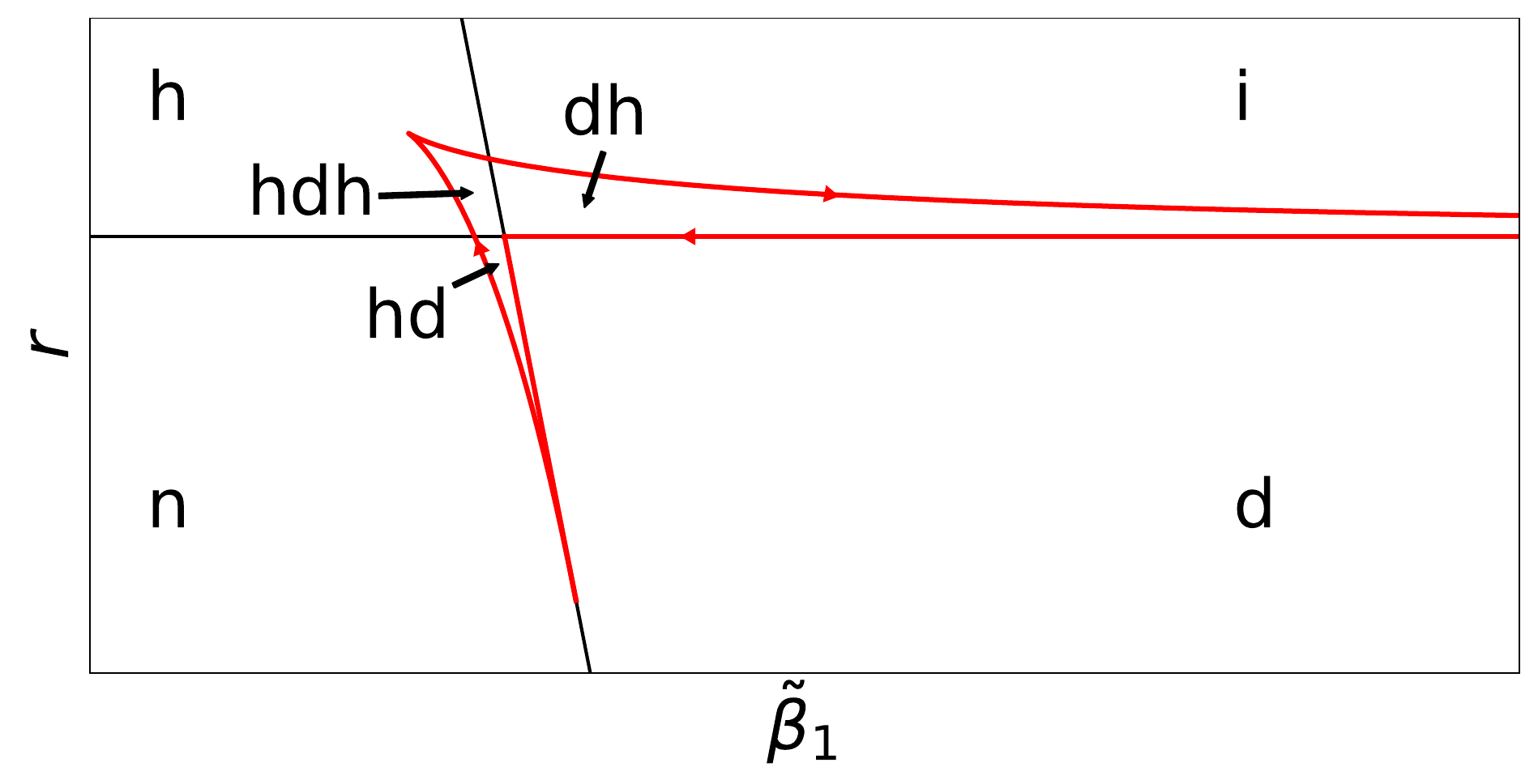}
		\caption{Regime 7.}
	\end{subfigure}

	\caption{Shapes of the forward curve for Nelson--Siegel--parameterized initial curves in seven regimes. The augmented envelope is colored in red.}
	\label{fig:Nelson_Siegel_indep_f}
\end{figure}

To determine the decomposition of $\Theta'$ for the yield curve, we remark, that the terminal sign of the Wronskian $W(a_{\mathrm{y}}, b_{\mathrm{y}}, c_{\mathrm{y}})$ is given by the sign of $\widetilde{\beta}_2$, if $\frac{1}{\tau}\le 2\lambda$, and $\mm$, if $2\lambda < \frac{1}{\tau}$. The Wronskian shares its initial sign with the initial sign of $W(a_{\mathrm{f}}, b_{\mathrm{f}}, c_{\mathrm{f}})$ \cite[Lem. A.2 f.]{KRS26} and thus $W(a_{\mathrm{y}}, b_{\mathrm{y}}, c_{\mathrm{y}})$ has a cusp iff $W(a_{\mathrm{f}}, b_{\mathrm{f}}, c_{\mathrm{f}})$ has a cusp.

We find nine different ways of separating $\Theta'$, when considering the yield curve, all of which can be attained by the right choice of $\lambda,\tau$ and $\widetilde{\beta}_2$, thus answering \textbf{Q1} and \textbf{Q2}. Visualizations are provided in \cref{fig:Nelson_Siegel_indep_y}.
\begin{enumerate}[1.]
	\item $\lambda < \frac{1}{\tau}\le 2\lambda$ and $\widetilde{\beta}_2 > 0$: The envelope starts on $r_{nd}$, moving into $Q_n$, where it has no cusps and thus does not intersect with $\ell_0$ or $\ell_\infty^{\mathrm{y}}$. It ends on $r_{nh}$. The shapes \texttt{normal, inverse, humped, dipped} and \texttt{hd} are attainable.
	\item $\frac{1}{\tau} < \lambda$ and $\widetilde{\beta}_2 < 0$: The envelope starts on $r_{nd}$, moving into $Q_n$, where it has no cusps and thus does not intersect with $\ell_0$ or $\ell_\infty^{\mathrm{y}}$. It asymptotically ends on $r_{nh}$. The shapes \texttt{normal, inverse, humped, dipped} and \texttt{hd} are attainable.
	\item $\{\lambda < \frac{1}{\tau} \le 2\lambda, \widetilde{\beta}_2 < 0$ and $D < 0\}$ or $\{2\lambda < \frac{1}{\tau}$ and $\{\widetilde{\beta}_2 < 0$ or $\widetilde{\beta_2} > 0$ and $D < 0\}\}$: The envelope starts on $r_{ih}$, moving into $Q_i$, where it has no cusps and thus does not intersect with $\ell_0$ or $\ell_\infty^{\mathrm{y}}$. It ends on $r_{id}$. The shapes \texttt{normal, inverse, humped, dipped} and \texttt{dh} are attainable.
	\item $\frac{1}{\tau} < \lambda, \widetilde{\beta}_2 > 0$ and $D < 0$: The envelope starts on $r_{ih}$, moving into $Q_i$, where it has no cusps and thus does not intersect with $\ell_0$ or $\ell_\infty^{\mathrm{y}}$. It asymptotically ends on $r_{id}$. The shapes \texttt{normal, inverse, humped, dipped} and \texttt{dh} are attainable.
	\item $\{\lambda < \frac{1}{\tau} \le 2\lambda$ and $ \widetilde{\beta}_2 < 0$ or  $2\lambda < \frac{1}{\tau}$ and $\widetilde{\beta}_2 > 0\}$ and $D > 0$ and \eqref{eq:end_env_HW_y_2}, but not \eqref{eq:start_env_HW_y_2}: The envelope starts on $r_{ih}$, moving into $Q_h$, where it has a cusp, crosses $\ell_0$ into $Q_i$ and ends on $r_{id}$. The shapes \texttt{normal, inverse, humped, dipped, dh} and \texttt{hdh} are attainable.
	\item $\frac{1}{\tau} < \lambda$ and $\widetilde{\beta}_2 > 0$ and $D > 0$, but not \eqref{eq:start_env_HW_y_1}: The envelope starts on $r_{ih}$, moving into $Q_h$, where it has a cusp, crosses $\ell_0$ into $Q_i$ and asymptotically ends on $r_{id}$. The shapes \texttt{normal, inverse, humped, dipped, dh} and \texttt{hdh} are attainable.
	\item $\{\lambda < \frac{1}{\tau} \le 2\lambda$ and $ \widetilde{\beta}_2 < 0$ or  $2\lambda < \frac{1}{\tau}$ and $\widetilde{\beta}_2 > 0\}$ and $D > 0$ and \eqref{eq:start_env_HW_y_2} and \eqref{eq:end_env_HW_y_2}: The envelope starts on $r_{nd}$, moving into $Q_n$, crosses $\ell_\infty^{\mathrm{y}}$ into $Q_h$, where it has a cusp, crosses $\ell_0$ into $Q_i$ and ends on $r_{id}$. The shapes \texttt{normal, inverse, humped, dipped, hd, dh} and \texttt{hdh} are attainable.
	\item $\frac{1}{\tau} < \lambda$ and $\widetilde{\beta}_2 > 0$ and $D > 0$ and \eqref{eq:start_env_HW_y_1}: The envelope starts on $r_{nd}$, moving into $Q_n$, crosses $\ell_\infty^{\mathrm{y}}$ into $Q_h$, where it has a cusp, crosses $\ell_0$ into $Q_i$ and asymptotically ends on $r_{id}$. The shapes \texttt{normal, inverse, humped, dipped, hd, dh} and \texttt{hdh} are attainable.
	\item $\{\lambda < \frac{1}{\tau} \le 2\lambda$ and $ \widetilde{\beta}_2 < 0$ or  $2\lambda < \frac{1}{\tau}$ and $\widetilde{\beta}_2 > 0\}$ and $D > 0$ and \eqref{eq:start_env_HW_y_2}, but not \eqref{eq:end_env_HW_y_2}: The envelope starts on $r_{nd}$, moving into $Q_n$, crosses $\ell_\infty^{\mathrm{y}}$ into $Q_h$, where it has a cusp and ends on $r_{nh}$. The shapes \texttt{normal, inverse, humped, dipped, hd} and \texttt{hdh} are attainable.
\end{enumerate}
\begin{figure}
	\centering
	\begin{subfigure}{0.49\textwidth}
		\includegraphics[width=\textwidth]{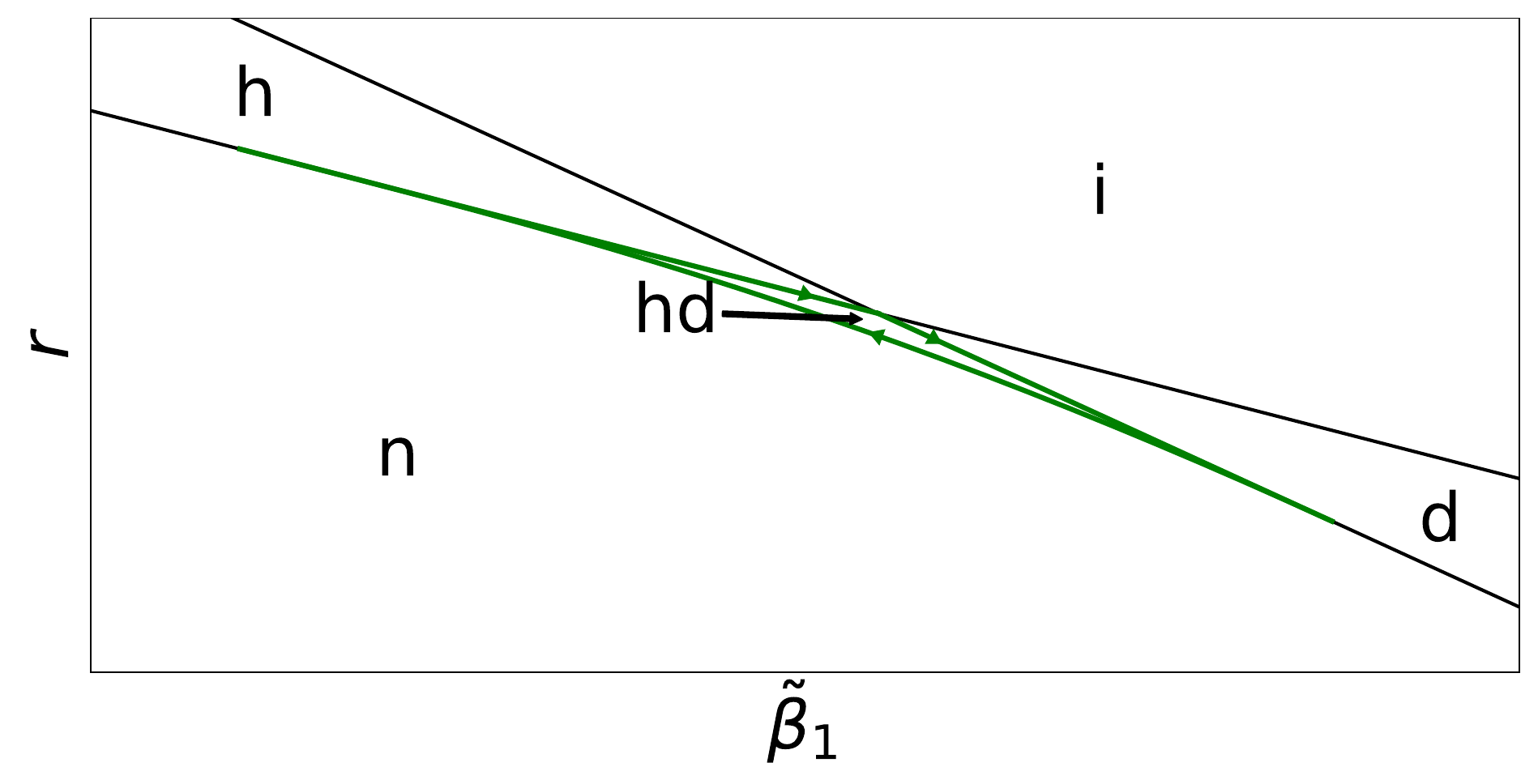}
		\caption{Regime 1.}
	\end{subfigure}
	\hfill
	\begin{subfigure}{0.49\textwidth}
		\includegraphics[width=\textwidth]{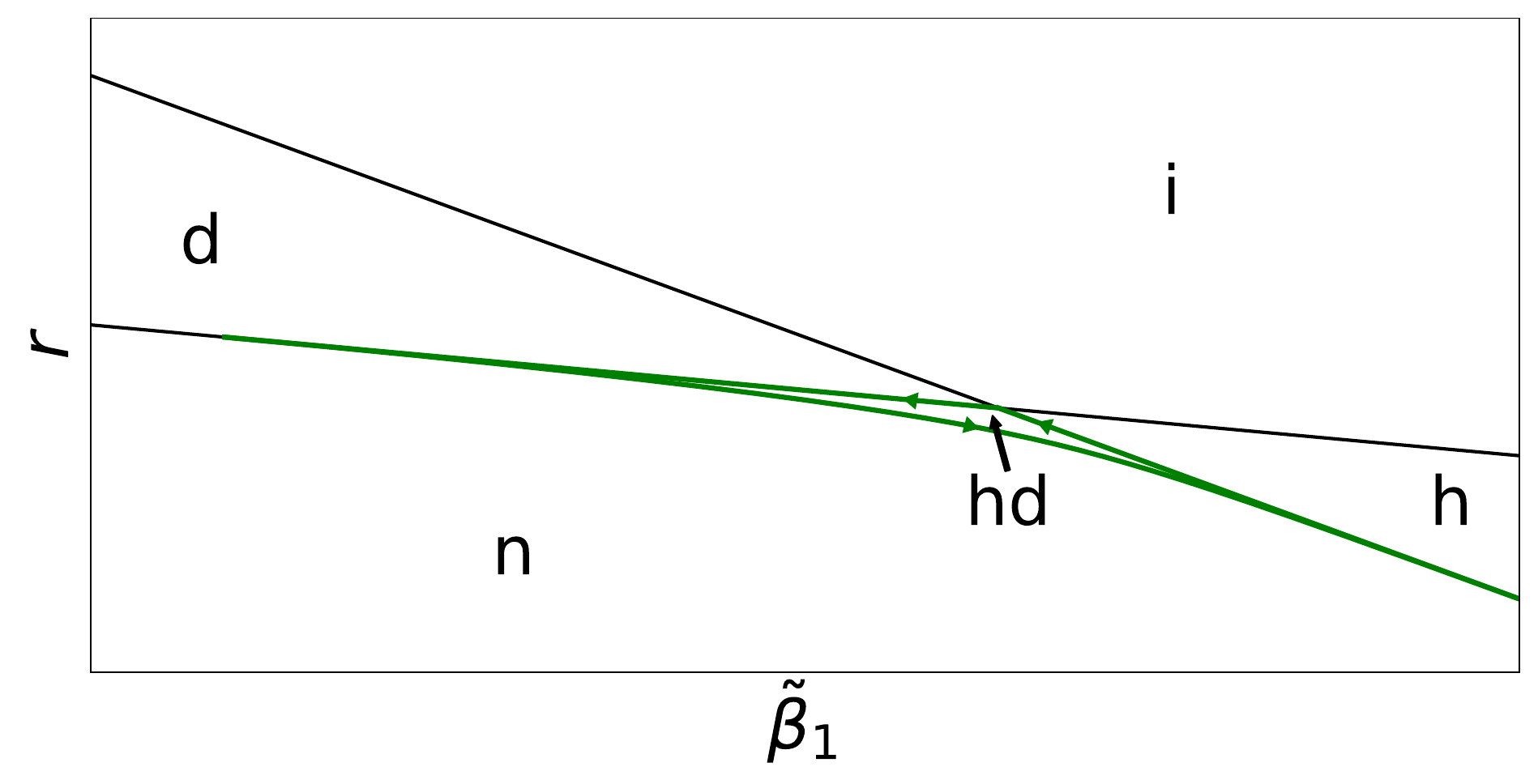}
		\caption{Regime 2.}
	\end{subfigure}
	
	\begin{subfigure}{0.49\textwidth}
		\includegraphics[width=\textwidth]{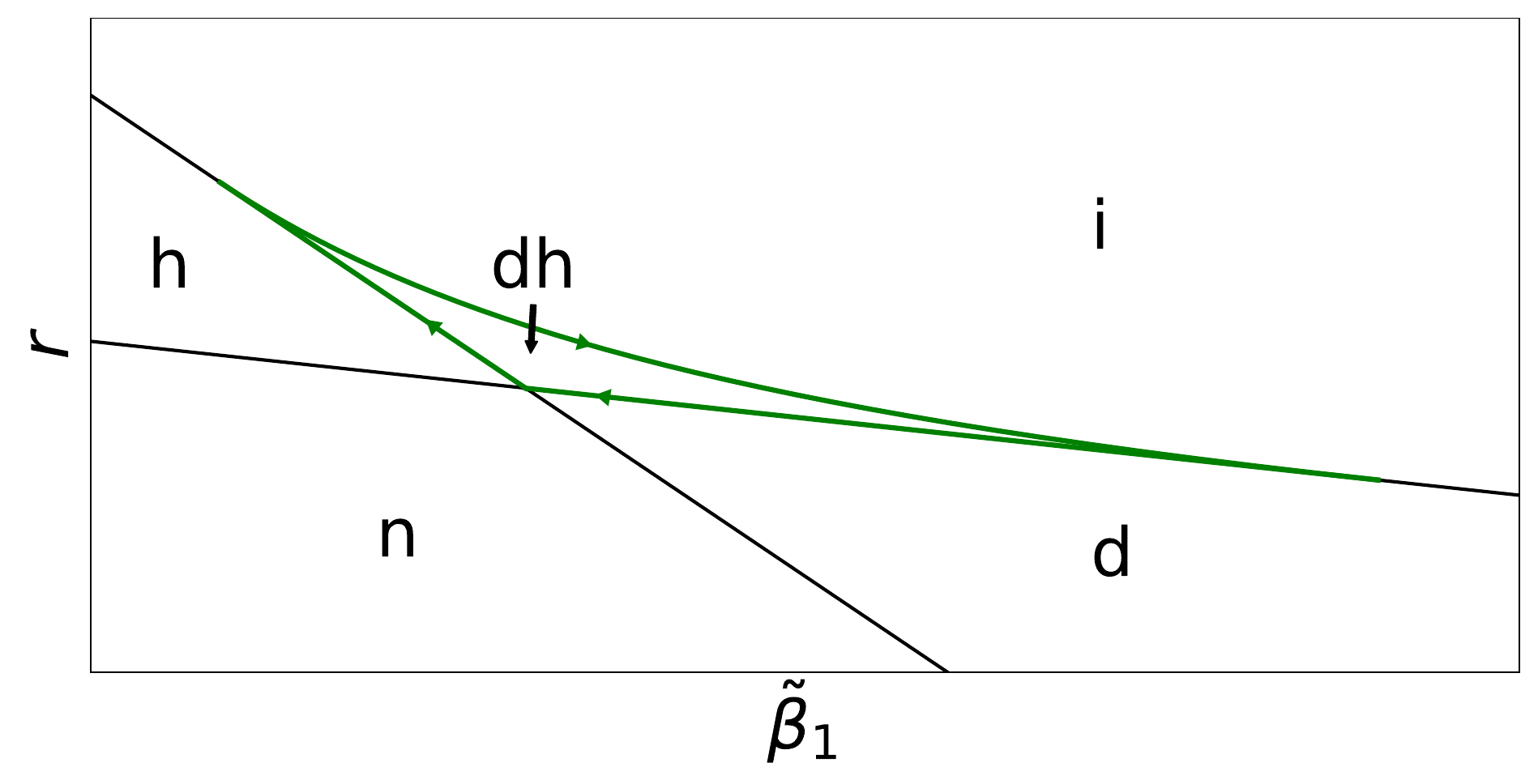}
		\caption{Regime 3.}
	\end{subfigure}
	\hfill
	\begin{subfigure}{0.49\textwidth}
		\includegraphics[width=\textwidth]{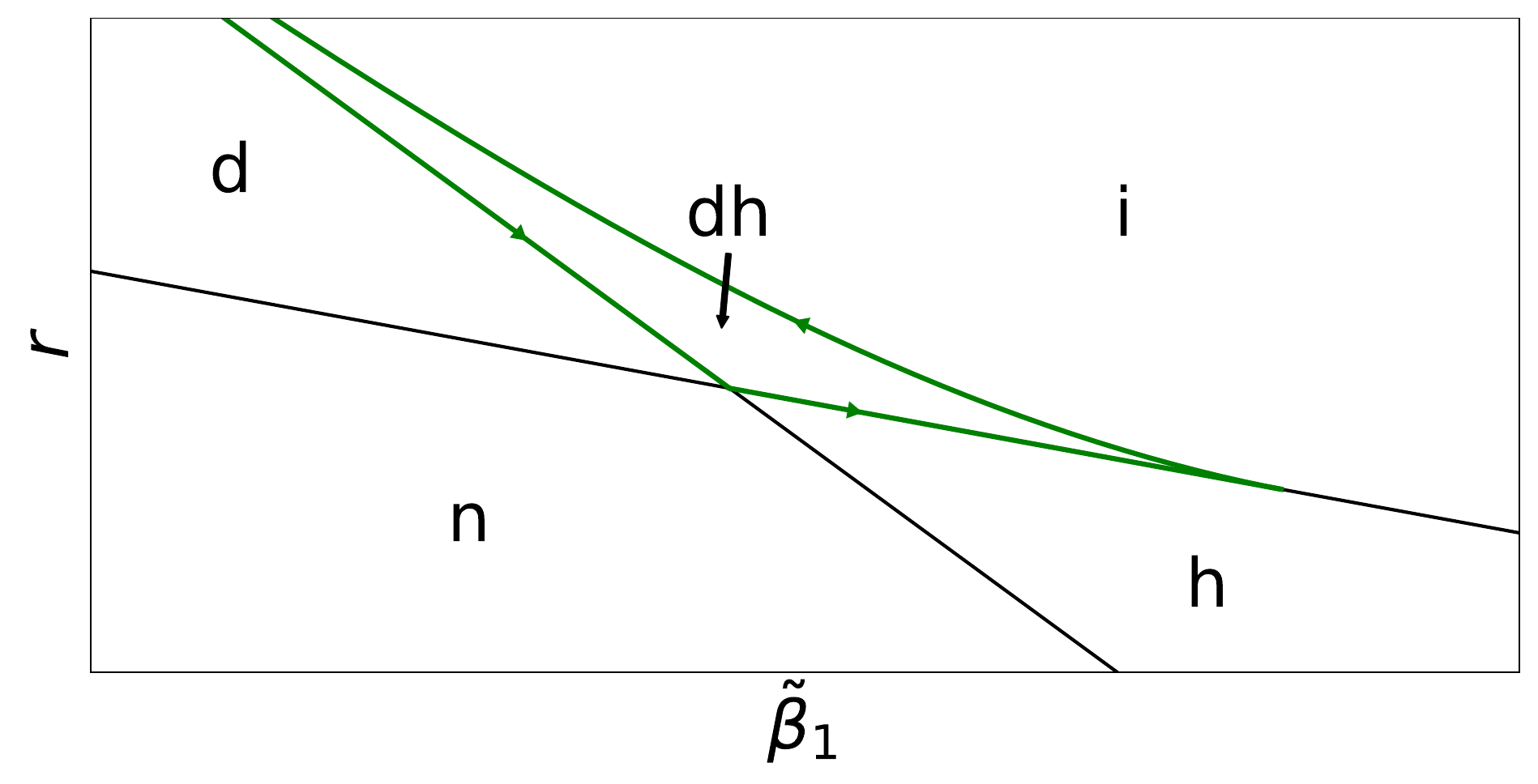}
		\caption{Regime 4.}
	\end{subfigure}
	
	\begin{subfigure}{0.49\textwidth}
		\includegraphics[width=\textwidth]{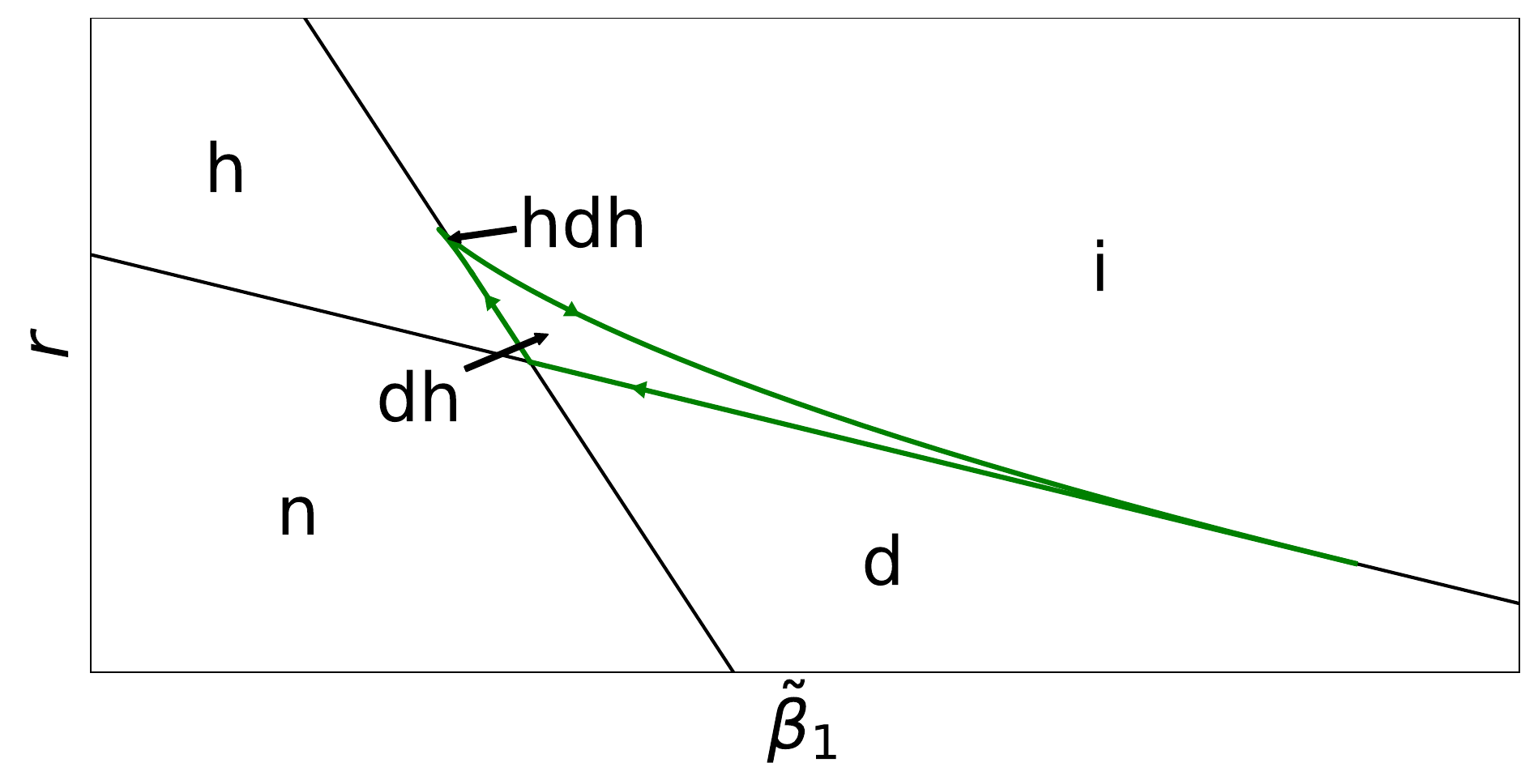}
		\caption{Regime 5.}
	\end{subfigure}
	\hfill
	\begin{subfigure}{0.49\textwidth}
		\includegraphics[width=\textwidth]{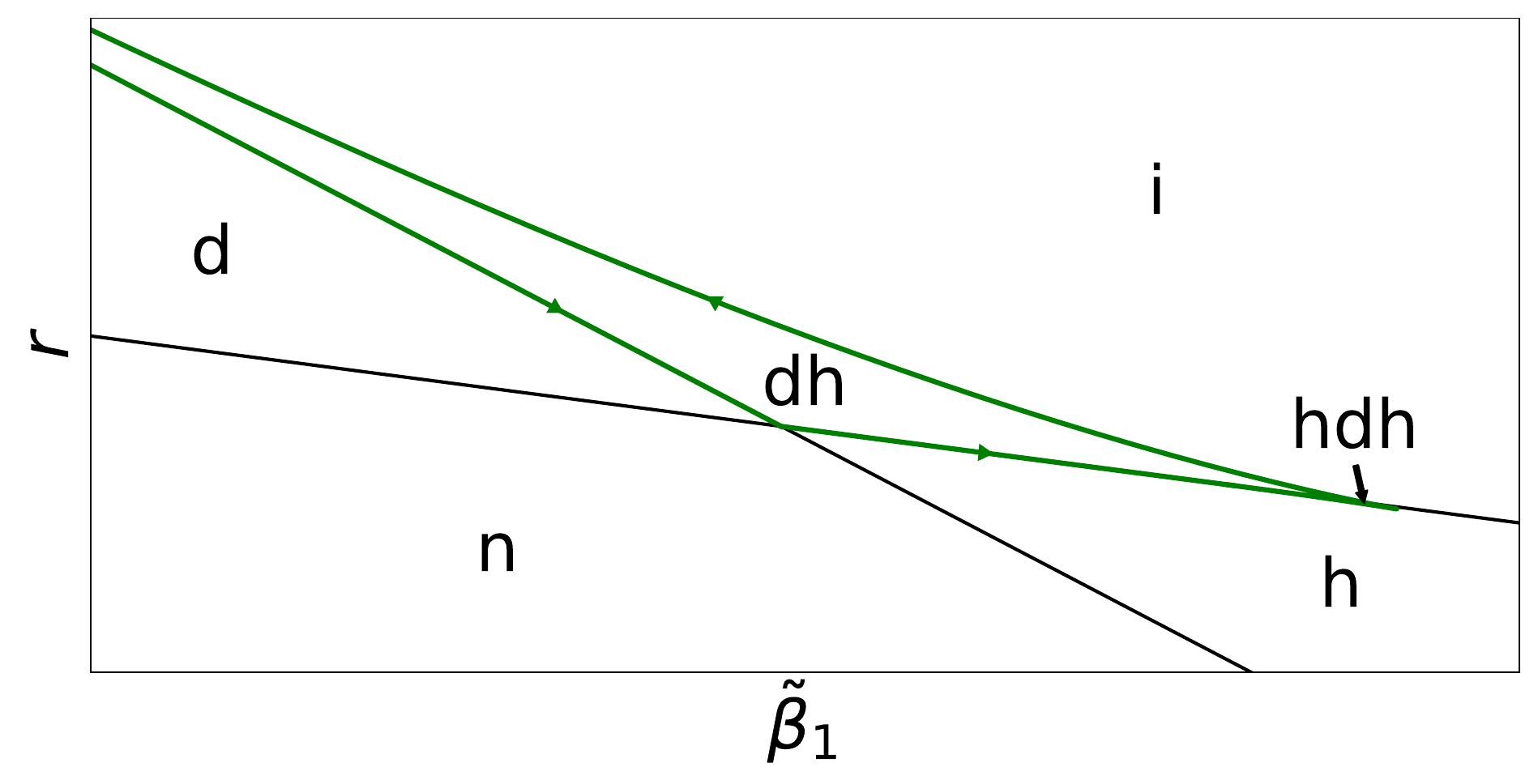}
		\caption{Regime 6.}
	\end{subfigure}
	
	\begin{subfigure}{0.49\textwidth}
		\includegraphics[width=\textwidth]{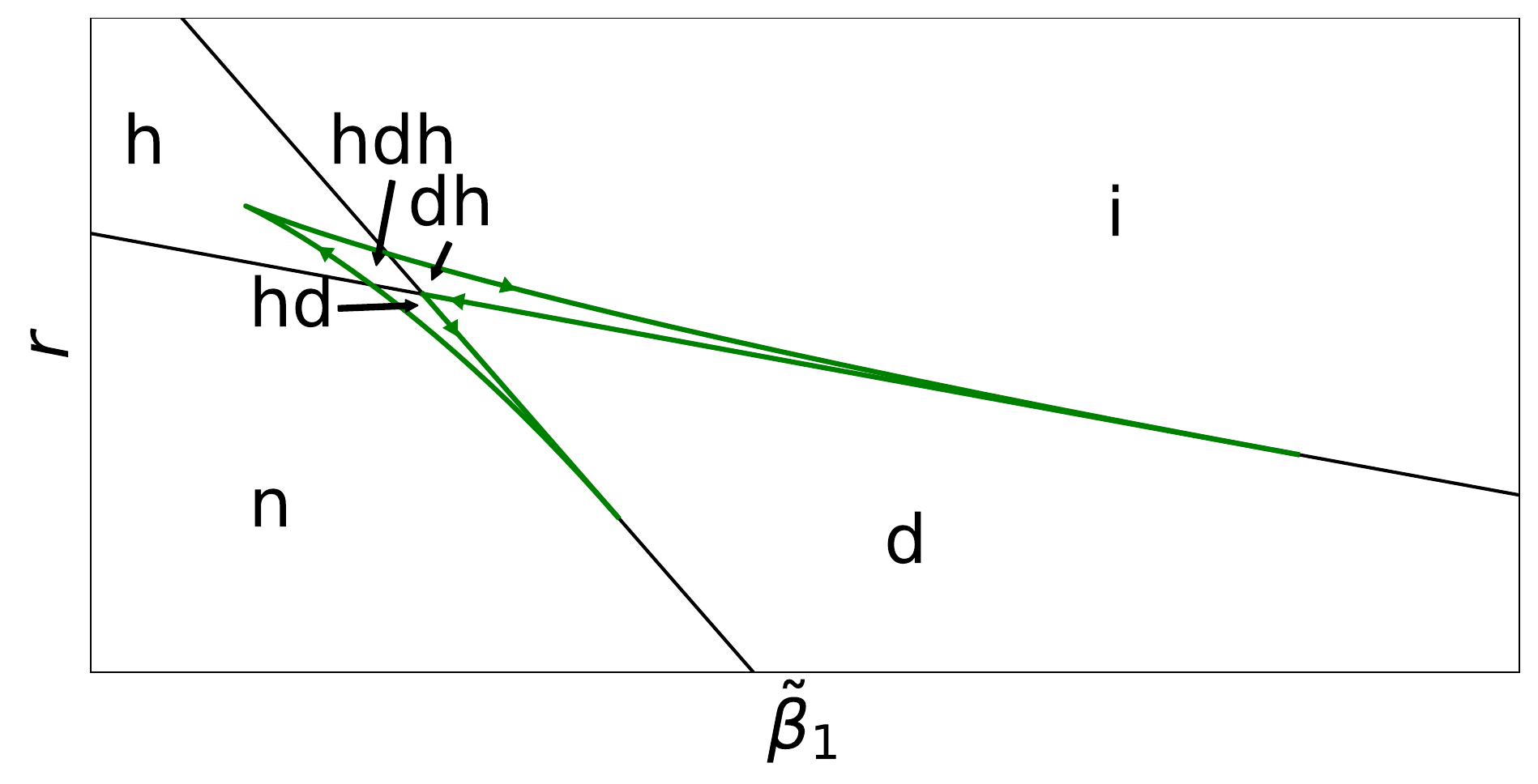}
		\caption{Regime 7.}
	\end{subfigure}
	\hfill
	\begin{subfigure}{0.49\textwidth}
		\includegraphics[width=\textwidth]{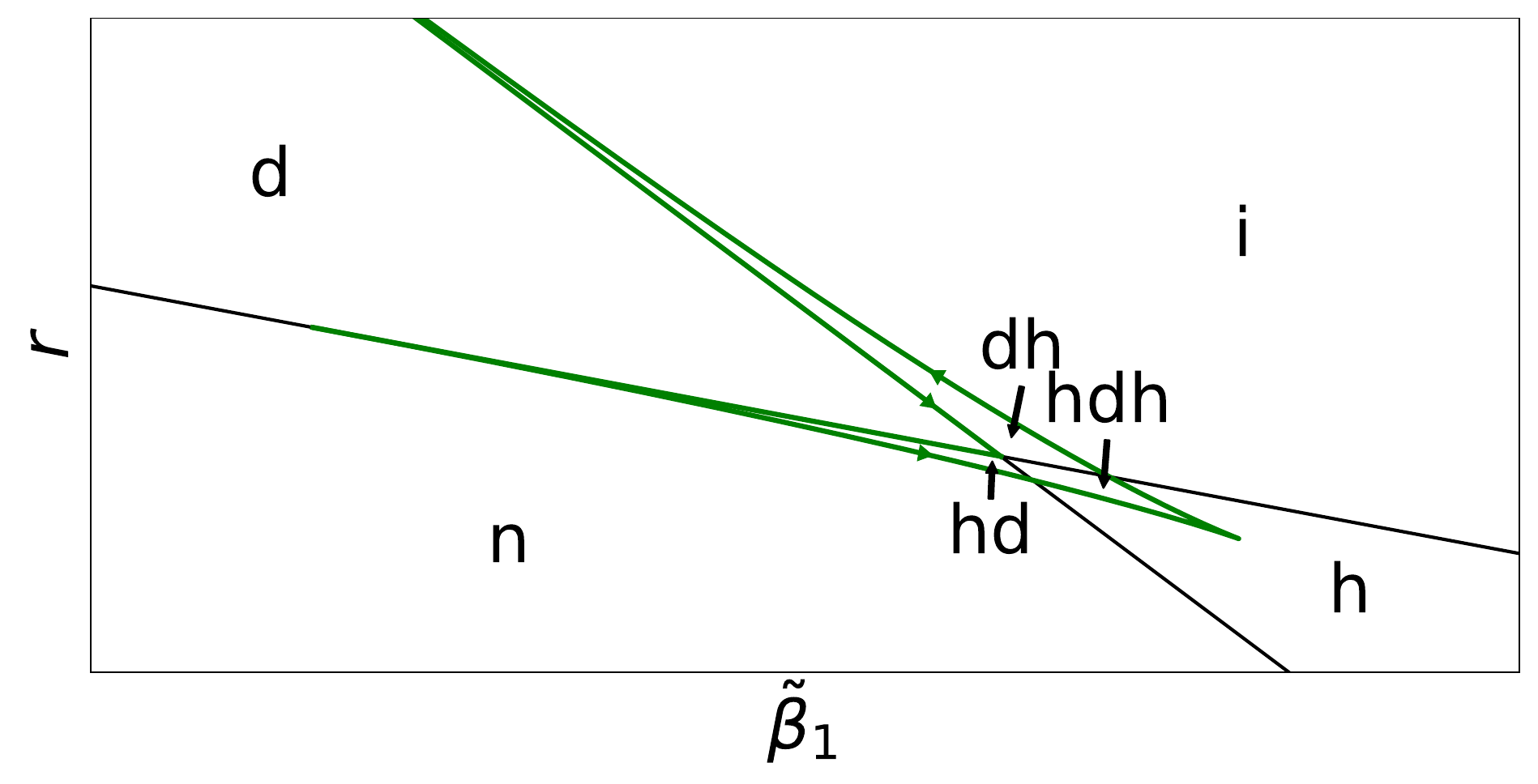}
		\caption{Regime 8.}
	\end{subfigure}
	
	\begin{subfigure}{0.49\textwidth}
		\includegraphics[width=\textwidth]{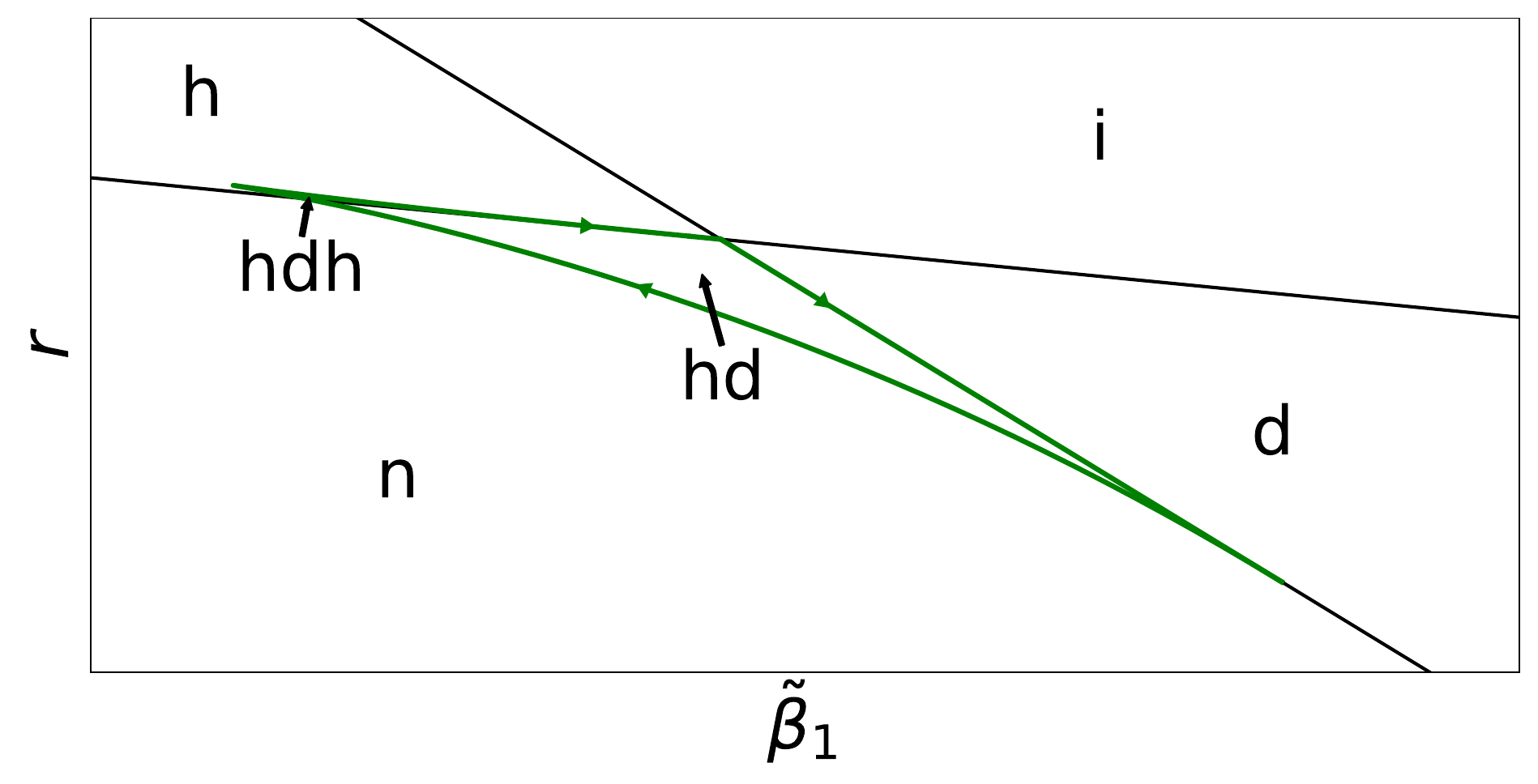}
		\caption{Regime 9.}
	\end{subfigure}
	
	\caption{Shapes of the yield curve for Nelson--Siegel--parameterized initial curves in nine regimes. The augmented envelope is colored in green.}
	\label{fig:Nelson_Siegel_indep_y}
\end{figure}

\begin{rem}
	The shape \texttt{dhd} is not attainable with Nelson--Siegel--parameter\-ized initial curves. To explicitly determine parameters for the shape \texttt{hdh}, one can thus use the approach described on \cite[p. 20]{KS66}.
\end{rem}

\section{A note on Bliss and Svensson}

We have derived upper bounds for the number of local extrema for Bliss--parameterized and Svensson--parameterized initial curves in Theorem~\ref{thm:bounds}. Studying \eqref{eq:for_der} and \eqref{eq:yield_der}, it is obvious that, by choosing $t$ small and $r$, such that
\begin{align*}
	\left(\frac{\sigma^2}{2\lambda^2}\left(1-{\mathrm{e}}^{-2\lambda t}\right) - f_0(t) + r\right) \approx 0,
\end{align*}
at least the shape of the initial curve must be attainable for small values of $t$. This is a partial answer to \textbf{Q1}.
An analysis as in the case of Nelson--Siegel--parameterized initial curves can be done, but leads to problems already observed in \cite{KRS23} with the scale-proximal case of the Vasicek model with $\rho < 0$, where the exact location of the cusp points of the envelope could not be determined. Still, it is clear that parameter space and state space are separated into regions only by  those points where initial or terminal behavior changes or which are part of the envelope, thus giving us an answer to \textbf{Q2}. Contrary to the classification problem and the segmentation problem, the asymptotic behavior of the model, and thus \textbf{Q3}, can be studied quite well; this is done in the next section.

\section{Asymptotic behavior}\label{sec:asymptotics}

Here we study how the model behaves as $t$ tends to $\infty$. We will show that, for the yield curve, asymptotically, most shapes vanish and only a few remain. Additionally, the asymptotic frequencies of the remaining shapes do not depend on the choice of initial curve.

\begin{lem}\label{lem:asymptotic_att_shapes}
	If one of the conditions
	\begin{enumerate}
		\item $\lambda > \frac{1}{\tau_2}$ and $\frac{1}{\tau_1} > \frac{1}{\tau_2}$ and $\beta_3 < 0$, or
		\item $\lambda > \frac{1}{\tau_1}$ and $\left\{\frac{1}{\tau_1} > \frac{1}{\tau_2}\text{ or }\beta_3 = 0\right\}$ and $\beta_2 < 0$, or
		\item $\lambda > \frac{1}{\tau_1}$ and $\left\{\frac{1}{\tau_1} > \frac{1}{\tau_2}\text{ or }\beta_3 = 0\right\}$ and $\beta_2 = 0$ and $\beta_1 < 0,$
	\end{enumerate}
	is satisfied, there is a $K > 0$, such that for every $t > K$ one can find \text{(half-)}open intervals $I_{\texttt{i}}(t), I_{\texttt{h}}(t), I_{\texttt{hd}}(t), I_{\texttt{n}}(t)$ with $I_{\texttt{i}}(t)\cup I_{\texttt{h}}(t)\cup I_{\texttt{hd}}(t)\cup I_{\texttt{n}}(t) = \RR$, such that the curve $x\mapsto y(x,t,r)$ is of shape $\texttt{s}\in\{\texttt{i}, \texttt{h}, \texttt{hd}, \texttt{n}\}$ when $r\in I_{\texttt{s}}(t)$.
	
	If none of the conditions is satisfied, there is a $K > 0$, such that for every $t > K$ one can find (half-)open intervals $I_{\texttt{i}}(t), I_{\texttt{h}}(t), I_{\texttt{n}}(t)$ with $I_{\texttt{i}}(t)\cup I_{\texttt{h}}(t)\cup I_{\texttt{n}}(t) = \RR$, such that the curve $x\mapsto y(x,t,r)$ is of shape $\texttt{s}\in\{\texttt{i}, \texttt{h}, \texttt{n}\}$ when $r\in I_{\texttt{s}}(t)$.
\end{lem}
\begin{thm}\label{thm:asymptotic_rel_freq_yield}
	In the Hull--White model with Svensson--parameterized initial curves, the asymptotic relative frequencies of the shapes of the yield curve $x\mapsto y(x,t,\mathbf{r}(t))$ to appear under the risk neutral measure $\QQ$ are given by
	\begin{align*}
		\QQ(x\mapsto y(x,t,\mathbf{r}(t)) \text{ is \texttt{inverse}}) &\overset{t\rightarrow\infty}{\longrightarrow} 0.5,\\
		\QQ(x\mapsto y(x,t,\mathbf{r}(t)) \text{ is \texttt{humped}}) &\overset{t\rightarrow\infty}{\longrightarrow} 0.5 - \Phi\left(-\frac{3\sigma}{\sqrt{8\lambda^3}}\right),\\
		\QQ(x\mapsto y(x,t,\mathbf{r}(t)) \text{ is \texttt{normal}}) &\overset{t\rightarrow\infty}{\longrightarrow} \Phi\left(-\frac{3\sigma}{\sqrt{8\lambda^3}}\right),
	\end{align*}
	where $\Phi$ denotes the cumulative distribution function of the standard normal distribution.
\end{thm}

\begin{thm}\label{thm:asymptotic_rel_freq_forward}
	In the Hull--White model with Svensson--parameterized initial curves, the asymptotic relative frequencies of the shapes of the forward curve $x\mapsto f(x,t,\mathbf{r}(t))$ to appear under the risk neutral measure $\QQ$ are given by
	\begin{align*}
		\QQ(x\mapsto f(x,t,\mathbf{r}(t)) \text{ is \texttt{inverse}}) &\overset{t\rightarrow\infty}{\longrightarrow} 0.5,\\
		\QQ(x\mapsto f(x,t,\mathbf{r}(t)) \text{ is \texttt{humped}}) &\overset{t\rightarrow\infty}{\longrightarrow} 0.5 - \Phi\left(-\frac{\sqrt{2}\sigma}{\lambda^{3/2}}\right),\\
		\QQ(x\mapsto f(x,t,\mathbf{r}(t)) \text{ is \texttt{normal}}) &\overset{t\rightarrow\infty}{\longrightarrow} \Phi\left(-\frac{\sqrt{2}\sigma}{\lambda^{3/2}}\right),
	\end{align*}
	if $\lambda < \min\{1/\tau_1,1/\tau_2\}$ or \{$1/\tau_2 > 1/\tau_1 = \lambda$ and $\beta_2 = 0$\};
	\begin{align*}
		\QQ(x\mapsto f(x,t,\mathbf{r}(t)) \text{ is \texttt{dipped}}) &\overset{t\rightarrow\infty}{\longrightarrow} 0.5,\\
		\QQ(x\mapsto f(x,t,\mathbf{r}(t)) \text{ is \texttt{hd}}) &\overset{t\rightarrow\infty}{\longrightarrow} 0.5 - \Phi\left(-\frac{\sqrt{2}\sigma}{\lambda^{3/2}}\right),\\
		\QQ(x\mapsto f(x,t,\mathbf{r}(t)) \text{ is \texttt{normal}}) &\overset{t\rightarrow\infty}{\longrightarrow} \Phi\left(-\frac{\sqrt{2}\sigma}{\lambda^{3/2}}\right),
	\end{align*}
	if \eqref{eq:forward_term_inc} is satisfied; and
	\begin{align*}
		\QQ(x\mapsto f(x,t,\mathbf{r}(t)) \text{ is \texttt{inverse}}) &\overset{t\rightarrow\infty}{\longrightarrow} 0.5,\\
		\QQ(x\mapsto f(x,t,\mathbf{r}(t)) \text{ is \texttt{humped}}) &\overset{t\rightarrow\infty}{\longrightarrow} 0.5,
	\end{align*}
	if \eqref{eq:forward_term_dec} is satisfied.
\end{thm}
\begin{rem}
	A look into \cite{diez2020yield} reveals that the asymptotic relative frequencies for the shapes of the yield curve found in Theorem~\ref{thm:asymptotic_rel_freq_yield} are identical to those of the Vasicek model. The same is true for the shapes of the forward curve.
\end{rem}

To prove these results, we will use a series of technical lemmas, which explain the asymptotic behavior of the solutions of \eqref{eq:env_forward} and \eqref{eq:env_yield}. In a first step we determine the limits of $x$ and $r$ as $t\rightarrow\infty$.

\begin{lem}\label{lem:env_limit}
	Let $f_0$ satisfy
	\begin{align}\label{ass:limits}
		\lim_{t\rightarrow\infty} f_0'(t) = 0\quad\text{and}\quad \lim_{t\rightarrow\infty} f_0''(t) = 0.
	\end{align}
	For the solutions $(x,t)\in(0,\infty)^2$ of \eqref{eq:env_2} it holds that $x\rightarrow \infty$ as $t\rightarrow\infty$. Furthermore, if $f_0$ is from the Svensson family or one of its subfamilies, there exists some $K > 0$, such that for every $t > K$ there is a solution of \eqref{eq:env_2} if and only if
	\begin{enumerate}
		\item $\beta_2 = 0$ and $\frac{1}{\tau_1} < \frac{1}{\tau_2}$ and $\frac{1}{\tau_1} = \lambda$ and $\beta_3 > 0$, or 
		\item $\beta_2 = 0$ and \{$\beta_3 = 0$ or $\frac{1}{\tau_1} < \frac{1}{\tau_2}$\} and $\beta_1(\frac{1}{\tau_1} - \lambda) > 0$ and $2\lambda > \frac{1}{\tau_1}$, or
		\item \{$\beta_3 = 0$ or $\frac{1}{\tau_1} < \frac{1}{\tau_2}$\} and
		\begin{enumerate}
			\item $\beta_2(\frac{1}{\tau_1} - \lambda) > 0$ and $2\lambda > \frac{1}{\tau_1}$, or
			\item $\frac{1}{\tau_1} = \lambda$ and $\beta_2 < 0$, or
		\end{enumerate}
		\item $\frac{1}{\tau_2} < \frac{1}{\tau_1}$ and
		\begin{enumerate}
			\item $\beta_3(\frac{1}{\tau_2} - \lambda) > 0$ and $2\lambda > \frac{1}{\tau_2}$, or
			\item $\frac{1}{\tau_2} = \lambda$ and $\beta_3 < 0$.
		\end{enumerate}
	\end{enumerate}
\end{lem}
\begin{proof}
	For $t\rightarrow \infty$ the left hand side of \eqref{eq:env_2} converges to 0, so the right hand side has to follow, for there to be equality, hence $x\rightarrow\infty$. 
	
	For $f_0$ from the Svensson family, the left hand side of \eqref{eq:env_2} equates to
	\begin{align*}
		&\frac{\beta_1(\frac{1}{\tau_1} - \lambda) + \beta_2(\lambda - \frac{2}{\tau_1})}{\tau_1}{\mathrm{e}}^{-\frac{x + t}{\tau_1}} + \frac{\beta_2(\frac{1}{\tau_1} - \lambda)}{\tau_1^2}(x + t){\mathrm{e}}^{-\frac{x + t}{\tau_1}}\\&\qquad + \frac{\beta_3(\lambda - \frac{2}{\tau_2})}{\tau_2}{\mathrm{e}}^{-\frac{x + t}{\tau_2}} + \frac{\beta_3(\frac{1}{\tau_2} - \lambda)}{\tau_2^2}(x + t){\mathrm{e}}^{-\frac{x + t}{\tau_2}}.
	\end{align*}
	If this expression vanishes for every $x >0$, there are no solutions. Otherwise, there exists some $\tilde K > 0$ and constants $c_1,c_2 > 0$, such that
	\begin{align*}
		c_1{\mathrm{e}}^{-\frac{x+t}{\tau}} < |f_0''(x+t) + \lambda f_0'(x+t)| < c_2x{\mathrm{e}}^{-\frac{x+t}{\tau}},\quad t > \tilde K,
	\end{align*}
	for $\tau = \tau_1$ or $\tau = \tau_2$. So there exists some $ K > \tilde K$, such that there is a solution for every $t > K$, if and only if the terminal sign of the left hand side is positive and $2\lambda > \frac{1}{\tau}$.
\end{proof}

\begin{lem}
	Let $f_0$ satisfy \eqref{ass:limits}. For the solutions $(x,t,r)\in(0,\infty)^2\times\RR$ of \eqref{eq:env_yield} it holds that $x\rightarrow \infty$ as $t\rightarrow \infty$.
\end{lem}
\begin{proof}
	Let us assume, that there exists some $C_1 > 0$, such that for every $C_2 > 0$ there is a solution $(\hat x, \hat t, \hat r)$ of \eqref{eq:env_yield} with $\hat t > C_2$ and $\hat x < C_1$.
	By Lemma~\ref{lem:env_limit} we can choose $C_2$ such that every solution $(\bar x, \hat t)$ of \eqref{eq:env_2} satisfies $\bar x > C_1$. Since
	\begin{align*}
		0 = \partial_x y(\hat x, \hat t, \hat r) = \frac{1}{x^2}\int_0^{\hat x} \xi\partial_x f(\xi, \hat t, \hat r)\dif\xi,
	\end{align*}
	there must be some $\xi_1\in(0,\hat x)$ with $\partial_x f(\xi_1, \hat t, \hat r) = 0$, i.e. $r_{\xi_1}^{\mathrm{f}}(\hat t) = \hat r$. By \eqref{eq:env_yield_alt}, also $\partial_x f(\hat x, \hat t, \hat r) = 0$, i.e. $r_{\hat x}^{\mathrm{f}}(\hat t) = \hat r$, so there is $\xi_2\in(\xi_1,\hat x)$ with $\partial_x r_{\xi_2}^{\mathrm{f}}(\hat t) = 0$. Since $0 = \partial_x f(x,t,r_x^{\mathrm{f}}(t))$ for every $x,t > 0$, by definition of the tracking function $r_x^{\mathrm{f}}$, it holds
	\begin{align*}
		0 &= \frac{\dif}{\dif x}\left(\partial_x f(x, \hat t,r_x^{\mathrm{f}}(\hat t))\right)\\ &= \partial_{xx} f(\xi_2, \hat t, r_{\xi_2}^{\mathrm{f}}(\hat t)) + \partial_x r_{\xi_2}^{\mathrm{f}}(\hat t)\partial_{xr}f(\xi_2, \hat t, r_{\xi_2}^{\mathrm{f}}(\hat t)),
	\end{align*}
	and hence $\partial_{xx} f(\xi_2, \hat t, r_{\xi_2}^{\mathrm{f}}(\hat t)) = 0$. But this is not possible, since it implies, that $(\xi_2, \hat t)$ solves \eqref{eq:env_2}.
\end{proof}

\begin{lem}\label{lem:env_limit_HW}
	Let $f_0$ be from the Svensson family or one of its subfamilies.
	\begin{enumerate}
		\item For the solutions $(x,t,r)\in(0,\infty)^2\times\RR$ of \eqref{eq:env_forward} it holds that
		\begin{align*}
			r \rightarrow \beta_0 - \frac{\sigma^2}{2\lambda^2},\quad \text{as}\quad t\rightarrow\infty.
		\end{align*}
		\item For the solutions $(x,t,r)\in(0,\infty)^2\times\RR$ of \eqref{eq:env_yield} it holds that
		\begin{align*}
			r \rightarrow \beta_0 - \frac{\sigma^2}{4\lambda^2},\quad \text{as}\quad t\rightarrow\infty.
		\end{align*}
	\end{enumerate}
\end{lem}
\begin{proof}
	If the left hand side of \eqref{eq:env_2} does not vanish for every $x + t > 0$, then there exists some $ K > 0$, such that for every $x+t > K$ it can be bounded from below by
	\begin{align}\label{eq:lower_bounds}
		\begin{cases}
			c_1(x+t){\mathrm{e}}^{-\frac{x+t}{\tau_1}},\\ \qquad \text{for }\left\{\frac{1}{\tau_1} < \frac{1}{\tau_2} \text{ or }  \beta_3 = 0\right\} \text{ and } \frac{1}{\tau_1}\neq \lambda \text{ and } \beta_2 \neq 0,\\
			c_1{\mathrm{e}}^{-\frac{x+t}{\tau_1}},\\ \qquad \text{for } \left\{\frac{1}{\tau_1} < \frac{1}{\tau_2} \text{ or }  \beta_3 = 0\right\} \text{ and } \frac{1}{\tau_1}\neq \lambda \text{ and } \beta_2 = 0,\\
			c_1(x+t){\mathrm{e}}^{-\frac{x+t}{\tau_2}},\\ \qquad \text{for } \frac{1}{\tau_2} < \frac{1}{\tau_1} \text{ and } \frac{1}{\tau_1}\neq \lambda \text{ and } \beta_3 \neq 0,\\
			c_1{\mathrm{e}}^{-\lambda(x+t)}, \\ \qquad \text{for } \left\{\frac{1}{\tau_1} < \frac{1}{\tau_2} \text{ or }  \beta_3 = 0\right\} \text{ and } \frac{1}{\tau_1} = \lambda,\\
			c_1{\mathrm{e}}^{-\lambda(x+t)}, \\ \qquad \text{for } \frac{1}{\tau_2} < \frac{1}{\tau_1} \text{ and } \frac{1}{\tau_2} = \lambda,
		\end{cases}
	\end{align}
	for some $c_1 > 0$. On the other hand we can also find an upper bound for $|f_0'(x+t)|$ for every $x+t > \tilde{K}$ for some $\tilde{K} > 0$.
	\begin{align*}
		|f_0'(x+t)| < \begin{cases}
			c_2(x+t){\mathrm{e}}^{-\frac{x+t}{\tau_1}} ,\\ \qquad \text{for } \left\{\frac{1}{\tau_1} < \frac{1}{\tau_2} \text{ or }  \beta_3 = 0\right\} \text{ and } \frac{1}{\tau_1}\neq \lambda \text{ and } \beta_2 \neq 0,\\
			c_2{\mathrm{e}}^{-\frac{x+t}{\tau_1}} ,\\ \qquad \text{for } \left\{\frac{1}{\tau_1} < \frac{1}{\tau_2} \text{ or }  \beta_3 = 0\right\} \text{ and } \frac{1}{\tau_1}\neq \lambda \text{ and } \beta_2 = 0,\\
			c_2(x+t){\mathrm{e}}^{-\frac{x+t}{\tau_2}},\\ \qquad \text{for } \frac{1}{\tau_2} < \frac{1}{\tau_1} \text{ and } \frac{1}{\tau_1}\neq \lambda \text{ and } \beta_3 \neq 0,\\
			c_2(x+t){\mathrm{e}}^{-\lambda(x+t)} , \\ \qquad \text{for } \left\{\frac{1}{\tau_1} < \frac{1}{\tau_2} \text{ or }  \beta_3 = 0\right\} \text{ and } \frac{1}{\tau_1} = \lambda,\\
			c_2(x+t){\mathrm{e}}^{-\lambda(x+t)}, \\ \qquad \text{for } \frac{1}{\tau_2} < \frac{1}{\tau_1} \text{ and } \frac{1}{\tau_2} = \lambda,
		\end{cases}
	\end{align*}
	for some $c_2 > 0$. Multiplying with ${\mathrm{e}}^{\lambda x}$ in \eqref{eq:env_2} and letting $t\rightarrow\infty$ (and hence $x\rightarrow\infty$), the right hand side has limit 0. Thus, the lower bounds given in the first three cases in \eqref{eq:lower_bounds}, multiplied with ${\mathrm{e}}^{\lambda x}$, also approach 0. The upper bounds of $|f_0'(x+t)|{\mathrm{e}}^{\lambda x}$ only differ by a constant from the lower bounds, so $f_0'(x+t){\mathrm{e}}^{\lambda x}\rightarrow 0$. In the other two cases, because of the lower bounds \eqref{eq:lower_bounds}, there must be $C > 0$, such that, for the solutions $(x,t)$ of \eqref{eq:env_2},  $x < t + C.$ Hence, $|f_0'(x+t)|{\mathrm{e}}^{\lambda x} < c_2(2t + C){\mathrm{e}}^{-\lambda t}\rightarrow 0$ as $t\rightarrow\infty$. Solutions of \eqref{eq:env_forward} are, by definition, points on the curve $r^{\mathrm{f}}_x(t)$, so plugging this limit into \eqref{eq:env_1} yields $r\rightarrow \beta_0 - \frac{\sigma^2}{2\lambda^2}$.
	
	Since there exists some $K > 0$, such that
	\begin{align*}
		|yf_0'(y+t)| \le c_1y^2{\mathrm{e}}^{\frac{-y}{\tau}}{\mathrm{e}}^{\frac{-t}{\tau}},\quad y + t > K
	\end{align*}
	for some $c_1 > 0$ and $\tau = \tau_1$ or $\tau = \tau_2$, and since $x\rightarrow\infty$ as $t\rightarrow\infty$ for the solutions of \eqref{eq:env_yield}, the integral in \eqref{eq:ell_infty_gen} converges to 0 as $t\rightarrow\infty$ and the right hand side of \eqref{eq:ell_infty_gen} converges to $\beta_0 - \frac{\sigma^2}{4\lambda^2}$. 
\end{proof}

We can now show that the envelope corresponding to the yield curve is asymptotically single-valued at most, which significantly limits the attainable shapes for larger $t$. A similar result can be obtained for the forward curve with Nelson--Siegel--parameterized initial curves by studying \eqref{eq:Lambert-W-NS}, but is omitted here.

\begin{lem}\label{lem:one_solution}
	There is a $K > 0$, such  that \eqref{eq:env_yield} has exactly one solution for every $t > K$, if and only if one of the conditions in Lemma~\ref{lem:asymptotic_att_shapes} is satisfied.
	\eqref{eq:env_yield} has no solution in every other case for every $t > K$.
\end{lem}
\begin{proof}
	In a first step, we want to establish one as an upper bound for the number of solutions when $t$ is large enough. 
	By Lemma~\ref{lem:env_limit_HW}, we can choose $K_1 > 0$, such that for every solution $(x,t,r)$ of \eqref{eq:env_forward} with $t > K_1$ it holds $|\beta_0 - \frac{\sigma^2}{2\lambda^2} - r| < \frac{\sigma^2}{8\lambda^2}$; and $K_2 > 0$, such that for every solution $(x,t,r)$ of \eqref{eq:env_yield} with $t > K_2$ it holds $|\beta_0 - \frac{\sigma^2}{4\lambda^2} - r| < \frac{\sigma^2}{8\lambda^2}$.
	Also, we can choose $K_3 > 0$, such that $|\beta_0 + \frac{\sigma^2}{2\lambda^2} - r_0(t)| < \frac{\sigma^2}{2\lambda^2}$ for every $t > K_3$; and $K_4 > 0$, such that $\partial_x r_x^{\mathrm{y}}(t)|_{x = 0+} < 0$ for every $t > K_4$. 
	Now, for every $t > K:=\max\{K_1,K_2,K_3,K_4\}$, the tracking function $x\mapsto r_x^{\mathrm{y}}(t)$ is initially decreasing and changes direction in $x^*$, if and  only if $(x^*,t, r_{x^*}^{\mathrm{y}}(t))$ is a solution of \eqref{eq:env_yield} (there are no points of regression since the envelopes of the forward and yield curve do not intersect). Similarly the tracking function $x\mapsto r_x^{\mathrm{f}}(t)$ can only change direction in $x^*$ when $(x^*,t, r_{x^*}^{\mathrm{f}}(t))$ is a solution of \eqref{eq:env_forward}, and it is also initially decreasing. 
	Thus, it starts at $r_0{(t)}$, strictly decreases, possibly changes direction multiple times in the interval $(\beta_0 - \frac{5\sigma^2}{8\lambda^2}, \beta_0 - \frac{3\sigma^2}{8\lambda^2})$ (where $\beta_0 - \frac{3\sigma^2}{8\lambda^2} < r_0{(t)}$) and finally strictly decreases or strictly increases. Therefore, possible shapes corresponding to points in the interval $(\beta_0 - \frac{3\sigma^2}{8\lambda^2}, r_0{(t)})$ are \texttt{normal, humped, hd}. 
	Assuming that there is more than one solution of \eqref{eq:env_yield} for some $t' > K$, the tracking function $x\mapsto r_x^{\mathrm{y}}(t')$ starts at $r_0{(t')}$, strictly decreases and changes direction at least twice in the interval $(\beta_0 - \frac{3\sigma^2}{8\lambda^2}, \beta_0 - \frac{\sigma^2}{8\lambda^2}) \subset (\beta_0 - \frac{3\sigma^2}{8\lambda^2}, r_0{(t')})$, yielding the shape \texttt{hdh} or a more complex one for corresponding points, which contradicts \cite[Thm.~3.3]{keller2021classification}.
	The tracking function $x\mapsto r_x^{\mathrm{y}}(t')$ changes its direction exactly once if and only if initial and terminal direction are different. Since it is initially decreasing, this happens exactly when it is terminally increasing. Letting $t\rightarrow\infty$ in \eqref{eq:term_dir_track_y} yields the desired result.
\end{proof}

\begin{proof}[Proof of Lemma~\ref{lem:asymptotic_att_shapes}]
	We assume here that one of the conditions of Lemma~\ref{lem:one_solution} is satisfied. The other case can be proven in a similar way. Since
	\begin{align*}
		\lim_{t\rightarrow\infty} r_0(t) = \beta_0 + \frac{\sigma^2}{2\lambda^2} > \beta_0 - \frac{\sigma^2}{4\lambda^2} = \lim_{t\rightarrow\infty} \bar{r}_\infty^{\mathrm{y}}(t),
	\end{align*}
	we can choose $K > 0$, such that $r_0(t) > r_\infty^{\mathrm{y}}(t)$ for every $t > K$. It can also be ensured that, for every $t > K$, $\partial_x r_x^{\mathrm{y}}(t)|_{x = 0^+} < 0$ and that \eqref{eq:env_yield} has exactly one solution, which is not a point of regression. Hence, the tracking function is initially decreasing and changes direction exactly once, not reaching its starting level again. $r_0(t), r_\infty^{\mathrm{y}}(t)$ and the point where the tracking function changes direction divide the real line into four (half-)open intervals corresponding to the four shapes \texttt{n}, \texttt{i}, \texttt{h} and \texttt{hd}.
\end{proof}
\begin{proof}[Proof of Theorem~\ref{thm:asymptotic_rel_freq_yield} and Theorem~\ref{thm:asymptotic_rel_freq_forward}]
	The short-rate $\mathbf{r}(t)$ is normally distributed with
	\begin{align*}
		\EE[\mathbf{r}(t)] &= (\mathbf{r}(0) - f_0(0)){\mathrm{e}}^{-\lambda t} + f_0(t) + \frac{\sigma^2}{2\lambda^2}\left(1 - {\mathrm{e}}^{-\lambda t}\right)^2,\\
		\VV(\mathbf{r}(t)) &= \frac{\sigma^2}{2\lambda}\left(1 - {\mathrm{e}}^{-2\lambda t}\right),
	\end{align*}
	such that
	\begin{align*}
		\lim_{t\rightarrow\infty} \EE[\mathbf{r}(t)] = \beta_0 + \frac{\sigma^2}{2\lambda^2},\qquad \lim_{t\rightarrow\infty} \VV[\mathbf{r}(t)] = \frac{\sigma^2}{2\lambda}.
	\end{align*}
	The calculation of the relative frequencies follows the same approach in all cases covered by the theorems; we will demonstrate it for only one of them. Assume that \eqref{eq:forward_term_inc} is satisfied and take a look at the forward curve. It holds that
	\begin{align*}
		\lim_{t\rightarrow\infty} r_0(t) = \beta_0 + \frac{\sigma^2}{2\lambda^2},
	\end{align*}
	and, for the solutions $(x,t,r)$ of \eqref{eq:env_forward}, $r \rightarrow \beta_0 - \frac{\sigma^2}{2\lambda^2}$ as $t\rightarrow\infty$. Furthermore, $\lim_{x\rightarrow\infty} r_x^{\mathrm{f}}(t) = \infty$ for every $t > 0$ and, for $t$ large enough, the tracking function is initially decreasing. So, starting at $r_0(t)$, the tracking function decreases, changes its direction (possibly multiple times) in a small interval around $\beta_0 - \frac{\sigma^2}{2\lambda^2}$ and terminally increases with limit $\infty$. Attained shapes are  $\texttt{dipped}$ for $\mathbf{r}(t) \ge r_0(t)$, \texttt{humped} for $\beta_0 - \frac{\sigma^2}{2\lambda^2} + \epsilon < \mathbf{r}(t) < r_0(t)$, \texttt{normal} for $\mathbf{r}(t) < \beta_0 - \frac{\sigma^2}{2\lambda^2} - \epsilon$, and non-specified shapes for $\beta_0 - \frac{\sigma^2}{2\lambda^2} - \epsilon < \mathbf{r}(t) < \beta_0 - \frac{\sigma^2}{2\lambda^2} + \epsilon$; for some $\epsilon > 0$ with $\epsilon \rightarrow0$ as $t\rightarrow\infty$. So, for $t$ large enough,
	\begin{align*}
		&\QQ(x\mapsto f(x,t,\mathbf{r}(t)) \text{ is \texttt{dipped}}) = \QQ(\mathbf{r}(t) \ge r_0(t))\\ 
		&\qquad = 1 - \Phi\left(\frac{r_0(t) - \EE[\mathbf{r}(t)]}{\sqrt{\VV(\mathbf{r}(t))}}\right)  \overset{t\rightarrow\infty}{\longrightarrow}  0.5.
	\end{align*}
	The other values are calculated in a similar fashion.
\end{proof}

\appendix

\section{Tchebycheff systems}

In order to give a full classification of attainable term structure shapes, we often need to establish upper bounds to the number of local extrema of forward and yield curve, or a related function (e.g. a Wronskian).
In \cite{keller2021classification} and \cite{KRS23}, Descartes systems have been used to fulfill this task. Their variation--diminishing property proved especially useful when studying different parameter regimes.
In the models studied in \cite{KRS26} and here, Descartes systems are not applicable. Instead, we have to work with the more general notion of Tchebycheff systems, thus, still obtaining reasonable upper bounds, but losing the variation--diminishing property.

In this section, which is based on definitions and results found in \cite{KS66},  we will introduce the concept of a (extended complete) Tchebycheff system and its key property. We show that exponential--monomial functions form a Tchebycheff system and study how these systems behave under certain transformations.

\subsection{Definitions and key properties}

We will start by defining the notion of a Tchebycheff system and a polynomial.

\begin{defn}[{\cite[Ch. I, Def. 1.1]{KS66}}]
	Let $u_0,\ldots, u_n$ denote continuous real--valued functions defined on a closed interval $[a,b]$. These functions will be called a \textbf{Tchebycheff system} over $[a,b]$ (T-system over $[a,b]$, or just T-system) provided the $n+1$st order determinants
	\begin{align}\label{eq:T-system_det}
		U\begin{pmatrix}
			0,1,\cdots,n\\t_0,t_1,\cdots,t_n
		\end{pmatrix} := \det\begin{pmatrix}
			u_0(t_0) & u_0(t_1) & \cdots & u_0(t_n)\\
			u_1(t_0) & u_1(t_1) & \cdots & u_1(t_n)\\
			\vdots & \vdots & & \vdots\\
			u_n(t_0) & u_n(t_1) & \cdots & u_n(t_n)\\
		\end{pmatrix}
	\end{align}
	are strictly positive whenever $a\le t_1 < t_2 < \cdots < t_n \le b$. The functions will be referred to as a \textbf{complete Tchebycheff system} (CT-system) if $\{u_i\}_{i=0}^r$ is a T-system for each $r = 0,1,\ldots,n$.
\end{defn}

\begin{defn}[{\cite[Ch. I, Def. 4.1]{KS66}}]
	Let the system $\{u_i\}_{i = 0}^n$ be defined on some interval $[a,b]$. A function of the form $u = \sum_{i = 0}^n a_iu_i$ where $a_i$ are real numbers will be called a \textbf{$u$-polynomial} (or just polynomial). A polynomial is said to be nontrivial if $\sum_{i=0}^n a_i^2 > 0$.
\end{defn}

It is crucial for our analysis to find upper bounds to the number of zeros of certain polynomials. Once we established that we are working with a Tchebycheff system, this is an easy task, as the following theorem shows.

\begin{thm}[{\cite[Ch. I, Thm. 4.1]{KS66}}]\label{thm:T_charact}
	If $\{u_i\}_{i = 0}^n$ is a T-system over $[a,b]$, then every nontrivial $u$-polynomial has at most $n$ distinct zeros in $[a,b]$.
\end{thm}

Explicitly determining if \eqref{eq:T-system_det} is strictly positive will often be quite difficult. To circumvent this problem, we will mainly use the stronger notion of an extended complete Tchebycheff system, making calculations a bit simpler.

Let $u_i\in C^p[a,b],\ i=0,1,\ldots,n$ for some $p\ge 1$. For a choice of $a\le t_0 \le t_1\le\cdots\le t_n\le b$ with $t_i = t_{i + 1} = \cdots = t_{i + q}, 0\le q\le p$ we will denote with
\begin{align*}
	U^*\begin{pmatrix}
		0,1,\cdots,n\\t_0,t_1,\cdots,t_n
	\end{pmatrix}
\end{align*}
the determinant, that is obtained by exchanging the $i + 1 + j$th column in \eqref{eq:T-system_det}, $0\le j\le q$, by the column vector
\begin{align*}
	\left(\frac{\partial^j}{\partial t_i^j}u_0(t_i), \frac{\partial^j}{\partial t_i^j}u_1(t_i),\cdots,\frac{\partial^j}{\partial t_i^j}u_n(t_i)\right)^\top
\end{align*}

\begin{defn}[{\cite[Ch. I, Thm. 2.4]{KS66}}]
	The functions $u_0, u_1,\ldots,u_n$ will be called an \textbf{extended Tchebycheff system of order $p$ over $[a,b]$} (ET-system of order $p$), provided $u_i\in C^p[a,b],\ i=0,1,\ldots,n$ and
	\begin{align}\label{eq:ET-system_det}
		U^*\begin{pmatrix}
			0,1,\cdots,n\\t_0,t_1,\cdots,t_n
		\end{pmatrix} > 0
	\end{align}
	for all choices $a\le t_0 \le t_1\le\cdots\le t_n\le b$, where equality occurs in groups of at most $p$ consecutive $t_i$ values. An extended Tchebycheff system of order $n + 1$ will be referred to simply as an ET-system. The functions will be referred to as an \textbf{extended complete Tchebycheff system} (ECT-system) if $\{u_i\}_{i=0}^r$ is an ET-system for each $r = 0,1,\ldots,n$.
\end{defn}

With the stronger notion of ET-systems, we can bound the number of zeros even when counting multiplicity.

\begin{thm}[{\cite[Ch. I, Thm. 4.3]{KS66}}]\label{thm:T-system_upper_bounds}
	If $\{u_i\}_{i = 0}^n$ is an ET-system over $[a,b]$, then every nontrivial $u$-polynomial has at most $n$ zeros counting multiplicity in $[a,b]$.
\end{thm}

The following Theorem provides us with an easy way to determine if a given system is ECT.

\begin{thm}[{\cite[Ch. XI, Thm. 1.1]{KS66}}]\label{thm:ECT_charact}
	Let $u_0,u_1,\ldots,u_n$ be of class $C^n[a,b]$. Then $\{u_i\}_{i = 0}^n$ is an ECT-system on $[a,b]$ if and only if for $r = 0,1,\ldots,n$ we have $W(u_0,\ldots,u_r) > 0$ on $[a,b]$, where $W(u_0,\ldots,u_r)$ denotes the Wronskian of the functions $u_0,u_1,\ldots,u_r$, i.e.
	\begin{align*}
		W(u_0,u_1,\ldots,u_r)(t) = \det\begin{pmatrix}
			u_0(t) & u_0'(t) & \cdots & u_0^{(r)}(t)\\
			u_1(t) & u_1'(t) & \cdots & u_1^{(r)}(t)\\
			\vdots & \vdots & & \vdots\\
			u_r(t) & u_r'(t) & \cdots & u_k^{(r)}(t)\\
		\end{pmatrix}.
	\end{align*}
\end{thm}

It should be noted at this point, that when changing the order of functions in an ECT-system, one does not necessarily obtain an ECT-system again. Nevertheless, the upper bound to the number of zeros of polynomials, given in  Theorem~\ref{thm:T-system_upper_bounds}, still holds.

Unless we specify otherwise, whenever we speak about an (EC)T-system, we mean a system on $[0,\infty)$, i.e. a system, that is (EC)T on every interval $[0,b]$ with $b > 0$.

\subsection{The exponential--polynomial family}

Let $r\in\NN$, $\lambda_1,\ldots,\lambda_r\in\RR$ pairwise distinct, and $m_1,\ldots,m_r\in\NN$ with $\sum_{k=1}^r m_k =: M$. We define
\begin{equation}
	\begin{gathered}\label{eq:exppolsys}
		u_0(x) := {\mathrm{e}}^{\lambda_1 x}, u_1(x) := x{\mathrm{e}}^{\lambda_1 x},\ldots, u_{m_1-1}(x) := x^{m_1-1}{\mathrm{e}}^{\lambda_1 x}\\
		u_{m_1}(x) := {\mathrm{e}}^{\lambda_2 x}, u_{m_1 + 1}(x) := x{\mathrm{e}}^{\lambda_2 x},\ldots, u_{m_1 + m_2-1}(x) := x^{m_2-1}{\mathrm{e}}^{\lambda_2 x}\\
		\vdots\\
		u_{\sum_{k=1}^{r-1}m_k}(x) := {\mathrm{e}}^{\lambda_r x}, \ldots, u_{M-1}(x) := x^{m_r-1}{\mathrm{e}}^{\lambda_r x}
	\end{gathered}
\end{equation}
A function $u$ is of the \textbf{exponential--polynomial type}, if it is a polynomial of the above defined functions. 

We want to show, that the system \eqref{eq:exppolsys} is an ECT-system and thus calculate the Wronskian; using the following result.

\begin{lem}
	For every $i,l\in\NN_0$, it holds
	\begin{align}\label{eq:exp_der}
		\frac{\mathrm{d}^i}{\dif x^i} \left(x^l\mathrm{e}^{\lambda x}\right) = \mathrm{e}^{\lambda x}\sum_{j=0}^{i\land l} {l\choose j} \frac{i!}{(i-j)!}\lambda^{i-j}x^{l-j}.
	\end{align}
	In particular it holds
	\begin{align}\label{eq:exp_der_0}
		\frac{\mathrm{d}^i}{\dif x^i} \left(x^l\mathrm{e}^{\lambda x}\right)\Big|_{x = 0} = l!{i\choose l}\lambda^{i - l},
	\end{align}
	where we use the convention ${i\choose l} := 0$ for $i < l$.
\end{lem}
\begin{proof}
	We use induction over $i$. Plugging $i = 0$ into \eqref{eq:exp_der} yields
	\begin{align*}
		\frac{\mathrm{d}^0}{\dif x^0} \left(x^l\mathrm{e}^{\lambda x}\right) = \mathrm{e}^{\lambda x} {0\choose 0} \frac{0!}{0!}\lambda^{0}x^{l-0} = x^l\mathrm{e}^{\lambda x}.
	\end{align*}
	Now assuming the formula holds for $i\in\NN_0$, we show that it also holds for $i + 1$.
	\begin{align*}
		&\frac{\mathrm{d}^{i + 1}}{\dif x^{i + 1}} \left(x^l\mathrm{e}^{\lambda x}\right) = \frac{\mathrm{d}^i}{\dif x^i}\left(\lambda x^l\mathrm{e}^{\lambda x} + lx^{l-1}\mathrm{e}^{\lambda x}\right)\\
			&\quad = \mathrm{e}^{\lambda x}\left[\lambda \sum_{j=0}^{i\land l} {l\choose j} \frac{i!}{(i-j)!}\lambda^{i-j}x^{l-j} + l \sum_{j=0}^{i\land (l - 1)} {l - 1\choose j} \frac{i!}{(i-j)!}\lambda^{i-j}x^{(l - 1) -j}\right]\\
			&\quad =  \mathrm{e}^{\lambda x}\left[{l\choose 0} \frac{i!}{i!}\lambda^{(i+1) - 0}x^{l-0} + \lambda \sum_{j=1}^{i\land l} {l\choose j} \frac{i!}{(i-j)!}\lambda^{i-j}x^{l-j}\right.\\
			&\qquad + l \sum_{j=1}^{i\land l} {l - 1\choose j - 1} \frac{i!}{(i-(j - 1))!}\lambda^{i-(j - 1)}x^{(l - 1) -(j - 1)}\\
			&\qquad \left.+ \mathds{1}_{\{i \le l - 1\}}l{l-1\choose i}\frac{i!}{0!}\lambda^{i-i}x^{l-1 - i}\right]\\
			&\quad =  \mathrm{e}^{\lambda x}\left[{l\choose 0} \frac{(i + 1)!}{(i + 1)!}\lambda^{(i+1) - 0}x^{l-0} + \sum_{j=1}^{i\land l} \frac{l!}{(l - j)!} {i\choose j}\lambda^{(i + 1)-j}x^{l-j}\right.\\
			&\qquad + \sum_{j=1}^{i\land l} \frac{l!}{(l-j)!} {i \choose j - 1}\lambda^{(i + 1)-j}x^{l -j}\\
			&\qquad \left.+ \mathds{1}_{\{i + 1 \le l\}}{l\choose i + 1}\frac{(i + 1)!}{0!}\lambda^{(i + 1) - (i + 1)}x^{l - (i + 1)}\right]\\
			&\quad =  \mathrm{e}^{\lambda x}\left[{l\choose 0} \frac{(i + 1)!}{(i + 1)!}\lambda^{(i+1) - 0}x^{l-0} + \sum_{j=1}^{i\land l} \frac{l!}{(l - j)!} {i + 1\choose j}\lambda^{(i + 1)-j}x^{l-j}\right.\\
			&\qquad \left.+ \mathds{1}_{\{i + 1 \le l\}}{l\choose i + 1}\frac{(i + 1)!}{0!}\lambda^{(i + 1) - (i + 1)}x^{l - (i + 1)}\right]\\
			&\quad =  \mathrm{e}^{\lambda x}\left[{l\choose 0} \frac{(i + 1)!}{(i + 1)!}\lambda^{(i+1) - 0}x^{l-0} + \sum_{j=1}^{i\land l} {l \choose j} \frac{(i + 1)!}{((i + 1) - j)!}\lambda^{(i + 1)-j}x^{l-j}\right.\\
			&\qquad \left.+ \mathds{1}_{\{i + 1 \le l\}}{l\choose i + 1}\frac{(i + 1)!}{0!}\lambda^{(i + 1) - (i + 1)}x^{l - (i + 1)}\right]\\
			&\quad = \mathrm{e}^{\lambda x}\sum_{j=0}^{(i + 1)\land l} {l\choose j} \frac{(i + 1)!}{((i + 1)-j)!}\lambda^{(i + 1)-j}x^{l-j}.
	\end{align*}
	Now let $x = 0$. If $i < l$, the sum in \eqref{eq:exp_der} vanishes. If $l \le i$,  the only summand that does not vanish corresponds to $j = l$, and evaluates to
	\begin{align*}
		{l\choose l}\frac{i!}{(i - l)!}\lambda^{i - l} = l!{i\choose l}\lambda^{i - l}.
	\end{align*}
\end{proof}

\begin{thm}\label{thm:Wronskian}
	Let $u_0,\ldots, u_{M-1}$ be as in \eqref{eq:exppolsys}. Then
	\begin{align}\label{eq:Wronskian}
		\begin{split}&W(u_0, \ldots,u_{M-1})(x) =\\ &\qquad\exp\left(x\sum_{k=1}^r m_k\lambda_k\right)\prod_{k=1}^r\left[\left(\prod_{l=1}^{m_k-1} l!\right)\left(\prod_{l=1}^{k-1}(\lambda_k - \lambda_l)^{m_km_l}\right)\right]
		\end{split}
	\end{align}
\end{thm}
\begin{proof}
	We show that the formula holds for $x = 0$.
	\begin{align*}
		&W(u_0, \ldots,u_{M-1})(0) = \det\left(\left(\frac{\mathrm{d}^i}{\dif x^i} u_j(0)\right)_{i=0,j=0}^{M-1, M-1}\right)\\ 
		&\quad = \det\left(\left(\frac{\mathrm{d}^i}{\dif x^i} \left(x^l\mathrm{e}^{\lambda_1 x}\right)\Big|_{x = 0}\right)_{i=0,l=0}^{M-1, m_1-1}\Big| \cdots \Big|\left(\frac{\mathrm{d}^i}{\dif x^i} \left(x^l\mathrm{e}^{\lambda_1 x}\right)\Large|_{x = 0}\right)_{i=0,l=0}^{M-1, m_r-1}\right)\\
		&\quad \overset{\eqref{eq:exp_der_0}}{=} \det\left(\left(l!{i\choose l}\lambda_1^{i - l}\right)_{i=0,l=0}^{M-1, m_1-1}\Big| \cdots \Big|\left(l!{i\choose l}\lambda_r^{i - l}\right)_{i=0,l=0}^{M-1, m_r-1}\right)\\
		&\quad = \left(\prod_{k=1}^{r}\prod_{l=1}^{m_k - 1}l!\right)\det\left(\left({i\choose l}\lambda_1^{i - l}\right)_{i=0,l=0}^{M-1, m_1-1}\Big| \cdots \Big|\left({i\choose l}\lambda_r^{i - l}\right)_{i=0,l=0}^{M-1, m_r-1}\right)\\
		&\quad = \left(\prod_{k=1}^{r}\prod_{l=1}^{m_k - 1}l!\right)\left(\prod_{k=1}^{r}\prod_{l=1}^{k-1}(\lambda_k - \lambda_l)^{m_km_l}\right).
	\end{align*}
	The last equality is due to the determinant of the confluent Vandermonde matrix \cite[Cor. 1]{HA198069}. By \cite[p. 166]{Wal98} the Wronskian satisfies the differential equation
	\begin{align*}
		W(u_0, \ldots,u_{M-1})'(x) = \left(\sum_{i=1}^r m_i\lambda_i\right)W(u_0, \ldots,u_{M-1})(x),
	\end{align*}
	which yields \eqref{eq:Wronskian} for $x \in \RR$.
\end{proof}

\begin{cor}\label{cor:ECT_forward}
	\eqref{eq:exppolsys} constitutes an ECT-system if $\lambda_1 < \lambda_2 < \ldots < \lambda_r$.
\end{cor}
\begin{proof}
	If $\lambda_1 < \lambda_2 < \ldots < \lambda_r$, the Wronskian \eqref{eq:Wronskian} is positive on $\RR$. Analogously, the Wronskian corresponding to the subsystems $\{u_i\}_{i=0}^k,$ $0\le k\le M-1$ is positive. By Theorem~\ref{thm:ECT_charact}, \eqref{eq:exppolsys} is an ECT-system.
\end{proof}

\subsection{Transformations of ECT-systems}

Corollary~\ref{cor:ECT_forward} is used to analyze the forward curve. To subsequently obtain bounds for the yield curve, we will need the following results.

\begin{lem}\label{lem:integral_representation}
	For $x > 0$ let
	\begin{align*}
		y(x) = \frac{1}{x}\int_0^x f(\xi)\dif \xi,
	\end{align*}
	with $f\in C^\infty([0,\infty))$. Then
	\begin{align*}
		y^{(n)}(x) = \frac{1}{x^{n+1}}\int_0^x \xi^n f^{(n)}(\xi)\dif \xi\quad  \text{and} \quad \lim_{x\downarrow 0} y^{(n)}(x) = \frac{f^{(n)}(0)}{n+1}.
	\end{align*}
\end{lem}
\begin{proof}
	We prove the first identity using induction. By assumption the identity holds for $n = 0$. Since
	\begin{align*}
		x^nf^{(n)}(x) &= \frac{\dif}{\dif x}\left(x^{n+1}y^{(n)}(x)\right)\\
		&= (n+1)x^ny^{(n)}(x) + x^{n+1}y^{(n+1)}(x),
	\end{align*}
	we can use partial integration to obtain
	\begin{align*}
		x^{n+2}y^{(n+1)}(x) &= x^{n+1}f^{(n)}(x) - (n+1)x^{(n+1)}y^{(n)}(x)\\
		&= x^{n+1}f^{(n)}(x) - (n+1)x^{n+1}\frac{1}{x^{n+1}}\int_0^x \xi^nf^{(n)}(\xi)\dif \xi\\
		&= x^{n+1}f^{(n)}(x) - \int_0^x (n+1)\xi^nf^{(n)}(\xi)\dif \xi\\
		&= \int_0^x x^{n+1}f^{(n+1)}(\xi)\dif \xi.
	\end{align*}
	Dividing by $x^{n+2}$ yields the desired result. To show the second identity we use l'Hospital's rule
	\begin{align*}
		\lim_{x\downarrow 0} y^{(n)}(x) &= \lim_{x\downarrow 0} \frac{1}{x^{n+1}}\int_0^x \xi^n f^{(n)}(\xi)\dif \xi\\ &= \lim_{x\downarrow 0} \frac{x^nf^{(n)}(x)}{(n+1)x^n} = \frac{f^{(n)}(0)}{n+1}.
	\end{align*}
\end{proof}

\begin{lem}\label{lem:ECT_yield}
	Let $\{u_i\}_{i=0}^n$ be an ECT-system on $[0,\infty)$. Then $\{v_i\}_{i=0}^n$ with
	\begin{align*}
		v_i(x) &= \frac{1}{x^2}\int_0^x \xi u_i(\xi)\dif \xi,\quad x > 0,\\
		v_i(0) &= \frac{u_i'(0)}{2},\quad 0\le i\le n,
	\end{align*}
	is also an ECT-system on $[0,\infty)$.
\end{lem}
\begin{proof}
	Since multiplying all the functions $u_i$ or all the functions $v_i$ with the same positive factor, dependent on $x$ and smooth, does not change the property of being (or not being) an ECT-system (see the integral representation from \cite[Ch. XI, Thm. 1.2.]{KS66}) $\{xu_i(x)\}_{i=0}^n$ is an ECT-system on $(0, \infty)$. And $\{v_i\}_{i=0}^n$ is an ECT-system on $(0, \infty)$, if and only if $(\tilde{v}_i)_{i=0}^n$ with
	\begin{align*}
		\tilde{v}_i(x) = \int_0^x \xi u_i(\xi)\dif \xi,\quad 0\le i\le n,
	\end{align*}
	is an ECT-system on $(0, \infty)$. For $x > 0$, it holds that
	\begin{align*}
		W(\tilde{v_0},\ldots,\tilde{v_i})(x) &= \int_0^x \begin{vmatrix}\xi u_1(\xi) & \cdots & \xi u_i(\xi) \\ xu_1(x) & \cdots & xu_i(x) \\ \frac{\dif}{\dif x}\left(xu_1(x)\right) & \cdots & \frac{\dif}{\dif x}\left(xu_i(x)\right) \\ \vdots & \ddots & \vdots \\ \frac{\dif^{i-1}}{\dif x^{i-1}}\left(xu_1(x)\right) & \cdots & \frac{\dif^{i-1}}{\dif x^{i-1}}\left(xu_i(x)\right)\end{vmatrix} \dif \xi.
	\end{align*}
	The integrand is positive since $\{xu_i(x)\}_{i=0}^n$ is an ECT-system, hence \\$W(v_0, \ldots, v_i)(x) > 0$ for every $x > 0,\ 0\le i\le n$. By Theorem~\ref{thm:ECT_charact}, $\{\tilde{v}_i\}_{i=0}^n$ is an ECT-system. Taking the limit $x\downarrow 0$ and using Lemma~\ref{lem:integral_representation}, we obtain
	\begin{align*}
		W(v_0, \ldots, v_i)(0) = \frac{1}{(i+1)!}W(u_0, \ldots, u_i)(0),
	\end{align*}
	so the Wronskian is positive for every $x \ge 0$ and $\{v_i\}_{i=0}^n$ is an ECT-system on $[0,\infty).$
\end{proof}

As a consequence of Lemma~\ref{lem:ECT_yield}, the derivative of the yield curve must be a polynomial of an ECT-system, if the derivative of the forward curve is a polynomial of (another) ECT-system. Crucially, the proof did not rely on a specific choice, so it can be used in the context of the exponential--polynomial family here, or in a different setting.

Finally, let us consider another transformation of an ECT-system $\{u_i\}_{i=0}^n$.
\begin{align*}
	v_i(x) := W(u_0,u_i)(x) = u_0(x)u_i'(x) - u_0'(x)u_i'(x),\quad i=1,\ldots,n.
\end{align*}
Using the identity
\begin{align*}
	\begin{vmatrix}
		u_0^{(k)} & u_i^{(k)} \\ u_0^{(l)} & u_i^{(l)}
	\end{vmatrix}
	&= \frac{1}{u_0} 
	\begin{vmatrix} 
		u_0^{(k)} & u_0 u_i^{(k)} - u_0^{(k)} u_i \\ u_0^{(l)} & u_0 u_i^{(l)} - u_0^{(l)} u_i
	\end{vmatrix}\\
	&= \frac{u_0^{(k)}}{u_0}
	\begin{vmatrix}
		u_0 & u_i \\ u_0^{(l)} & u_i^{(l)}
	\end{vmatrix}
	- \frac{u_0^{(l)}}{u_0}
	\begin{vmatrix}
		u_0 & u_i \\ u_0^{(k)} & u_i^{(k)}
	\end{vmatrix},
\end{align*}
we obtain
\begin{align*}
	&\frac{{\mathrm{d}}^j}{{\mathrm{d}} x^j} v_i = \frac{{\mathrm{d}}^j}{{\mathrm{d}} x^j} (u_0u_i' - u_0'u_i)
	= \sum_{k=0}^j {j \choose k}\left[u_0^{(k)}u_i^{(j-k+1)} - u_0^{(j-k+1)}u_i^{(k)}\right]\\ &= \sum_{k=0}^j {j \choose k}
	\begin{vmatrix}
		u_0^{(k)} & u_i^{(k)} \\ u_0^{(j-k+1)} & u_i^{(j-k+1)}
	\end{vmatrix}
	= \begin{vmatrix}
		u_0 & u_i \\ u_0^{(j+1)} & u_i^{(j+1)}
	\end{vmatrix}\\
	&\qquad+ \sum_{k=1}^j {j \choose k}\left[\frac{u_0^{(k)}}{u_0}
	\begin{vmatrix}
		u_0 & u_i \\ u_0^{(j-k+1)} & u_i^{(j-k+1)}
	\end{vmatrix}
	-\frac{u_0^{(j-k+1)}}{u_0}
	\begin{vmatrix}
		u_0 & u_i \\ u_0^{(k)} & u_i^{(k)}
	\end{vmatrix}
	\right],
\end{align*}
which yields
\begin{align*}
	W(v_1,\ldots,v_n) &= \det\left(\left(\begin{vmatrix}
		u_0 & u_i \\ u_0^{(j+1)} & u_i^{(j+1)}
	\end{vmatrix}\right)_{j = 0,i = 1}^{n-1, n}\right)\\ &= W(u_0,\ldots,u_n),
\end{align*}
where the last equality is due to condensation \cite[Sec. 20]{Ait56} and leads us to the following observation:
\begin{lem}\label{lem:ECT_der}
	If $\{u_i\}_{i=0}^n$ is an ECT-system on $[0,\infty)$, then so is $\{v_i\}_{i=1}^n$ with
	\begin{align*}
		v_i(x) = u_0(x)u_i'(x) - u_0'(x)u_i(x),\quad i=1,\ldots,n.
	\end{align*}
\end{lem}

\section{Checking assumptions \textbf{({A}1)} -- \textbf{({A}6)}} \label{sec:assumptions}

\begin{lem}
	The functions $b_{\mathrm{f}}, b_{\mathrm{y}}, c_{\mathrm{f}}$ and $c_{\mathrm{y}}$ have constant negative sign; hence \textbf{({A}1)} is satisfied.
\end{lem}
\begin{proof}
	This is trivial in the case of the forward curve. For the yield curve, notice that $h_1(0,\tau) = 0$ and
	\begin{align*}
		\partial_xh_1(x,\tau) = -\frac{x{\mathrm{e}}^{-\frac{x}{\tau}}}{\tau^2} < 0.
	\end{align*}
\end{proof}
\begin{lem}
	The first-order Wronskians for the forward curve are given by
	\begin{align*}
		W(b_{\mathrm{f}}, c_{\mathrm{f}})(x) &= \lambda\left(\frac{1}{\tau} - \lambda\right){\mathrm{e}}^{-\left(\frac{1}{\tau} + \lambda\right)x},\\
		W(c_{\mathrm{f}}, a_{\mathrm{f}})(x) &= {\mathrm{e}}^{-\lambda x}\left(\lambda\widetilde{\beta}_2\left(1 - x\left(\frac{1}{\tau} - \lambda\right)\right){\mathrm{e}}^{-\frac{x}{\tau}} + 2\lambda^3\widetilde{\sigma}^2{\mathrm{e}}^{-2\lambda x}\right),\\
		W(a_{\mathrm{f}}, b_{\mathrm{f}})(x) &= {\mathrm{e}}^{-\frac{x}{\tau}}\left(-\widetilde{\beta}_2{\mathrm{e}}^{-\frac{x}{\tau}} + 2\lambda\widetilde{\sigma}^2\left(\frac{1}{\tau} - 2\lambda\right){\mathrm{e}}^{-2\lambda x}\right.\\&\qquad\qquad\left. + \lambda\widetilde{\beta}_0\left(\frac{1}{\tau} - \lambda\right){\mathrm{e}}^{-\lambda x}\right).
	\end{align*}
	In particular, $W(b_{\mathrm{f}}, c_{\mathrm{f}})$ has constant sign for all $x\in (0,\infty)$, and thus \textbf{({A}4)} holds for $\mathcal{F}^{\mathrm{f}}$. The same is true for $W(b_{\mathrm{y}}, c_{\mathrm{y}})$ and $\mathcal{F}^{\mathrm{y}}$.
\end{lem}
\begin{proof}
	The Wronskians are obtained by direct calculation and the sign property of $W(b_{\mathrm{f}}, c_{\mathrm{f}})$ is obvious. Applying \cite[Lem. A.2 d.]{KRS26}, the claim on $W(b_{\mathrm{y}}, c_{\mathrm{y}})$ follows.
\end{proof}
\begin{lem}
	The Wronskian $W(a_{\mathrm{f}}, b_{\mathrm{f}}, c_{\mathrm{f}})$ is given by
	\begin{align*}
		W(a_{\mathrm{f}}, b_{\mathrm{f}}, c_{\mathrm{f}})(x) &= \lambda\left(\frac{1}{\tau} - \lambda\right){\mathrm{e}}^{-\left(\frac{1}{\tau} + \lambda\right)x}\\&\qquad\left(\widetilde{\beta}_2\tau^2\left(\frac{1}{\tau} - \lambda\right){\mathrm{e}}^{-\frac{x}{\tau}} - 2\lambda^2\widetilde{\sigma}^2\left(\frac{1}{\tau} - 2\lambda\right){\mathrm{e}}^{-2\lambda x}\right).
	\end{align*}
	It has at most one zero in $(0,\infty)$, thus, \textbf{({A}5)} holds for $\mathcal{F}^{\mathrm{f}}$.
\end{lem}
\begin{proof}
	The Wronskian is obtained by direct calculation. $({\mathrm{e}}^{-\frac{1}{\tau}}, {\mathrm{e}}^{-2\lambda x})$ or $({\mathrm{e}}^{-2\lambda x}, {\mathrm{e}}^{-\frac{1}{\tau}})$ constitutes an ECT-system; see Corollary~\ref{cor:ECT_forward}. By Theorem~\ref{thm:T-system_upper_bounds}, a polynomial of this system has at most one zero.
\end{proof}
\begin{lem}
	The Wronskian $W(a_{\mathrm{y}}, b_{\mathrm{y}}, c_{\mathrm{y}})$ is given by
	\begin{align*}
		&W(a_{\mathrm{y}}, b_{\mathrm{y}}, c_{\mathrm{y}})(x) = \\
		&\qquad \frac{\tau^5\widetilde{\beta}_2}{x^4}{\mathrm{e}}^{-\frac{x}{\tau}}\left(\frac{1}{\tau^3}{\mathrm{e}}^{-\frac{x}{\tau}} + \lambda\left(\left(2\lambda - \frac{3}{\tau}\right) + \left(\frac{1}{\tau} - \lambda\right)\frac{x}{\tau}\right){\mathrm{e}}^{-\lambda x}\right.\\
		&\qquad\left.- \frac{1}{\lambda}\left(\frac{1}{\tau} - \lambda\right)^2\left(2\lambda + \frac{1 + \lambda x}{\tau}\right){\mathrm{e}}^{-(\frac{1}{\tau} + \lambda)x}\right)\\
		&\qquad  + \frac{\widetilde{\sigma}^2}{2x^4}{\mathrm{e}}^{-\lambda x}\left(\left(\frac{1}{\tau} - \lambda\right){\mathrm{e}}^{-\frac{x}{\tau}} - \lambda^3\tau^2{\mathrm{e}}^{-2\lambda x} + \left(8\lambda - \frac{4}{\tau}\right){\mathrm{e}}^{-\left(\frac{1}{\tau} + 2\lambda\right)x}\right.\\
		&\qquad\left. + \tau^2\left(\frac{1}{\tau} - \lambda\right)\left(\frac{1}{\tau} - 2\lambda\right)\left(2\lambda\left(1 + \frac{x}{\tau}\right) + \frac{3}{\tau}\right){\mathrm{e}}^{-\left(\frac{1}{\tau} + 3\lambda\right)x}\right)
	\end{align*}
	It has at most one zero in $(0,\infty)$, thus, \textbf{({A}5)} holds for $\mathcal{F}^{\mathrm{y}}$.
\end{lem}
\begin{proof}
	The equation follows by direct calculation. By \cite[Lem. A.2 e.]{KRS26}, the Wronskian cannot have more zeros than $W(a_{\mathrm{f}}, b_{\mathrm{f}}, c_{\mathrm{f}})$.
\end{proof}
\begin{lem}
	The limiting line $\ell_0$ is the same for the forward and yield curve and given by
	\begin{align*}
		\ell_0\ :\quad \lambda\widetilde{\beta} _0 + 2\lambda\widetilde{\sigma}^2 - \widetilde{\beta}_1 - \lambda r = 0;
	\end{align*}
	the envelope contacts $\ell_0$ at the point
	\begin{align*}
		&\eta(0) = \left(\frac{\widetilde{\beta}_2\tau^2 + 2\lambda^2\widetilde{\sigma}^2}{\frac{1}{\tau} - \lambda},\ \widetilde{\beta}_0 + 2\widetilde{\sigma}^2 - \frac{\widetilde{\beta}_2\tau^2 + 2\lambda^2\widetilde{\sigma}^2}{\lambda\left(\frac{1}{\tau} - \lambda\right)}\right)
	\end{align*}
	The limiting lines $\ell_\infty$ are given by
	\begin{align*}
		\ell_\infty^{\mathrm{f}}\ &: \quad \widetilde{\beta}_0 - r = 0,\quad\text{if }\lambda < \frac{1}{\tau},\quad \text{and}\\
		\ell_\infty^{\mathrm{y}}\ &:\quad -2\tau^5\widetilde{\beta}_2 + \frac{\widetilde{\beta}_0}{\lambda} + \frac{\widetilde{\sigma}^2}{2\lambda} - \tau^2\widetilde{\beta_1} - \frac{r}{\lambda} = 0,
	\end{align*}
	for the forward curve and the yield curve respectively. If $\frac{1}{\tau} < \lambda$, the terminal sign of $\partial_xf(x)$ is the sign of $-\beta_2$. The lines $\ell_0$ and $\ell_\infty^{\mathrm{f/y}}$ are not parallel, their intersection points are given by
	\begin{align*}
		M_{\mathrm{f}} = \left(2\lambda\widetilde{\sigma}^2,\ \widetilde{\beta}_0\right)
	\end{align*}
	and
	\begin{align*}
		M_{\mathrm{y}} = \begin{pmatrix}
			\lambda\widetilde{\beta}_0 + 2\lambda\widetilde{\sigma}^2 - \frac{\lambda}{\frac{1}{\lambda} - \lambda\tau^2}\left(-2\tau^5\widetilde{\beta}_2 + \widetilde{\beta}_0\left(\frac{1}{\lambda} - \lambda\tau^2\right) + \widetilde{\sigma}^2\left(\frac{1}{2\lambda}  -2\lambda\tau^2\right)\right)\\
			\frac{1}{\frac{1}{\lambda} - \lambda\tau^2}\left(-2\tau^5\widetilde{\beta}_2 + \widetilde{\beta}_0\left(\frac{1}{\lambda} - \lambda\tau^2\right) + \widetilde{\sigma}^2\left(\frac{1}{2\lambda}  -2\lambda\tau^2\right)\right)
		\end{pmatrix}.
	\end{align*}
	If $\lambda < \frac{1}{\tau}$, the envelope $\eta_{\mathrm{y}}$ contacts $\ell_\infty^{\mathrm{y}}$ at
	\begin{align*}
		\eta_{\mathrm{y}}(\infty) = \left(-2\tau^3\widetilde{\beta}_2 + \frac{\widetilde{\sigma}^2}{2\lambda\tau^2},\ \widetilde{\beta}_0\right),
	\end{align*}
	and the envelope $\eta_{\mathrm{f}}$ contacts $\ell_\infty^{\mathrm{f}}$ asymptotically (in the sense of \cite[Rem. 3.3]{KRS23}) with
	\begin{align*}
		\eta_{{\mathrm{f}},1}(x) \overset{x\rightarrow\infty}{\longrightarrow} \sgn(-\widetilde{\beta}_2)\infty\qquad\text{and}\qquad \eta_{{\mathrm{f}},2}(x) \overset{x\rightarrow\infty}{\longrightarrow} \widetilde{\beta}_0.
	\end{align*}
	If $\frac{1}{\tau} < \lambda$, it holds
	\begin{align*}
		\eta_{\mathrm{f}/\mathrm{y}, 1}(x) \overset{x\rightarrow\infty}{\longrightarrow} \sgn(-\widetilde{\beta}_2)\infty\qquad\text{and}\qquad \eta_{\mathrm{f}/\mathrm{y}, 2}(x) \overset{x\rightarrow\infty}{\longrightarrow}\sgn(\widetilde{\beta}_2)\infty,
	\end{align*}
	and the envelope $\eta_{\mathrm{y}}$ contacts $\ell_\infty^{\mathrm{y}}$ asymptotically.
\end{lem}
\begin{proof}
	Most of the results follow by direct calculation. The limits are obtained by a comparison of exponents. To find the contact points of $\eta_{\mathrm{y}}$ with $\ell_0$ and $\ell_\infty^{\mathrm{y}}$, one can use \cite[Lem. A.2 a. and b.]{KRS26}.
\end{proof}
\textbf{({A}6)} holds, since the slope function is strictly monotone or, equivalently, $W(b,c)$ does not change sign. As a direct consequence of the lemma above, \textbf{({A}3)} holds in the yield case, and in the forward case if $\lambda < 1/\tau$. In the other case \textbf{({A}3)} holds by applying \textbf{({A}6)} and when restricting to $(0,T)$, see Remark~\ref{rem:truncated}. By scaling with a positive function \cite[Rem. 3.1]{KRS23} or restricting to $(0,T)$ if $\lambda < 1/\tau$, it can also be assured that \textbf{({A}2)} holds.

\bibliographystyle{alpha}
\bibliography{references}

@book{Ait56,
	author    = {Aitken, Alexander C},
	title     = {Determinants and matrices},
	publisher = {Oliver and Boyd},
	year      = {1956},
	series    = {University Mathematical Texts},
	edition   = {9th},
	address   = {Edinburgh}
	}

@book{bruce1992curves,
	author = {Bruce, James W and Giblin, Peter J},
	publisher = {Cambridge University Press},
	title = {Curves and Singularities: a geometrical introduction to singularity theory},
	year = {1992}
	}

@techreport{Bli96,
	address = {Atlanta, GA},
	author = {Robert R. Bliss},
	number = {96-12a},
	publisher = {Federal Reserve Bank of Atlanta},
	title = {Testing term structure estimation methods},
	type = {Working Paper},
	year = {1996}
	}

@article{bauer2018economic,
	author = {Bauer, Michael D and Mertens, Thomas M},
	journal = {FRBSF Economic Letter},
	pages = {8--07},
	publisher = {Federal Reserve Bank of San Francisco},
	title = {Economic forecasts with the yield curve},
	volume = {7},
	year = {2018}
	}

@article{CG96,
	author = {Corless, Robert M and Gonnet, Gaston H and Hare, David EG and Jeffrey, David J and Knuth, Donald E},
	journal = {Advances in Computational mathematics},
	pages = {329--359},
	publisher = {Springer},
	title = {On the {Lambert} {W} function},
	volume = {5},
	year = {1996}
	}

@article{diez2020yield,
	author = {Diez, Franziska and Korn, Ralf},
	journal = {European Actuarial Journal},
	pages = {91--120},
	publisher = {Springer},
	title = {Yield curve shapes of Vasicek interest rate models, measure transformations and an application for the simulation of pension products},
	volume = {10},
	year = {2020}
	}

@article{DPR05,
	Author = {Diebold, Francis X. and Piazzesi, Monika and Rudebusch, Glenn D.},
	Title = {Modeling Bond Yields in Finance and Macroeconomics},
	Journal = {American Economic Review},
	Volume = {95},
	Number = {2},
	Year = {2005},
	Pages = {415–420}
	}

@article{diebold2006macroeconomy,
	title = {The macroeconomy and the yield curve: a dynamic latent factor approach},
	journal = {Journal of Econometrics},
	volume = {131},
	number = {1},
	pages = {309-338},
	year = {2006},
	issn = {0304-4076},
	author = {Francis X. Diebold and Glenn D. Rudebusch and S. {Bora\u{g}an Aruoba}}
	}

@book{Fil09,
	author = {Filipovic, Damir},
	publisher = {Springer Science \& Business Media},
	title = {Term-structure models: A graduate course},
	year = {2009}
	}

@article{gurkaynak2007us,
	author = {G{\"u}rkaynak, Refet S and Sack, Brian and Wright, Jonathan H},
	journal = {Journal of Monetary Economics},
	number = {8},
	pages = {2291--2304},
	publisher = {Elsevier},
	title = {The {US} Treasury yield curve: 1961 to the present},
	volume = {54},
	year = {2007}
	}

@article{HA198069,
	title = {A note on the determinant of a functional confluent vandermonde matrix and controllability},
	journal = {Linear Algebra and its Applications},
	volume = {30},
	pages = {69-75},
	year = {1980},
	author = {T.T. Ha and J.A. Gibson},
	}

@article{HW90,
	author = {Hull, John and White, Alan},
	title = {Pricing Interest-Rate-Derivative Securities},
	journal = {The Review of Financial Studies},
	volume = {3},
	number = {4},
	pages = {573-592},
	year = {1990},
	month = {10},
	}

@article{keller2018correction,
	author = {Keller-Ressel, Martin},
	journal = {Finance and Stochastics},
	pages = {503--510},
	publisher = {Springer},
	title = {Correction to: Yield curve shapes and the asymptotic short rate distribution in affine one-factor models},
	volume = {22},
	year = {2018}
	}

@article{keller2021classification,
	author = {Keller-Ressel, Martin},
	journal = {International Journal of Theoretical and Applied Finance},
	number = {05},
	pages = {2150027},
	publisher = {World Scientific},
	title = {The Classification Of Term Structure Shapes In The Two-Factor {V}asicek Model---A Total Positivity Approach},
	volume = {24},
	year = {2021}
	}

@article{keller2008yield,
	author = {Keller-Ressel, Martin and Steiner, Thomas},
	journal = {Finance and Stochastics},
	number = {2},
	pages = {149--172},
	publisher = {Springer},
	title = {Yield curve shapes and the asymptotic short rate distribution in affine one-factor models},
	volume = {12},
	year = {2008}
	}

@article{KRS23,
	author = {Keller-Ressel, Martin and Sachse, Felix},
	journal = {International Journal of Theoretical and Applied Finance},
	number = {04n05},
	pages = {2350013},
	publisher = {World Scientific Publishing Co. Pte. Ltd.},
	title = {State Space Decomposition And Classification Of Term Structure Shapes In The Two-Factor Vasicek Model},
	volume = {26},
	year = {2023}
	}

@article{KRS26,
	author = {Keller-Ressel, Martin and Sachse, Felix},
	title = {Term Structure Shapes and Their Consistent Dynamics in the Svensson Family},
	journal = {Mathematical Finance},
	volume = {36},
	number = {1},
	pages = {140-155},
	year = {2026}
	}

@book{KS66,
	author = {Karlin, Samuel and Studden, William J},
	publisher = {Interscience Publishers},
	title = {Tchebycheff systems: With applications in analysis and statistics},
	year = {1966}
	}

@article{litterman1991common,
	title={Common factors affecting bond returns},
	author={Litterman, Robert B and Scheinkman, Jose},
	journal={The journal of fixed income},
	volume={1},
	number={1},
	pages={54--61},
	year={1991},
	publisher={Portfolio Management Research}
	}

@article{NS87,
	author = {Nelson, Charles R and Siegel, Andrew F},
	journal = {Journal of business},
	pages = {473--489},
	publisher = {JSTOR},
	title = {Parsimonious modeling of yield curves},
	year = {1987}
	}

@phdthesis{sachse2026thesis,
	author  = {Sachse, Felix},
	title   = {Term Structure Shapes},
	school  = {TU Dresden},
	year    = {2026},
	address = {Dresden, Germany}
	}

@article{schlogl1997factor,
	title={Factor Models and the Shape of the Term Structure},
	author={Schloegl, Erik and Sommer, Daniel},
	journal={Available at SSRN: https://ssrn.com/abstract=1145},
	year={1997}
	}

@misc{Sve94,
	author = {Svensson, Lars EO},
	publisher = {National bureau of economic research Cambridge, Mass., USA},
	title = {Estimating and interpreting forward interest rates: Sweden 1992-1994},
	year = {1994}
	}

@article{vasicek1977equilibrium,
	author = {Vasicek, Oldrich},
	journal = {Journal of Financial Economics},
	number = {2},
	pages = {177--188},
	publisher = {Elsevier},
	title = {An equilibrium characterization of the term structure},
	volume = {5},
	year = {1977}
	}

@book{Wal98,
	title = {Ordinary Differential Equations},
	journal = {Graduate Texts in Mathematics},
	publisher = {Springer New York},
	series    = {Graduate Texts in Mathematics},
	author = {Walter,  Wolfgang},
	year = {1998}
	}
  
\end{document}